\documentclass[tbtags,reqno]{amsart}

\usepackage{mathrsfs}
\usepackage{latexsym}
\usepackage{amsthm,amsopn,amsfonts,amssymb,epsfig,tabmac,color,pstricks,pb-diagram,mathdots,stmaryrd}
\usepackage[active]{srcltx}
\numberwithin{equation}{section}

\makeatletter
\renewcommand{\subsubsection}{\@startsection
{subsubsection}
{3}
{0mm}
{\baselineskip}
{-0.5\baselineskip}
{\normalfont\normalsize\bfseries}}
\makeatother

\newtheorem{theorem}{Theorem}
\newtheorem{lemma}[theorem]{Lemma}
\newtheorem{proposition}[theorem]{Proposition}

\newtheorem{conjecture}[theorem]{Conjecture}

\newtheorem{corollary}[theorem]{Corollary}

\theoremstyle{remark}
\newtheorem{remark}[theorem]{Remark}
\newtheorem*{acknow}{Acknowledgments}

\def\sas{\smallskip}

\def\la{{\lambda}}

\def\cal L{{\mathcal L}}

\def\sas{\smallskip}

\newcommand{\tcercle}[1]{\ensuremath{\setlength{\unitlength}{1ex}\begin{picture}(2.8,2.8)\put(1.4,1.4){\circle{2.8}\makebox(-5.6,0){#1}}\end{picture}}}

\let\d\partial
\let\n\noindent

\def\B{{\mathcal B}}

\let\la\lambda
\let\La\Lambda
\let\a\alpha

\let\Om\Omega
\let\om\omega
\let\ta\theta

\let\rw\rightarrow
\let\Rw\Rightarrow

\def\em{{\emptyset}}

\newcommand{\LL}{\ensuremath{\langle\!\langle}}
\newcommand{\RR}{\ensuremath{\rangle\!\rangle}}

\def\cd{\circledast}
\def\S{\mathcal{S}}
\def\B{\mathcal{B}}
\def\F{\mathcal{F}}

\def\lrw{\leftrightarrow}

\def\co#1{{\color{red}#1}}

\def\cg#1{{\color{gray}#1}}

\let\L\langle \let\R\rangle

\def\beq{\begin{equation}}
\def\eeq{\end{equation}}
\def\bea{\begin{align}}
\def\eea{\end{align}}

\title{ Double Macdonald polynomials as the stable limit of Macdonald superpolynomials}
\begin{document}

\author{O. Blondeau-Fournier}
\address{D\'epartement de physique, de g\'enie physique et
d'optique, Universit\'e Laval,  Qu\'ebec, Canada,  G1V 0A6.}
\email{olivier.b-fournier.1@ulaval.ca}
\author{L. Lapointe}
\address{Instituto de Matem\'atica y F\'{\i}sica, Universidad de
Talca, 2 norte 685, Talca, Chile.}
\email{lapointe@inst-mat.utalca.cl }
\author{P. Mathieu}
\address{D\'epartement de physique, de g\'enie physique et
d'optique, Universit\'e Laval,  Qu\'ebec, Canada,  G1V 0A6.}
\email{pmathieu@phy.ulaval.ca}

\begin{abstract}
Macdonald superpolynomials provide a remarkably rich generalization of the usual Macdonald polynomials.  
The starting point of this work is the observation of a previously unnoticed stability property of the Macdonald superpolynomials when the fermionic sector $m$ is sufficiently large: their decomposition in the monomial basis is
then independent of $m$. These stable superpolynomials are readily mapped into bisymmetric polynomials, an operation that spoils the ring structure but drastically simplifies the associated vector space.
Our main result is a factorization of the  (stable) bisymmetric Macdonald polynomials,  called double Macdonald polynomials and indexed by pairs of partitions, into a product of Macdonald polynomials (albeit subject to non-trivial plethystic transformations). As an off-shoot, we note that, after multiplication by a $t$-Vandermonde determinant, this provides explicit formulas for a 
 large class of 
Macdonald polynomials with prescribed symmetry.  The factorization of the
double Macdonald polynomials
leads immediately to the generalization of basically every elementary properties of the Macdonald polynomials to the double case
(norm, kernel, 
duality, evaluation, positivity, etc).  When lifted back to superspace, this validates various previously formulated conjectures in the stable regime.

The $q,t$-Kostka coefficients associated to the double Macdonald polynomials are shown to be $q,t$-analogs of the dimensions of the 
irreducible representations of the hyperoctahedral group $B_n$.  Moreover, a Nabla operator on the double Macdonald polynomials is defined and its
action on a certain bisymmetric Schur function  can be interpreted as
the Frobenius  series of a bigraded module of dimension $(2n+1)^n$, a formula again characteristic of the Coxeter group of type $B_n$.

Finally, as a side result, we obtain  a simple identity involving products of four Littlewood-Richardson coefficients. 
\end{abstract}

\maketitle

\section{Introduction}

\subsection{From superpolynomials to bisymmetric polynomials}
The Macdonald polynomials in superspace, which have been recently introduced in \cite{BDLM1,BDLM2}, 
provide a combinatorially-rich generalization of the Macdonald polynomials. We show in this article that when the fermionic sector is large enough, 
the Macdonald polynomials in superspace embody a natural form of a bisymmetric extension of the Macdonald 
polynomials, whose corresponding Kostka and Nabla combinatorics is that of the hyperoctahedral group $B_n$ (whereas it is that of the symmetric group $S_n$ 
in the usual case).  
This bisymmmetric version of Macdonald polynomials, which now depend on two alphabets, will be referred to as {\it double Macdonald 
polynomials}.\footnote{Even though the double Macdonald polynomials are connected to the hyperoctahedral group $B_n$, it should however
 be kept in mind that they are not the Macdonald polynomials  defined on the root lattice of type $B$ \cite{Mac1,Mac2}. They are also different from the Macdonald polynomials with 
hyperoctahedral symmetries constructed in \cite{diejen}}.

Let us put these statements in context by first recalling the definition of a superpolynomial or equivalently said,  a polynomial in superspace. In addition to be a polynomial in the usual indeterminates $x_1,\cdots,x_N$ over $\mathbb Q(q,t)$, it is a function of $N$ Grassmannian (also called anticommuting or fermionic) variables $\ta_1,\cdots,\ta_N$. Such polynomials are said to be symmetric if they are invariant with respect to the simultaneous interchange of $(x_i,\ta_i)\lrw (x_j,\ta_j)$. Due to the latter property, if a superpolynomial is symmetric,
it is always a sum of monomials $\ta_{i_1}\cdots \ta_{i_m}$ multiplied  by a polynomial antisymmetric in  
the variables $x_{i_1},\cdots, x_{i_m}$ and symmetric in the remaining ones. 
The superpolynomials we consider are always bi-homogeneous in the variables $x$ and $\theta$, and their fermionic degree (the degree in $\theta$) 
is
denoted $m$.
Finally, a symmetric superpolynomial is indexed by a superpartition $\La$:
a pair of partitions $(\La^a;\La^s)$ such that $\La^a$ has $m$ distinct parts (the $m$-th one being allowed to be 0). The partitions $\La^a$ (resp. $\La^s$) 
captures the degree of the antisymmetric (resp. symmetric) polynomial 
associated to each term $\ta_{i_1}\cdots \ta_{i_m}$ of the superpolynomial.

To make these comments more concrete, consider the generalization of the monomial and power-sum basis to superspace \cite{BDLM1}:
\beq m_\La={\sum_{\sigma\in S_N}}' \ta_{\sigma(1)}\cdots \ta_{\sigma(m)}\,x_{\sigma(1)}^{\La_1}\cdots x_{\sigma(N)}^{\La_N}
\eeq 
(where the prime indicates a sum over distinct terms) and    
\begin{equation}\label{spower}
p_\La=\tilde{p}_{\La_1}\cdots\tilde{p}_{\La_m}p_{\La_{m+1}}\cdots p_{\La_{N}},\qquad\text{where}\quad \tilde{p}_k=\sum_{i=1}^N\theta_ix_i^k\qquad\text{and}\qquad p_r=
\sum_{i=1}^Nx_i^r \, 
\end{equation}  
(with $k\geq 0$ and $r \geq 1$).  If $m=2$ and $N=3$ we have for instance: 
\beq
m_{(3,1;1)}=\ta_1\ta_2(x_{1}^3x_2-x_2^3x_1)x_3+\ta_1\ta_3(x_{1}^3x_3-x_3^3x_1)x_2+\ta_2\ta_3(x^2_{3}x_3-x_3^3x_2)x_1 
\eeq
and
\begin{align}
p_{(3,1;1)}&= (\ta_1x_1^3+\ta_2x_2^3+\ta_3x_3^3)(\ta_1x_1+\ta_2x_2+\ta_3x_3)(x_1+x_2+x_3)
\nonumber\\  &=\bigl(\ta_1\ta_2(x_{1}^3x_2-x_2^3x_1)+
\ta_1\ta_3(x_{1}^3x_3-x_3^3x_1)+\ta_2\ta_3(x^2_{3}x_3-x_3^3x_2) \bigr)(x_1 +x_2+x_3).
\end{align}

To any superpolynomial, we can associate a bisymmetric 
polynomial.
For a superpolynomial of fermionic degree $m$, this is done by (1)- extracting the coefficient of $\ta_1\cdots \ta_m$ of the superpolynomial,  and (2)-dividing the result by the Vandermonde determinant in the variables $x_1,\cdots,x_m$. By construction, the resulting polynomial is symmetric in 
both the variables $x=(x_1,\cdots, x_m)$ and $y=(x_{m+1},\cdots, x_N)$.  For the examples presented above, this procedure yields $(x_1^2x_2+x_1x_2^2)x_3$ and $(x_1^2x_2+x_1x_2^2)(x_1+x_2+x_3)$ respectively.

When, as was just described, we pass to bisymmetric polynomials,
the natural ring structure of the space of  
superpolynomials is lost (the ring structure of the space of
bisymmetric polynomials 
is not the right one, especially when considering connections 
with supersymmetry \cite{DLM0,DLMsvir}). 
However, in this article, we are only interested in the underlying vector space.
As such, it will prove more convenient to work with bisymmetric polynomials. We will then have at our disposal the very powerful (and much better known) 
language of symmetric function theory.

A bisymmetric polynomial is naturally indexed by a pair of partitions extracted from the superpartition.  To be more precise,
the correspondence  between a superpartition of fermionic degree $m$ and the pair of partitions $\la, \mu$ is
 \begin{equation}\label{cora}
\La =(\La^a ; \La^s) \lrw (\la, \mu) = (\La^a - \delta^{m}, \La^s),
\end{equation}
where $\delta^{m}:=(m-1,\ldots, 0)$ stands for the staircase partition. 
Observe that from \eqref{cora}, it is immediate that
$\ell(\lambda) \leq m$.
With this correspondence, the monomial and power-sums symmetric 
superpolynomials are
then associated to the following bisymmetric polynomials
respectively (see Appendix \ref{mpbi} for a detailed derivation of this correspondence)
\begin{equation}\label{modea}
m_{\La} \quad \longleftrightarrow \quad m_{\la,\mu}(x,y):=s_\la(x) \,m_\mu(y),
\end{equation}
and
\begin{equation}\label{psdea}
p_{\La} \quad \longleftrightarrow \quad 
p_{\la,\mu}(x,y) := s_\la(x) \,  p_\mu (x,y) \, ,
\end{equation}
where $s_\la$, $m_\la$, and $p_\la$ are respectively the Schur, monomial and power-sum symmetric functions.  In the case of $m_{\lambda,\mu}$, the functions
in the product depend upon
complementary sets of variables, while $p_{\lambda,\mu}$ does not ($p_{\mu}(x,y)$
is the usual power-sum symmetric functions in the union of the variables $x$ and $y$).
In the two cases, it is still noteworthy that the corresponding bisymmetric polynomials have such a simple factorization (we stress that $m_\La$ and $p_\La$ are both functions
of the variables $x_1,\cdots,x_N,\ta_1,\cdots,\ta_N$).

The bisymmetric polynomials associated to the Macdonald superpolynomials 
$P_\La(x_1,\cdots,x_N,\ta_1,\cdots,\ta_N;q,t)$ \cite{BDLM1,BDLM2} will simply be denoted $P_\La(x,y;q,t)$ \footnote{The meaning of  $P_\La$ will be clear from the context and actually, in the body of the text, it will always refer to the bisymmetric version.
The two forms are distinguished by their explicit variable-dependence: either $(x,\ta)$ or $(x,y)$.}
For reasons to become clear shortly,  it is more convenient at this point to 
keep using superpartitions as indices even for the
 bisymmetric version  of the Macdonald superpolynomials ($\La$ captures for instance the information on the fermionic degree $m$ which corresponds to the maximal length $\la$ can have).  We now translate in the language of bisymmetric functions Theorem 1 of \cite{BDLM2}
which establishes the existence of the Macdonald superpolynomials.  It relies on the dominance ordering on superpartitions, which is defined
as 
\beq \label{domidebut}
\Lambda \geq \Omega  \qquad {\rm iff} \qquad \Lambda^* \geq \Omega^* \qquad {\rm and} \qquad 
\Lambda^\circledast \geq \Omega^\circledast
\eeq
where the order on partitions is the dominance ordering, and where
\beq
 \Lambda^*=(\lambda+\delta^{m}) \cup \mu \quad  {\rm and} \quad
\Lambda^\circledast=
 (\lambda+\delta^{m+1}) \cup \mu \, .
\eeq
\begin{theorem}\label{theo1} Given a superpartition $\Lambda \leftrightarrow (\lambda,\mu)$ and $N -m\geq |\lambda|+|\mu|$, there exists
a unique bisymmetric polynomial  $P_{\Lambda}=P_{\Lambda}(x,y;q,t)$,
with $x=(x_1,\dots,x_m)$ and $y=(x_{m+1},\dots,x_N)$, such that: 
\begin{equation}\label{mac1}
\begin{array}{lll} 1)& P_{\Lambda} =
m_{\lambda,\mu} + \text{lower terms},\\
&\\
2)&\LL  P_{\Lambda},  P_{\Omega} \RR_{q,t} =0\quad\text{if}\quad 
\Lambda \ne \Omega 
\end{array}\end{equation}
The dominance  ordering on pairs of partitions is such
that $(\lambda,\mu) \geq (\omega,\eta)$ iff the corresponding superpartitions
$\Lambda$ and $\Omega \leftrightarrow (\omega,\eta)$ are such that
$\Lambda \geq \Omega$.  The scalar product is defined on the power-sums \eqref{psdea} as\footnote{In the expression of the scalar product, we have dropped a 
factor $(-1/q)^{m(m-1)/2}$ which does not affect the orthogonality.}
\begin{equation}\label{newsp}
\LL {p_{\lambda,\mu}}, \, {p_{\omega,\eta}}\RR_{q,t}= \delta_{\lambda \omega} 
\delta_{\mu \eta} q^{|\lambda|}
\,z_\mu(q,t) \, ,
 \end{equation}  
where $z_{\mu}(q,t)$ is given in \eqref{ortoqt}.
\end{theorem}
We stress that the monomial expansion of $P_{\Lambda}$ is independent of $N-m$ 
(granted that $N-m$ is large enough) and thus $N-m$ can be considered infinite.
Unexpectedly, a similar independence upon $m$ holds.
Let us denote by $n$ the total degree of the pair $(\la,\,\mu)$, that is,
\beq n=|\la|+|\mu| \, .
\eeq
In the  previous theorem, the restriction $\ell(\lambda) \leq m$
in the correspondence \eqref{cora} is lifted when $m \geq n$.  
 What is 
remarkable, and totally unexpected from the supersymmetric point of view, 
is that even though the ordering 
\eqref{domidebut} appears to be highly dependant of $m$, 
the
monomial expansion \eqref{mac1} 
of the bisymmetric polynomial $P_\La$ 
 does not depend on $m$  
whenever $m \geq n$.  This will be referred to as the  {\it stable sector}
of
the  bisymmetric Macdonald polynomials. This phenomenon is described schematically in Figure \ref{fig1}.
Let us illustrate the stability property of
$P_\La$, where $\Lambda \leftrightarrow (\emptyset,(2))$, 
by displaying the monomial
 decompositions when $n=2$ for four different values of $m$: 
\begin{align}P_{(0;2)}&=m_{\emptyset,(2)}
+{\frac { \left( 1-t \right)  \left(1+ q\right) }{1-qt}}\,m_{\emptyset,(1,1)}+{\frac { \left( 1-t \right) }{1-qt}}\,m_{{(1),(1)}}
\nonumber\\
P_{(1,0;2)}&=m_{\emptyset,(2)}+{\frac { \left( 1-t \right)  \left(1+ qt \right) }{1-q{t}^{2}}}
\,m_{\emptyset,(1,1)}
\nonumber\\
P_{(2,1,0;2)}&=m_{\emptyset,(2)}+{\frac { \left( 1-t \right)  \left(1+ qt \right) }{1-q{t}^{2}}}
\,m_{\emptyset,(1,1)}
\nonumber\\
P_{(3,2,1,0;2)}&=m_{\emptyset,(2)}+{\frac { \left( 1-t \right)  \left(1+ qt \right) }{1-q{t}^{2}}}
\,m_{\emptyset,(1,1)}
\end{align}
The first case, for which $m=1<2=n$, does not belong to the stable sector. In the other three cases, 
one recovers three identical expressions (even though the corresponding
monomials depend on different sets of variables).

\begin{figure}[ht]\label{fig1}
\caption{{To each dot corresponds a family of superpolynomials indexed by superpartitions of total degree $|\Lambda^a|+|\Lambda^s|=n+m(m-1)/2$ and whose component $\La^a$ has exactly $m$ parts.  The sector $m\geq n$ (indicated in gray) is the domain in which the bisymmetric polynomials are stable. In other words, for a fixed value of $m$, all the points below and on the diagonal represent equivalent families, for which  convenient representatives are the points on the diagonal (larger dots). In the complementary region, the superspace formalism (via the dominance ordering for superpartitions) is mandatory.}}
\label{mplusgrandn}
\begin{center}
\begin{pspicture}(0,0)(5,5)

\psset{linestyle=none, fillstyle=solid,fillcolor=lightgray}
\pscustom{
  \psline(0.5,0.5)(4.2,0.5)(4.2,4.2)
}

\psset{linestyle=solid}
\psline{->}(0.5,0.5)(4.5,0.5)

\psline{->}(0.5,0.5)(0.5,4.5)

{
\rput(0.25,0.5){{\scriptsize $0$}}
\rput(0.25,1.0){{\scriptsize $1$}} \rput(0.25,1.5){{\scriptsize $2$}}
\rput(0.25,2.0){{\scriptsize $3$}} \rput(0.25,2.5){{\scriptsize $4$}}
\rput(0.25,3.0){{\scriptsize $5$}} \rput(0.25,3.5){{\scriptsize $6$}}
\rput(0.25,4.0){{\scriptsize $7$}} \rput(0.25,4.7){{\scriptsize $n$}}

\rput(0.5,0.25){{\scriptsize $0$}}
\rput(1.0,0.25){{\scriptsize $1$}} \rput(1.5,0.25){{\scriptsize $2$}} \rput(2.0,0.25){{\scriptsize $3$}}
\rput(2.5,0.25){{\scriptsize $4$}} \rput(3.0,0.25){{\scriptsize $5$}} \rput(3.5,0.25){{\scriptsize $6$}}
\rput(4.0,0.25){{\scriptsize $7$}}
\rput(4.7,0.25){{\scriptsize $m$}}
}

{
\psset{linestyle=solid}

\psline{-}(.5,.5)(4.2,4.2) 
}
{
\psset{dotsize=2pt}
\psdots(0.5,0.5)(0.5,1)(0.5,1.5)(0.5,2)(0.5,2.5)(0.5,3)(0.5,3.5)(0.5,4)
(1,0.5)(1,1)(1,1.5)(1,2)(1,2.5)(1,3)(1,3.5)(1,4)
(1.5,0.5)(1.5,1)(1.5,1.5)(1.5,2)(1.5,2.5)(1.5,3)(1.5,3.5)(1.5,4)
(2,0.5)(2,1)(2,1.5)(2,2)(2,2.5)(2,3)(2,3.5)(2,4)
(2.5,0.5)(2.5,1)(2.5,1.5)(2.5,2)(2.5,2.5)(2.5,3)(2.5,3.5)(2.5,4)
(3,0.5)(3,1)(3,1.5)(3,2)(3,2.5)(3,3)(3,3.5)(3,4)
(3.5,0.5)(3.5,1)(3.5,1.5)(3.5,2)(3.5,2.5)(3.5,3)(3.5,3.5)(3.5,4)
(4,0.5)(4,1)(4,1.5)(4,2)(4,2.5)(4,3)(4,3.5)(4,4)
}

{\psset{dotsize=3pt}
\psdots(.5,.5)(1,1)(1.5,1.5)(2,2)(2.5,2.5)(3,3)(3.5,3.5)(4,4)}
\end{pspicture}
\end{center}
\end{figure}

In the stable sector, it will thus be more natural to index 
the bisymmetric polynomial $P_\La:= P_{\lambda,\mu}$  
by the pair of partitions $(\lambda,\mu)$,
even more so that in this sector
the ordering \eqref{domidebut} can be replaced by the
following dominance ordering on pairs of partitions (cf. Proposition \ref{lemdo} in Appendix \ref{proof_order}):
for $(\la,\mu)$ and $(\omega, \eta)$ both of total degree $n$, 
\begin{equation}\label{domibi}
(\la,\mu) \geq (\omega,\eta)  \quad {\rm iff}  \quad
 \la_1+ \cdots + \la_i \geq \omega_1 +\cdots +\omega_i \quad \text{and} \quad  |\la|+ \mu_1 + \cdots + \mu_j \geq    |\omega |+ \eta_1 + \cdots + \eta_j \,  \, \, \forall i,j  ,
\end{equation}
where it is understood that $\lambda_k=0$ if $k>\ell(\lambda)$ (and similarly
for $\mu, \omega$ and $\eta$).

Because they are labelled by two partitions and, as we will see shortly, they are naturally viewed as a function of two sets of (commuting) variables, the bisymmetric
Macdonald polynomials in the stable sector
will be called double Macdonald polynomials.
Therefore, in the stable sector
Theorem~\ref{theo1} becomes:
\begin{theorem} \label{Pdef1} Let $(\lambda,\mu)$ be of total degree $n$.
Then the double Macdonald polynomials $P_{\lambda,\mu}(x,y;q,t)$,
where  $x=(x_1,\dots,x_m)$  and $y=(x_{m+1},\dots,x_N)$ (with  $m,N-m \geq n$)
are the unique bisymmetric polynomials such that
\begin{eqnarray}\label{Macsta}
& & P_{\la,\mu}(x,y;q,t) = m_{\la,\mu}(x,y) + \sum_{\omega,\eta < \la,\mu} c_{\la,\mu;\omega,\eta}(q,t) \, m_{\omega,\eta}(x,y)\nonumber \\
&  &\LL P_{\la,\mu} , P_{\omega,\eta} \RR_{q,t} = 0 \quad\text{if}\quad
(\la,\mu) \neq (\omega,\eta)
\end{eqnarray}
where the ordering on pairs of partitions 
and the scalar product are respectively defined in
\eqref{domibi} and \eqref{newsp}.
\end{theorem}
If we let $m$ and $N-m$ go to infinity, we obtain double Macdonald
{\it functions}.  We will nevertheless restrict ourselves to the finite case in this article. It should also be stressed that there is no solution to the two conditions \eqref{Macsta} in the non-stable sector if the ordering \eqref{domibi} is used. Therefore, in order to interpolate between the usual and double
Macdonald polynomials (corresponding respectively to the cases 
$m=0$ and $m\geq n$), the construction relying on the super-dominance ordering \eqref{domidebut} is necessary.


\subsection{Statement of the main results}
As already pointed out, the bisymmetric version of the monomials and the power-sums display a very simple factorization pattern, irrespectively of the value of $m$. Such a generic factorization is not expected to be observed for the bisymmetric Macdonald polynomials. However, it turns out that in the stable sector, such a factorization occurs, albeit in a non-obvious way.  Using the 
plethystic notation (reviewed in Section~\ref{secdef}) with $X=x_1+\cdots+x_m$
and $Y=x_{m+1}+\cdots+x_N$, the factorization of the double Macdonald polynomials reads (cf. Theorem \ref{teofac})
\beq \label{factor}
P_{\lambda,\mu}(x,y;q,t) = P_{\lambda}^{(q,qt)} \left[ X + \frac{q(1-t)}{1-qt}Y\right] P_{\mu}^{(qt,t)} \left[ Y\right] \, ,\eeq
where $P_{\lambda}^{(q,t)}(x)$ denotes the usual Macdonald polynomial 
$P_{\lambda}(x;q,t)$.

We briefly digress in order to comment on the consequences of this result for
the Macdonald polynomials with prescribed symmetry \cite{BDF}.  Let
$\mathcal A_t^{(x)}$ be the $t$-antisymmetrization (or Hecke antisymmetrization)
operator acting on the variables $x$, $\mathcal S_t^{(y)}$ be the $t$-symmetrization (or Hecke symmetrization)
operator acting on the variables $y$, $\Delta_t(x)$ be the $t$-Vandermonde
determinant in the variables $x$, and $ E_\gamma(x,y;q,t)$ be the non-symmetric Macdonald polynomials in the variables $x_1,\dots,x_m,y_1,\dots,y_{N-m}$ (the reader is refered to \cite{BDF,BDLM2} for the relevant definitions).
It is known \cite{BDLM2} that the double Macdonald polynomial (via the superpolynomial construction) $P_{\lambda,\mu}(x,y;q,t)$ is related to
the Macdonald polynomial with prescribed symmetry $ {\mathcal A_t^{(x)}} {\mathcal S^{(y)}_t}  E_\gamma(x,y;q,t)$ through
\beq
P_{\lambda,\mu}(x,y;q,t)  \propto \frac{1}{\Delta_t(x)} {\mathcal A_t^{(x)}} {\mathcal S^{(y)}_t}  E_\gamma(x,y;q,t)
\eeq
where $\gamma=(\gamma_1,\dots,\gamma_N)$ is any composition such that $(\gamma_1,\dots,\gamma_m)$ and $(\gamma_{m+1},\dots,\gamma_N)$ rearrange to the partitions $\lambda+\delta^m$ and 
$\mu$ respectively, and where $\propto$ means that the result holds up to a constant.  The factorization \eqref{factor} translates into a factorization
of Macdonald polynomials with prescribed symmetry. 
\begin{theorem}  Let $\gamma=(\gamma_1,\dots,\gamma_N)$ 
be a composition such that $\gamma_1,\dots,\gamma_m$ are all distinct
and such that
$m, N-m \geq |\gamma|- m(m-1)/2$.  Then, the Macdonald polynomials with prescribed symmetry  ${\mathcal A_t^{(x)}} {\mathcal S^{(y)}_t}  E_\gamma(x,y;q,t)$ is such that
\begin{equation} \label{eqprescribedfact}
{\mathcal A_t^{(x)}} {\mathcal S^{(y)}_t}  E_\gamma(x,y;q,t) \propto
\Delta_t(x) \, P_{\lambda}^{(q,qt)} \left[ X + \frac{q(1-t)}{1-qt}Y\right] P_{\mu}^{(qt,t)} \left[ Y\right] \, ,
\end{equation}
where $\lambda+\delta^m$ and $\mu$ are the partitions
corresponding respectively to the rearrangements of  $(\gamma_1,\dots,\gamma_m)$ and $(\gamma_{m+1},\dots,\gamma_N)$. 
\end{theorem}
The theorem means that if we take any
non-symmetric Macdonald polynomials indexed by a composition 
of sufficiently 
low degree, $t$-antisymmetrize with respect to the first $m$ variables, and $t$-symmetrize with respect to 
the remaining ones, then the
result is either zero (if there are repeated entries in the first $m$ 
entries of the composition) or, quite amazingly, a $t$-Vandermonde determinant times a
product of two Macdonald polynomials!  We should add that the  degenerate case 
$\la=\em $ of \eqref{eqprescribedfact} was already known  \cite[Proposition 2]{BDF}.

Returning to our main line, we stress that the factorization  \eqref{factor} offers 
the royal road to the study the properties of the double Macdonald polynomials. To a large extent, this study amounts to lift to the factorized form known properties of the usual Macdonald polynomials (cf. \cite[Chapter VI]{Mac}). In this way, we readily obtain the norm, the duality and the evaluation
 (respectively given by Corollary \ref{cor_norm}, 
Proposition \ref{propdu} and Corollary \ref{coroeval}). 

As a side result, we point out an interesting consequence of the duality:  a simple identity involving products of four Littlewood-Richardson coefficients (see Proposition \ref{les4c}). For any partitions 
$\lambda,\mu,\nu$ and $\omega$, we have
\beq
\sum_{\gamma,\eta,\sigma,\tau} (-1)^{|\tau|} 
c^{\gamma}_{\tau' \nu}
c^{\lambda}_{\gamma \eta} c^{\sigma}_{\eta \mu}  c^\omega_{\sigma \tau}
 =  \delta_{\lambda \nu} \delta_{\mu \omega}
\eeq
where $\gamma,\eta,\sigma,\tau$ run over 
all partitions and where $c^{\lambda}_{\mu \nu}$ is the corresponding
Littlewood-Richardson coefficient.

Implementing
the plethystic substitutions $X \mapsto X$ and
$X+Y \mapsto (X+Y)/(1-t)$ on the r.h.s. of \eqref{factor} defines the modified 
double Macdonald polynomials:
\beq \label{HHa}
H_{\lambda,\mu}(x,y;q,t) =  H_{\lambda}^{(q,qt)} \left[ X+ qY \right] H_{\mu}^{(qt,t)} \left[ tX + Y\right]\,,
\eeq
where $H_{\lambda}^{(q,t)}[X]= J_{\lambda}^{(q,t)}\left[ X/(1-t)\right]$ is the modified Macdonald polynomial ($J_{\lambda}^{(q,t)}(x)$ is the  integral form of the Macdonald polynomial $P_{\lambda}^{(q,t)}(x)$).
The expansion of $H_{\lambda,\mu}(x,y;q,t) $ in terms of the Schur functions associated to the irreducible characters of $B_n$, namely $ s_{\la,\mu}(x,y)=s_\la(x)\, s_\mu(y)$, define the double Kostka coefficients $ K_{\kappa, \gamma \; \la,\mu}(q,t)$:
\begin{equation} \label{Kosbisym}
H_{\la,\mu}(x,y;q,t)= \sum_{\kappa,\gamma} K_{\kappa, \gamma \; \la,\mu}(q,t) s_{\kappa, \gamma}(x,y).
\end{equation}
We show that  $K_{\kappa, \gamma \; \la,\mu}(1,1)$ is equal to 
the dimension of the irreducible
representation of $B_n$ indexed by the pairs of partitions $(\kappa,\gamma)$ (see Proposition \ref{Prop1}). 
This  is the first genuine contact with the hyperoctahedral group.
We then show that the basic properties of the double Kostka coefficients, namely, their positivity and symmetries, are immediate consequences of the factorization \eqref{HHa} and the related properties of the usual $q,t$-Kostkas (cf. Proposition \ref{Proppos} and Corollary \ref{Cor_symK}). 

Next, we  define a deformation of the Nabla operator \cite{BG}, 
denoted $\nabla^B$, whose eigenfunctions are 
$H_{\lambda,\mu}(x,y;q,t^{-1})$ and whose eigenvalues are given by a specific ratio of two double Kostkas.
Somewhat surprisingly, we can evaluate the Schur expansion (which happens to be positive) of $\nabla^Bs_{\em,(n)}$ exactly (cf. Proposition \ref{propaction}).
{From} the ensuing expression,  
we deduce the following two results which will provide our most significant  connection with the hyperoctahedral group  (using the notation $ [n]_{q,t}=({q^n-t^n})/({q-t})$): 
\begin{align}&\qquad
\L \nabla^B s_{\emptyset;(n)} , s_{\emptyset;(n)} \R_B =  
\frac{1}{(qt)^{\binom{n}2}}\left[ 
\begin{array}{c}
2n\\
n
\end{array}
\right]_{q,t}  \label{qtcatal}\\
&\qquad
\L \nabla^B s_{\em,(n)} , p_{\em,(1^n)} \R_B= \left( \frac{ [n+1]_{q,t} + [n]_{q,t} }{(qt)^{(n-1)/2}}\right)^n , \label{2nqt}
\end{align}
where $\L\cdot, \cdot \R_B$ is the hyperoctahedral version of the Hall scalar product.

In \cite{HaiC}, Haiman conjectured that for every Coxeter group $W$ there exists a doubly graded quotient ring $R^W$ of the coinvariant ring $C^W$
whose Hilbert series ${\rm Hilb}_{q,t}(R^W)$ satisfies
\beq \label{specialcaseintro}
{\rm Hilb}_{t^{-1},t}(R^W) = t^{-hn/2} \left([h+1]_t \right)^n.
\eeq
The existence of such modules has been demonstrated in \cite{Gor}
using the representation theory of Cherednik algebras.  When 
$W=B_n$, this formula specializes to
\beq \label{eqhilbintro}
{\rm Hilb}_{t^{-1},t}(R^{B_n}) =  \left(\frac{[2n+1]_t}{t^n} \right)^n,
\eeq
which is exactly the rhs of \eqref{2nqt} when $q=t^{-1}$.  Furthermore,
it is known \cite{BEG} that the alternating component of $R^{B_n}$
is given by the $t$-Catalan for the reflection group $B_n$
\beq
t^{-n^2}\prod_{i=1}^n \frac{(1-t^{2i+2n})}{(1-t^{2i})} \, ,
\eeq
which again corresponds to the rhs of \eqref{qtcatal} specialized to $q=t^{-1}$.
It is thus natural to surmise 
that 
\beq
{\rm Frob}_{t^{-1},t}(R^{B_n}) \sim \nabla_{q=t^{-1}}^B s_{\emptyset;(n)}
\eeq
where the symbol $\sim$ means that the equality holds up to a 
relabeling of the indices (to ensure that the trivial module appears only
at bidegree $(0,0)$, we probably need, in view of Corollary~\ref{coro1},
to relabel the indexing pair of partitions as $s_{\lambda,\mu} \mapsto s_{\lambda',\mu'}$). 

As for the generic $q,t$-case, it 
does not appear that ${\rm Frob}_{q,t}(R^{B_n}) \sim \nabla^B s_{\emptyset;(n)}$
(at least not with the natural grading stemming from \cite{Gor}).  Already in the $n=2$ case, the $B_2$-Catalan is 
$[5]_{q,t}+qt[1]_{q,t}$ (see for instance \cite{stump}) while \eqref{qtcatal} 
gives
$([5]_{q,t}+q^2t^2[1]_{q,t})/qt$ (which is in some sense a homogenization of the 
$B_2$-Catalan).

\subsection{Outline}Apart from the brief Section \ref{secdef}, 
 reviewing the notation,  and the Conclusion, the article is essentially divided in two parts.
The first one, Section  \ref{hyperS}, is devoted to the study of the double Macdonald polynomials defined in Theorem \ref{Pdef1}. The pivotal result is the establishment of an equivalence between two scalar products, from which the factorization form \eqref{factor} is deduced. The rest of the section is concerned with the derivation of direct consequences of this main formula. The second main part, Section \ref{nablaS}, is concerned with the investigation of generalizations of the Nabla operator. After reviewing how this operator is defined in the usual case, we present our heuristic approach that yields, among other things, the results mentioned in the previous subsection.  In the short Conclusion, we reassert a new role for the Macdonald superpolynomials as the precise and fully explicit  objects that interpolates between the usual and double versions
 of Macdonald 
polynomials.

Four appendices complete this article. Appendix \ref{mpbi} contains the details of the statements made in the starting subsection concerning the transformation of the monomials and the power sums, from superspace to bisymmetric functions. Appendix \ref{proof_order} is mainly  concerned with the proof of the
equivalence between the two dominance
orderings \eqref{domidebut} and  \eqref{domibi} in the stable sector. This crucial result is relegated to an appendix because it is fairly technical and also because in disguised form, it is probably known. As discussed in the Conclusion, the results demonstrated here for the double Macdonald polynomials can be readily lifted to superspace. In Appendix~\ref{app_proof_conj}, we show that certain conjectured results in superspace are now validated in the stable sector, by matching the statements in \cite{BDLM1,BDLM2} with those of the present paper. Tables of $B_n$ Kostka coefficients up to $n=3$ are presented in 
Appendix~\ref{app_table}.

\begin{acknow}
 The authors are extremely grateful to Patrick Desrosiers and 
Stephen Griffeth for helpful discussions, and to Mark Haiman for sharing his data on wreath Macdonald polynomials.
This work was  supported by NSERC, FQRNT, 
FONDECYT (Fondo Nacional de Desarrollo Cient\'{\i}fico y
Tecnol\'ogico de Chile) grant \#1130696, and by CONICYT (Comisi\'on Nacional de Investigaci\'on Cient\'ifica y Tecnol\'ogica de Chile) via 
the ``proyecto anillo ACT56''.
\end{acknow}

\section{Definitions} \label{secdef}

A partition $\lambda=(\lambda_1,\lambda_2,\dots)$ of degree $|\lambda|=\sum_i \lambda_i$ is a vector of non-negative integers such that
$\lambda_i \geq \lambda_{i+1}$ for $i=1,2,\dots$.
 The length $\ell(\lambda)$
of $\lambda$ is the number of non-zero entries of $\lambda$.
Each partition $\lambda$ has an associated Ferrers diagram
with $\lambda_i$ lattice squares in the $i^{th}$ row,
from the top to bottom. Any lattice square in the Ferrers diagram
is called a cell (or simply a square), where the cell $(i,j)$ is in the $i$th row and $j$th
column of the diagram.  
The conjugate $\lambda'$ of  a partition $\lambda$ is the partition whose diagram is
obtained by reflecting  the diagram of $\lambda$ about the main diagonal.
Given a cell $s=(i,j)$ in $\lambda$, we let 
\begin{equation} \label{eqarms}
a_{\lambda}(s)=\lambda_i-j\, , \qquad {\rm and} \qquad l_{\lambda}(s)=\lambda_j'-i \,  .
\end{equation}
The quantities $a_{\lambda}(s)$ and $l_{\lambda}(s)$ 
are respectively called the arm-length and leg-length.
We will also need their co-version:
\begin{equation} \label{eqcoarms}
a'_{\lambda}(s)=j-1\, , \qquad {\rm and} \qquad l'_{\lambda}(s)=i-1 \,  .
\end{equation}
We say that the diagram $\mu$ is contained in $\la$, denoted
$\mu\subseteq \la$, if $\mu_i\leq \la_i$ for all $i$.  We also let
$\lambda+\mu$ be the partitions whose entries are $(\lambda+\mu)_i=\lambda_i+\mu_i$, and $\lambda \cup \mu$ be the partition obtained by reordering the entries of the concatenation of $\lambda$ and $\mu$.  The dominance ordering on partitions is such that $\lambda \geq \mu$ iff $|\lambda|=|\mu|$ and $\lambda_1+\cdots+\lambda_i \geq \mu_1+\cdots+\mu_i$ for all $i$.

The Macdonald polynomials $P_\la(x;q,t)$, in the variables 
$x=x_1,x_2,\dots$, are characterized by the two conditions \cite{Mac}
 \begin{equation}\label{co12}\begin{array}{lll} 1)& P_{\lambda}(x;q,t) =
m_{\lambda} + \text{lower terms},\\
&\\
2)&\L  P_{\la}, P_{\mu} \R_{q,t} =0\quad\text{if}\quad \la\ne\mu \, .
\end{array}\end{equation}
The triangular decomposition refers to the dominance order on partitions 
 and the $m_\la$'s are the monomial symmetric
functions:
\beq
m_\lambda = {\sum_{\sigma \in S_N}}' x_{\sigma(1)}^{\lambda_1} \cdots x_{\sigma(N)}^{\lambda_N} 
\eeq
where the prime indicates a sum over distinct terms $x_{\sigma(1)}^{\lambda_1} \cdots x_{\sigma(N)}^{\lambda_N}$.
The  orthogonality relation is defined in the power-sum basis $p_\la=p_{\la_1}\cdots p_{\la_\ell}$, with { $p_r=\sum_{i\geq 1} x^r_i$ }, as
\begin{equation}\label{ortoqt}  \L {p_\la},
{p_\mu}\R_{q,t}=  \delta_{\la\mu}\, z_{\lambda}(q,t)
\qquad\text{where}\qquad
z_{\la}(q,t)= z_{\lambda}
\prod_{i=1}^{\ell(\la)}\frac{1-q^{\la_i}}{1-t^{\la_i}} = 
 \prod_{i \geq 1} i^{n_{\la}(i)} {n_{\la}(i)!} \left(\prod_{i=1}^{\ell(\la)}\frac{1-q^{\la_i}}{1-t^{\la_i}}\right) ,
\end{equation}
$n_{\la}(i)$  being the number of parts in $\la$ equal to $i$. 
We stress the notational distinction between the scalar product $\L\cdot,\cdot\R_{q,t}$ and its  bisymmetric version $\LL \cdot, \cdot\RR_{q,t}$.

The Jack polynomials $P_{\lambda}(x;\alpha)$ and Schur functions $s_{\lambda}$
can be defined respectively as the limits $q=t^{\alpha}, t \to 1$ and $q=t$
of the Macdonald polynomials.  In the latter case, the scalar product
\eqref{ortoqt} reduces to the Hall scalar product $\L p_{\lambda},p_\mu \R =
\delta_{\lambda \mu} z_{\lambda}$ which is such that $\L s_{\lambda},s_\mu \R =
\delta_{\lambda \mu}$.

We will use the language of $\lambda$-rings (or plethysms) \cite{Ber,Las}.  
 The power-sum $p_i$ acts
on the ring of rational functions in $x_1,\dots,x_N,q,t$ with coefficients in a field $\mathbb K$ (usually taken to be $\mathbb Q$) as
\beq 
p_i \left[ \frac{   \sum_{\alpha} c_\alpha u_\alpha}{  \sum_{\beta} d_\alpha v_\beta}
  \right] = \frac{ \sum_{\alpha} c_\alpha u_\alpha^i}{ \sum_{\beta} d_\alpha v_\beta^i}
\eeq
where $c_\alpha,d_{\beta} \in \mathbb K$ and  where
$u_\alpha,v_{\beta}$ are monomials in $x_1,\dots,x_N,q,t$.
Since the power-sums form a basis of the ring of symmetric functions, this action extends uniquely to an action of
the ring of symmetric functions 
on the ring of rational functions in $x_1,\dots,x_N,q,t$ with coefficients in $\mathbb K$.  In this notation, a symmetric function $f(x)$
 is equal to $f[X]$, where 
$x=(x_1,x_2,\dots,x_N)$ and $X=x_1+x_2+\cdots+x_N$.

\section{Double Macdonald polynomials}\label{hyperS}

\subsection{A remarkable factorization property}
As indicated in the Introduction, we are interested in the stable bisymmetric version of the Macdonald polynomials as defined in Theorem
\ref{Pdef1}. The stability property is captured by the condition $m\geq n$. In other words,
if $m$ and $N-m$ are sufficiently large ($\geq |\lambda|+|\mu|$), then the bisymmetric
Macdonald polynomial $P_{\lambda,\mu}(x,y;q,t)$ stabilizes, in the sense that its monomial expansion becomes
independent of $m$ and $N$.  Within the stability sector, we can thus let $m \to \infty$ and 
$N-m \to \infty,$ and obtain ``double Macdonald functions'' indexed by two infinite sets of indeterminates $x=x_1,x_2,\dots$ and 
$y=y_1,y_2,\dots$ (corresponding respectively to $x_1,\dots,x_m$ and $x_{m+1},\dots,x_N$
in the limit $m \to \infty$ and 
$N-m \to \infty$).   The distinction between ``functions''and ``polynomials'' 
will not be necessary here and we shall always consider that our alphabets are finite. 

{From} now on, we will use
the ``plethystic'' notation which is central to the derivation of our results
(see Section~\ref{secdef}).  For the remainder of this article, $X$ and $Y$ will
stand respectively for
$X=x_1+x_2+\cdots+x_m$ and 
$Y=y_1+y_2+\cdots+y_{N-m}=x_{m+1}+\cdots+x_{N}$.

We first prove that the scalar product \eqref{newsp} can be rewritten in a 
much more convenient
way for our purposes.  It is in some sense the strongest result of this section.
\begin{lemma} \label{lemmaequivsp}
If $m \geq |\lambda|+|\mu|$ and 
$N-m \geq |\lambda|+|\mu|$, then the scalar product \eqref{newsp} is equal 
 to the
scalar product $\LL \cdot, \cdot \RR'$
defined as
\beq \label{scalarautre}
\LL p_{\lambda}\left[ X + \frac{q(1-t)}{1-qt}Y\right] 
 p_{\mu} \left[ Y\right],
p_{\nu}\left[ X + \frac{q(1-t)}{1-qt}Y\right] 
 p_{\omega} \left[ Y\right] \, \RR' = \delta_{\lambda \nu} \delta_{\mu \omega}
q^{|\lambda|}
z_{\lambda}(q,qt) z_{\mu}(qt,t) \, .
\eeq
\end{lemma}
\begin{proof}
 The scalar product \eqref{newsp} can be rewritten as
 \beq \label{newsp2}
\LL s_{\lambda}[X]
 p_{\mu}[X+Y],
s_{\nu}[X]
 p_{\omega}[X+Y]\RR_{q,t} = q^{|\lambda|}\L s_{\lambda},\,s_\nu\R\,\L
 p_{\mu},\,
 p_{\omega}\R_{q,t} = q^{|\lambda|} \delta_{\lambda \nu} \times \delta_{\mu \omega}
 z_{\mu}(q,t).
\eeq
Equation  \eqref{newsp2} is then equivalent to
\beq \label{newnewsp}
\LL p_{\lambda}\left[ X\right] 
 p_{\mu} \left[ X+Y\right],
p_{\nu}\left[ X \right] 
 p_{\omega} \left[ X+Y\right] \RR_{q,t} = \delta_{\lambda \nu} \delta_{\mu \omega}
q^{|\lambda|} z_{\lambda}z_{\mu}(q,t)
\eeq
given that the Hall scalar product 
is such that $\langle p_{\lambda}, p_{\nu} \rangle = \delta_{\lambda \nu} 
z_{\lambda}$.
Notice that by homogeneity the extra factor $q^{|\lambda|}$ can be carried from \eqref{newsp2} to 
\eqref{newnewsp}.

We now show that the scalar products 
$\LL \cdot, \cdot \RR_{q,t}$
and $\LL \cdot, \cdot \RR'$ coincide.
Let 
\beq p_r\left[X+ \frac{q(1-t)}{1-qt}Y \right]=u_r,\quad
p_r\left[Y \right]=v_r,\quad p_r\left[X\right]= \bar u_r,\quad\text{and}\quad 
p_r\left[X+Y\right]= \bar v_r.\eeq
 The quantities $u_r$ and $v_r$ 
are related to
$\bar u_r$ and $\bar v_r$ through 
\beq \label{reluv}
u_r=
p_r\left[\frac{(1-q)X}{1-qt} + \frac{q(1-t)(X+Y)}{1-qt} \right]=
 \frac{1-q^r}{1-q^rt^r} \bar u_r + \frac{q^r(1-t^r)}{1-q^rt^r} \bar v_r
\eeq
and
\beq \label{reluvbar}
v_r= p_r[-X+X+Y]=-\bar u_r+ \bar v_r \, .
\eeq
With respect to the scalar product $\LL \cdot , \cdot \RR'$, we have
(using  $\partial_{x}=\frac{\partial}{\partial x}$)
\beq \label{eqperpp}
u_r^{\perp'} = \frac{1-q^r}{1-q^rt^r} q^r r \partial_{u_r} \qquad
{\rm and} \qquad v_r^{\perp'} = \frac{1-q^rt^r}{1-t^r}  r \partial_{v_r}
\eeq
while with respect to the scalar product $\LL \cdot , \cdot \RR_{q,t}$, 
we have
\beq \label{eqperp}
\bar u_r^{\perp} = q^r r \partial_{\bar u_r} \qquad
{\rm and} \qquad \bar v_r^{\perp} = \frac{1-q^r}{1-t^r}  r \partial_{\bar v_r}
\eeq
where $h^\perp$ and $h^{\perp'}$ are defined respectively 
such that
$\LL h f , g \RR_{q,t}= \LL f , h^\perp g \RR_{q,t}$ and
$\LL h f , g \RR'= \LL f , h^{\perp'} g \RR'$.
Given that $\{u_{\lambda} v_{\mu}\}_{\lambda,\mu}$ is a basis of the space of bisymmetric functions of a given total degree $n=|\lambda|+|\mu|$, the lemma will follow if we can show that
\begin{equation} 
\LL u_{\lambda} v_{\mu} , u_{\nu} v_{\omega}  \RR' 
=
\LL u_{\lambda} v_{\mu} ,  u_{\nu}  v_{\omega}  \RR_{q,t} \, ,
\end{equation}
or equivalently, that
\begin{equation} 
\LL  u_{\emptyset} v_{\emptyset} ,  u_{\lambda}^{\perp'} v_{\mu}^{\perp'}  u_{\nu}  v_{\omega}  \RR' 
= \LL  u_{\emptyset} v_{\emptyset} ,  u_{\lambda}^\perp v_{\mu}^{\perp}  u_{\nu}  v_{\omega}  \RR_{q,t} \, .
\end{equation}
In order to do so, it suffices to compare the recursions induced by $u_r$
and $v_r$;  repeated applications of the recursions will, by homogeneity, either lead to zero or
\begin{equation} 
\LL u_{\emptyset} v_{\emptyset} ,  u_{\emptyset} v_{\emptyset}    \RR_{q,t} 
=
\LL u_{\emptyset} v_{\emptyset} ,   u_{\emptyset} v_{\emptyset}   \RR' = 1
\end{equation}
and the result will follow.   

 Note that the conditions 
$m \geq |\lambda|+|\mu|$ and $N-m \geq |\lambda|+|\mu|$
ensure that $u_1,u_2,\dots$ and $v_1,v_2,\dots$ can be considered 
independent (and similarly for  $\bar u_1,\bar u_2,\dots$ and $\bar v_1,\bar v_2,\dots$).  Observe also from \eqref{reluv} and \eqref{reluvbar}
that $\partial_{\bar u_r}$ and $\partial_{\bar v_r}$
commute with $u_s$ and $v_s$ if $r\neq s$.
Now let $\lambda=(r^k) \cup \hat \lambda$, $\mu=(r^m) \cup \hat \mu$,
$\nu=(r^\ell) \cup \hat \nu$ and $\omega=(r^n) \cup \hat \omega$,
where $\hat \lambda$,  $\hat \mu$,  $\hat \nu$ and  $\hat \omega$
do not contain parts of size $r$. 
 On the one hand, we have
from \eqref{eqperpp}
\beq
\LL u_{\lambda} v_{\mu} ,  u_{\nu}  v_{\omega}  \RR' 
= \LL u_r^k v_r^m
u_{\hat \lambda} v_{\hat \mu} , u_r^\ell v_r^n
 u_{\hat \nu}  v_{\hat \omega}  \RR' 
= \LL u_r^{k-1} v_r^m
 u_{\hat \lambda} v_{\hat \mu} , u_r^{\perp'} u_r^\ell v_r^n
 u_{\hat \nu}  v_{\hat \omega}  \RR' = q^r r \ell \frac{1-q^r}{1-q^rt^r}
\LL u_r^{k-1} v_r^m
 u_{\hat \lambda} v_{\hat \mu} , u_r^{\ell-1} v_r^n
u_{\hat \nu}  v_{\hat \omega}  \RR' \label{eqcomp1}
\eeq
On the other hand, using \eqref{reluv}, \eqref{eqperp}, 
and the chain rule for derivatives, e.g.,
\beq \frac{\d}{\d \bar u_r}=\frac{\d u_r}{\d \bar u_r}\frac{\d}{\d u_r}+\frac{\d v_r}{\d \bar u_r}\frac{\d}{\d v_r},
\eeq 
we get
\begin{align}
 \LL u_{\lambda} v_{\mu} &,  u_{\nu}  v_{\omega}  \RR_{q,t} \nonumber \\
& = \LL u_r^k v_r^m
u_{\hat \lambda} v_{\hat \mu} , u_r^\ell v_r^n
 u_{\hat \nu}  v_{\hat \omega}  \RR_{q,t} \nonumber \\
& = 
\frac{1-q^r}{1-q^rt^r}  \LL u_r^{k-1} v_r^m
u_{\hat \lambda} v_{\hat \mu} , \bar u_r^\perp  u_r^\ell v_r^n
 u_{\hat \nu}  v_{\hat \omega}  \RR_{q,t} +
 \frac{q^r(1-t^r)}{1-q^rt^r}  \LL u_r^{k-1} v_r^m
u_{\hat \lambda} v_{\hat \mu} , \bar v_r^\perp  u_r^\ell v_r^n
 u_{\hat \nu} v_{\hat \omega}  \RR_{q,t} \nonumber  \\
& = 
\left( \frac{1-q^r}{1-q^rt^r} \right)^2 q^r r \ell  \LL u_r^{k-1} v_r^m
u_{\hat \lambda} v_{\hat \mu} , u_r^{\ell-1} v_r^n
 u_{\hat \nu}  v_{\hat \omega}  \RR_{q,t} - 
 \frac{1-q^r}{1-q^rt^r} q^r r n  \LL u_r^{k-1} v_r^m
u_{\hat \lambda} v_{\hat \mu} , u_r^{\ell} v_r^{n-1}
 u_{\hat \nu}  v_{\hat \omega}  \RR_{q,t} \nonumber  \\
& + 
\left( q^r \frac{1-t^r}{1-q^r t^r} \right)^2
r \ell \frac{1-q^r}{1-t^r}  \LL u_r^{k-1} v_r^m
u_{\hat \lambda} v_{\hat \mu} , u_r^{\ell-1} v_r^n
 u_{\hat \nu} v_{\hat \omega}  \RR_{q,t} + q^r \frac{1-t^r}{1-q^rt^r}rn 
\frac{1-q^r}{1-t^r} \LL u_r^{k-1} v_r^m
u_{\hat \lambda} v_{\hat \mu} , u_r^{\ell} v_r^{n-1}
 u_{\hat \nu}  v_{\hat \omega}  \RR_{q,t} \nonumber  \\
& =  q^r r \ell \frac{1-q^r}{1-q^rt^r}
\LL u_r^{k-1} v_r^m
 u_{\hat \lambda} v_{\hat \mu} , u_r^{\ell-1} v_r^n
 u_{\hat \nu}  v_{\hat \omega}  \RR_{q,t} \label{eqcomp2}
\end{align}
Comparing \eqref{eqcomp1} and \eqref{eqcomp2} we see that the two $u_r$ recursions coincide.  

We now do the same for the
$v_r$ recursions. From \eqref{eqperp} we find
\beq
\LL u_{\lambda} v_{\mu} , u_{\nu} v_{\omega}  \RR' 
= \LL u_r^k v_r^m
u_{\hat \lambda} v_{\hat \mu} , u_r^\ell v_r^n
 u_{\hat \nu}  v_{\hat \omega}  \RR' 
= \LL u_r^{k} v_r^{m-1}
 u_{\hat \lambda} v_{\hat \mu} , v_r^{\perp'} u_r^\ell v_r^n
 u_{\hat \nu}  v_{\hat \omega}  \RR' = \frac{1-q^rt^r}{1-t^r} nr
\LL u_r^{k} v_r^{m-1}
 u_{\hat \lambda} v_{\hat \mu} , u_r^{\ell} v_r^{n-1}
 u_{\hat \nu}  v_{\hat \omega}  \RR'
\eeq
while, using \eqref{reluvbar} and \eqref{eqperp}, we get
\begin{align}
\LL u_{\lambda} v_{\mu} ,  u_{\nu} v_{\omega}  \RR_{q,t}
& = \LL u_r^k v_r^m
u_{\hat \lambda} v_{\hat \mu} , u_r^\ell v_r^n
 u_{\hat \nu}  v_{\hat \omega}  \RR_{q,t} \nonumber \\
& = 
\LL u_r^{k} v_r^{m-1}
u_{\hat \lambda} v_{\hat \mu} , \bar v_r^\perp  u_r^\ell v_r^n
 u_{\hat \nu}  v_{\hat \omega}  \RR_{q,t} -
 \LL u_r^{k} v_r^{m-1}
u_{\hat \lambda} v_{\hat \mu} , \bar u_r^\perp  u_r^\ell v_r^n
u_{\hat \nu}  v_{\hat \omega}  \RR_{q,t} \nonumber \\
& = 
\frac{1-q^r}{1-t^r}r \ell q^r \frac{1-t^r}{1-q^rt^r} \LL u_r^{k} v_r^{m-1}
u_{\hat \lambda} v_{\hat \mu} , u_r^{\ell-1} v_r^n
 u_{\hat \nu} v_{\hat \omega}  \RR_{q,t} + 
 \frac{1-q^r}{1-t^r} r n  \LL u_r^{k} v_r^{m-1}
u_{\hat \lambda} v_{\hat \mu} , u_r^{\ell} v_r^{n-1}
 u_{\hat \nu}  v_{\hat \omega}  \RR_{q,t} \nonumber \\
& - 
q^rr \ell  \frac{1-q^r}{1-q^r t^r}  \LL u_r^{k} v_r^{m-1}
u_{\hat \lambda} v_{\hat \mu} , u_r^{\ell-1} v_r^n
 u_{\hat \nu}  v_{\hat \omega}  \RR_{q,t} + q^r rn  \LL u_r^{k} v_r^{m-1}
u_{\hat \lambda} v_{\hat \mu} , u_r^{\ell} v_r^{n-1}
u_{\hat \nu} v_{\hat \omega}  \RR_{q,t} \nonumber \\
& =   \frac{1-q^rt^r}{1-t^r} rn
\LL u_r^{k} v_r^{m-1}
 u_{\hat \lambda} v_{\hat \mu} , u_r^{\ell} v_r^{n-1}
 u_{\hat \nu}  v_{\hat \omega}  \RR_{q,t}
\end{align}
The two $v_r$ recursions are thus also seen to coincide, which completes
the proof of the assertion that the scalar products $\LL \cdot, \cdot \RR'$
and $\LL \cdot, \cdot \RR_{q,t}$ are equal.
\end{proof}
We now have all the tools to establish our key result;
the double Macdonald polynomials have a totally unexpected, albeit rather non-trivial, decomposition into a product of two ordinary Macdonald polynomials. As mentioned in the introduction, the special case $\lambda=\emptyset$
of the factorization is essentially contained in \cite{BDF}.
\begin{theorem}\label{teofac} If $m \geq |\lambda|+|\mu|$ and 
$N-m \geq |\lambda|+|\mu|$, then
\beq \label{facto}
P_{\lambda,\mu}(x,y;q,t) = P_{\lambda}^{(q,qt)} \left[ X + \frac{q(1-t)}{1-qt}Y\right] P_{\mu}^{(qt,t)} \left[ Y\right] \, ,\eeq
where $P_{\lambda}^{(q,t)}(x)$ stands for the usual Macdonald polynomial
 $P_{\lambda}(x;q,t)$.
\end{theorem}

It should be commented that although 
the factorization of Theorem~\ref{teofac} seems highly asymmetric in 
$X$ and $Y$, the plethystic substitution \eqref{plethysubs} will transform
it into the much more symmetrical expression \eqref{HH}.

\begin{proof} We need to show that the products of Macdonald polynomials 
appearing on the rhs of  \eqref{facto} 
are unitriangular when expanded in the monomial basis and orthogonal with respect to the scalar product $\LL \cdot, \cdot \RR_{q,t}$ defined in
\eqref{newsp}.

The products of Macdonald polynomials
$P_{\lambda}^{(q,qt)} \left[ X + \frac{q(1-t)}{1-qt}Y\right] P_{\mu}^{(qt,t)} \left[ Y\right]$ are (basically by definition) orthogonal with respect to the
scalar product $\LL \cdot, \cdot \RR'$
defined in \eqref{scalarautre} (notice that by homogeneity the extra factor $q^{|\lambda|}$ does not have any effect on the orthogonality).  By 
Lemma~\ref{lemmaequivsp}, the orthogonality is then immediate.

We now show the unitriangularity of 
$
P_{\lambda}^{(q,qt)} \left[ X + \frac{q(1-t)}{1-qt}Y\right] P_{\mu}^{(qt,t)} \left[ Y\right]
$ in the monomial basis.  By triangularity of the Macdonald polynomials, 
both in the Schur and monomial bases, we
get
\beq
P_{\lambda}^{(q,qt)} \left[ X + \frac{q(1-t)}{1-qt}Y\right] P_{\mu}^{(qt,t)} \left[ Y\right] = \sum_{\nu \leq \lambda } * \,
s_\nu \left[ X + \frac{q(1-t)}{1-qt}Y\right] \sum_{\omega \leq \mu} *\,m_\omega[Y],
\eeq 
where $*$ stands for some irrelevant coefficients (that we will keep denoting $*$).
We have
\beq
s_{\nu} \left[ X + \frac{q(1-t)}{1-qt}Y\right]
=\sum_{\rho \subseteq \nu, \gamma \subseteq \nu} c^{\nu}_{\rho \gamma} \,s_{\rho}[X] \,
s_{\gamma}  \left[\frac{q(1-t)}{1-qt}Y\right]=
\sum_{\rho \subseteq \nu, \sigma : |\sigma|+|\rho|=|\nu|} * \, s_{\rho}[X] \,
s_{\sigma}  \left[ Y \right]
\eeq
Note that in the last equation, we only used $|\sigma|=|\gamma|$
and $|\gamma|+|\rho|=|\nu|$ (there is no triangularity when
$s_{\gamma}  \left[\frac{q(1-t)}{1-qt}Y\right]$ is expanded in the
$s_{\sigma}[Y]$ basis).  Moreover, it is an elementary fact that
$$
s_{\sigma}[Y] m_{\omega}[Y] = m_{\sigma+\omega}[Y] +
\sum_{\beta < \sigma+\omega} * \, m_{\beta}[Y] \, .
$$
Therefore
\beq
P_{\lambda}^{(q,qt)} \left[ X + \frac{q(1-t)}{1-qt}Y\right] P_{\mu}^{(qt,t)} \left[ Y\right] = \sum_{\nu \leq \lambda, \omega \leq \mu} 
\sum_{\rho \subseteq \nu, \sigma : |\sigma|+|\rho|=|\nu|}
\sum_{\beta \leq \sigma+\omega}
* \,
s_\rho [ X]\, m_\beta[Y]
\eeq
{From} $\lambda \geq \nu$ and $\nu \supseteq \rho$, we obtain
\beq
\lambda_1 + \cdots + \lambda_i \geq \nu_1 + \cdots + \nu_i \geq
\rho_1 + \cdots + \rho_i 
\eeq
while from $\mu \geq \omega$, $ |\rho|+ |\sigma|=|\nu|=|\lambda|$ and $\omega+\sigma \geq \beta$, we get
\beq
|\lambda| + \mu_1 + \cdots + \mu_i \geq   |\rho|+ |\sigma| + \omega_1 + \cdots + \omega_i   \geq
 |\rho|+ \omega_1+\sigma_1 + \cdots + \omega_i+\sigma_i \geq |\rho|+\beta_1+\cdots + \beta_i \, ,
\eeq
which proves the triangularity.  The unitriangularity is immediate from the unitriangularity of the Macdonald polynomials when expanded in the Schur or the monomial basis.
\end{proof}

For later references, we state explicitly the Jack limit.
\begin{corollary}  
If $m \geq |\lambda|+|\mu|$ and 
$N-m \geq |\lambda|+|\mu|$, then in the limit $q=t^\alpha, t \to 1$ we obtain
\beq \label{eqjack}
P_{\lambda,\mu}(x,y;\alpha) = P_{\lambda}^{(\alpha/(\alpha+1))} \left[ X + \frac{1}{\alpha+1}Y\right] P_{\mu}^{(\alpha+1)} \left[ Y\right] \, ,\eeq
where $P_{\lambda}^{(\alpha)}(x)$ stands for the usual Jack polynomial
 $P_{\lambda}(x;\alpha)$.  
\end{corollary}
It is important to stress that in \eqref{eqjack}, the plethystic notation is such that $p_r$ acts on the ring of rational functions in the variables $x_1,\dots,x_m,y_1,\dots,y_{N-m}$ over the field $\mathbb Q(\alpha)$.  Consequently,
$\alpha$ is not affected by the plethysm\footnote{This is easily seen: 
$$ p_r\left[ X + \frac{1}{\alpha+1}Y\right]=\lim_{\substack{q=t^\a\\ t\rw 1}}p_r\left[ X + \frac{q(1-t)}{1-qt}Y\right]= \lim_{\substack{q=t^\a\\ t\rw 1}}\left(p_r\left[ X\right] + \frac{q^r(1-t^r)}{1-q^rt^r}p_r\left[Y\right]\right)= p_r[X]+  \frac{1}{\alpha+1}p_r[Y] \, .
$$}, that is,
\beq
p_r\left[ X + \frac{1}{\alpha+1}Y\right] = p_r[X]+  \frac{1}{\alpha+1}\, p_r[Y] \, .
\eeq

\subsection{Norm}  
Using Theorem~\ref{teofac} we obtain rather directly the expression for the norm of the double Macdonald polynomials.

\begin{corollary}\label{cor_norm}
The norm of the double Macdonald polynomial $P_{\lambda,\mu}(x,y;q,t)$ is:
\begin{equation} \label{normmacdobn}
\LL P_{\la,\mu}(x,y;q,t), P_{\la,\mu}(x,y;q,t) \RR_{q,t} =q^{|\la|} b_\la(q,qt)^{-1} b_\mu(qt,t)^{-1} =:b_{\la,\mu}(q,t)^{-1}  , 
\end{equation}
where 
\begin{equation}
b_\la(q,t) = \prod_{s\in \la} \frac{1-q^{a(s)} t^{l(s)+1}}{1-q^{a(s)+1} t^{l(s)}}.
\end{equation}
\end{corollary}
\begin{proof}
Let $Z=X + \frac{q(1-t)}{1-qt}Y$.   In Lemma~\ref{lemmaequivsp} we have shown that the scalar product \eqref{newsp} is equivalent to the scalar product
\begin{equation}
\LL p_\la[Z] p_\mu[Y] , p_\nu[Z] p_\omega[Y] \RR' = \delta_{\la \nu} \delta_{\mu \omega} q^{|\la|} z_\la(q,qt) z_\mu(qt,t) = q^{|\la|}\langle p_\la, p_\nu \rangle_{q,qt} \times \langle p_\mu, p_\omega \rangle_{qt,t} \, .
\end{equation}
Therefore, the factorized form of $P_{\la,\mu}(x,y;q,t)$ also implies a factorization of its scalar product:
\begin{equation}
\LL P_{\la,\mu}(x,y;q,t), P_{\la,\mu}(x,y;q,t) \RR_{q,t} =q^{|\la|} \langle P_\la(x;q,qt) , P_\la(x;q,qt) \rangle_{q,qt} \times \langle P_{\mu}(x;qt,t) , P_{\mu}(x;qt,t)\rangle_{qt,t} \, , 
\end{equation} 
and the result follows from the norm of the Macdonald polynomials \cite[VI.4.11]{Mac}
\begin{equation}
\langle P_\la (x;q,t) , P_\la(x;q,t) \rangle_{q,t} =  b_\la(q,t)^{-1}.
\end{equation}
\end{proof}
If we define the normalized version of the double Macdonald polynomials as
\begin{equation}
Q_{\la,\mu}(x,y;q,t) = b_{\la,\mu}(q,t) P_{\la,\mu}(x,y;q,t),
\end{equation}
where $b_{\la,\mu}(q,t)=\LL P_{\la,\mu}, P_{\la,\mu}\RR_{q,t}^{-1}$  was defined in \eqref{normmacdobn},  we obtain
\begin{equation}\label{psQP}
\LL Q_{\la,\mu}(x,y;q,t) , P_{\nu,\omega}(x,y;q,t) \RR_{q,t} = \delta_{\la \nu} \delta_{\mu \omega}.
\end{equation}

\subsection{Kernel}
The form \eqref{newsp} of the scalar product leads to a natural 
generalization of the Macdonald kernel.  Let $u$  (resp. $v$)
be the union of the alphabets $(x_1,x_2,\dots)$ and  $(y_1,y_2,\dots)$
(resp.  $(z_1,z_2,\dots)$ and  $(w_1,w_2,\dots)$).  Defining
\begin{equation}
\Pi:=\Pi(u,v;q,t)= \prod_{i,j} \frac{(t u_i v_j ; q)_\infty}{( u_i v_j ; q)_\infty} \; \prod_{ k,l} \frac{1}{1-q^{-1}x_k z_l}\label{ker}
\end{equation}
we obtain by standard manipulations that
\begin{equation}
\Pi = \left\{  \sum_{\mu} z_\mu(q,t)^{-1}  p_\mu(x,y)p_\mu(z,w) \right\} \;  \left\{  \sum_{\la} \frac{1}{q^{|\la|}}  s_\la( x)s_\la( z) \right\}=
 \sum_{\la, \mu} z_{\la,\mu}(q,t)^{-1} p_{\la,\mu}(x,y) p_{\la,\mu}(z,w) \, .
\end{equation}
The duality \eqref{psQP} then implies that
\begin{equation}
\Pi = \sum_{\la,\mu} P_{\la,\mu}(x,y;q,t) Q_{\la,\mu}(z,w;q,t).
\end{equation}

\subsection{Specializations}  In \cite{BDLM1}, a picture 
describing the various specializations of the Macdonald polynomials in superspace was presented.  Figure~\ref{newlimits} gives the corresponding picture in the case of the double Macdonald polynomials.  In the figure,
there are two Hall-Littlewood limits ($P_{\lambda,\mu}(x,y;t)$
and $\bar P_{\lambda,\mu}(x,y;t)$), one Jack limit ($P_{\lambda,\mu}^{(\alpha)}(x,y)$)
and a one-parameter Schur limit ($s_{\lambda,\mu}(x,y;t)=P_{\lambda,\mu}(x,y;q,t)$) that specializes to the corresponding limits of the Hall-Littlewood and Jack limits.
\begin{figure}[ht]
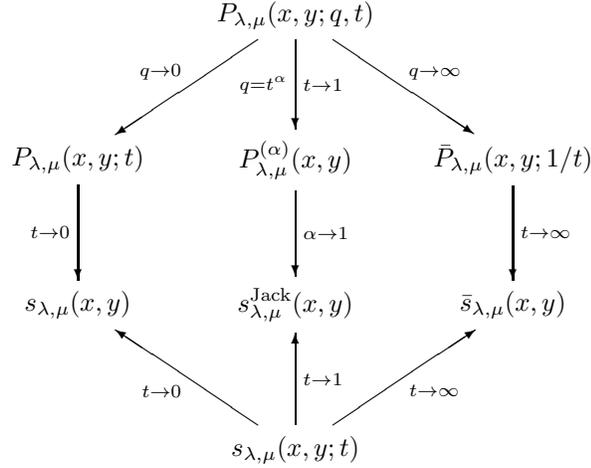

\caption{{\footnotesize  Limiting  cases of the 
double Macdonald polynomials}}
\label{newlimits}
$$\begin{diagram}
\node{}\node{P_{\lambda,\mu}(x,y;q,t)} \arrow{sw,t}{q\to 0}\arrow{s,lr}{q=t^\alpha}{t\to 1}\arrow{se,t}{q\to \infty} \node{}\\
\node{P_{\lambda,\mu}(x,y;t)} \arrow{s,l}{t\to 0} 
\node{P_{\lambda,\mu}^{(\alpha)}(x,y)}
\arrow{s,r}{\alpha\to 1} \node{\bar P_{\lambda,\mu}(x,y;1/t)}
 \arrow{s,r}{t\to \infty} \\
\node{s_{\lambda,\mu}(x,y)}  \node{s^\mathrm{Jack}_{\lambda,\mu}(x,y)}   \node{\bar s_{\lambda,\mu}(x,y)}  \\
\node{} \node{s_{\lambda,\mu}(x,y;t)}\arrow{nw,b}{t\to 0}\arrow{n,r}{t\to 1}\arrow{ne,b}{t\to \infty} \node{}
\end{diagram}
$$
\end{figure}
Given the factorized form \eqref{facto}
of $P_{\lambda,\mu}(x,y;q,t)$, we can give each of these limits explicitly.
The Jack limit was presented in \eqref{eqjack}.  The Hall-Littlewood limits are
\beq
P_{\lambda,\mu}(x,y;t) = s_{\lambda} (x)  P_{\mu}(y;t)\qquad {\rm and}
\qquad \bar P_{\lambda,\mu}(x,y;t) = s_{\lambda}\left[X+(1-1/t)Y \right] P_{\mu}(y;1/t) \, ,
\eeq
where $P_{\mu}(y;t)$ is the Hall-Littlewood polynomial (the limit $q=0$ of the 
corresponding Macdonald polynomial). 
Finally, the Schur limit is 
\beq
s_{\lambda,\mu}(x,y;t)= P_{\lambda}^{(t,t^2)} \left[ X + \frac{t}{1+t}Y\right] P_{\mu}^{(t^2,t)} \left[ Y\right] \, ,
\eeq
which specializes to
\beq
s_{\lambda,\mu}(x,y) = s_{\lambda}(x) s_{\mu}(y) \,, \qquad
s^\mathrm{Jack}_{\lambda,\mu}(x,y) = P_{\lambda}^{(1/2)} \left[ X + \frac{1}{2}Y\right] P_{\mu}^{(2)} \left[ Y\right] \qquad {\rm and} \qquad \bar s_{\lambda,\mu}(x,y) =
s_{\lambda}[X+Y] s_{\mu}(y) \, ,
\eeq
where we recall that $P^{(\alpha)}(y)$ is the Jack polynomial.
Not considered in Figure~\ref{newlimits} are the limits $q=1$ and $t=1$,
which give respectively 
\beq
e_{\lambda,\mu}(x,y)=e_{\lambda}[X+Y]s_{\mu}(y) \qquad
{\rm and} \qquad m_{\lambda,\mu}(x,y)=s_\lambda(x)m_\mu(y)
\eeq
the analogs of the elementary and monomial symmetric functions.

\subsection{Duality} 
Let $\omega_X$ be the standard involution
\begin{equation}
\omega_X \, p_r[X] = (-1)^{r-1}p_r[X] = (-1)^r p_r[-X]
\end{equation}
which is such that $\omega_X s_{\la}[X]=s_{\la'}[X]$.  The involution
 $\omega = \omega_{X} \circ \omega_{Y}$ is such that for any elements $a(q,t)$ and $b(q,t)$
of $\mathbb Q(q,t)$ we have
\beq
\omega p_r\left[a(q,t)X + b(q,t) Y \right] = 
a(q^r,t^r) \, \omega_X \, p_r[X] +  b(q^r,t^r)\, 
\omega_Y \, p_r[Y] =  (-1)^{r-1}  p_r\left[a(q,t)X + b(q,t) Y \right] \, .
\eeq
Hence $\omega$ acts as the usual involution on symmetric functions in any alphabet made out of a combination of $X$ and $Y$.  

Let us now define the more general automorphism $\omega^B_{q,t}$ as
\begin{align}
& \omega_{q,t}^B \, p_r[X+Y]  =(-1)^{r-1}\frac{1-q^{-r}}{1-t^{-r}}p_r[X+Y] = \omega \, p_r\left[ \frac{t}{q} \left( \frac{1-q}{1-t} \right) (X+Y) \right] \\
& \omega^B_{q,t} \, p_r[X]=(-1)^r t^r p_r[X]= \omega \, p_r[-tX]
 & 
\end{align}
Clearly, the inverse of $\omega_{q,t}^B$ is given by
\begin{equation}\label{OmegaInv}
(\omega_{q,t}^{B})^{-1}  = \left(\frac{q}{t}\right)^n  \omega^B_{t^{-1}, q^{-1}}.
\end{equation}
\begin{proposition}\label{propdu} The following dualities hold:
\begin{align}
 \omega_{q,t}^B \, P_{\la,\mu}(x,y;q,t) &= Q_{\mu', \la'}(x,y;t^{-1},q^{-1}), 
\label{stat1}\\
\omega_{q,t}^B\, Q_{\la,\mu}(x,y;q,t) &= (t/q)^n P_{\mu', \la'}(x,y;t^{-1},q^{-1}), \label{stat2}
\end{align}
where $|\la|+|\mu|=n$.   In particular (see \eqref{psQP}):
\begin{equation} \label{dualomqt}
\LL \omega_{q,t}^B \,  P_{\mu',\la'}(x,y;q,t)  ,   P_{\nu,\omega}(x,y;t^{-1},q^{-1})  \RR_{t^{-1},q^{-1}} = \delta_{\la\nu} \delta_{\mu \omega}.
\end{equation}
\end{proposition} \label{propdual}
Observe that as expected, 
the $B_n$-analogue of the conjugation sends the pair of partitions
$\mu,\lambda$ to the pair of partitions $\lambda',\mu'$.
\begin{proof}  We will prove  \eqref{stat1}.  Relation
\eqref{OmegaInv} will then imply \eqref{stat2}.

First, from the factorized form of the double Macdonald polynomial \eqref{facto}, we have:
\begin{align}
P_{\mu',\la'}(x,y;t^{-1},q^{-1}) &= P_{\mu'}^{(t^{-1},(qt)^{-1})}\left[ X + \frac{(q^{-1})(1-t^{-1})Y  }{1-(qt)^{-1}}\right] P_{\la'}^{( (qt)^{-1},q^{-1}  )}[Y]  \nonumber \\
& = P_{\mu'}^{(t ,qt )}\left[ X + \frac{(1-q)Y  }{1-qt}\right] P_{\la'}^{( qt,q  )}[Y]
\end{align}
where, in the second equality, we used the symmetry $P_\la^{(q^{-1}, t^{-1})}[X] = P_\la^{(q,t)}[X]$.  The duality can now be computed explicitly.  We have
\begin{align}
\omega^B_{q,t} \, P_{\la,\mu}(x,y;q,t) &= \omega P_{\la}^{(q,qt)}\left[ t \frac{1-q}{1-qt} Y \right] \, \omega P_{\mu}^{(qt,t)} \left[  \left(\frac{t}{q} \right)   \frac{1-qt}{1-t} \left(   X+\frac{Y(1-q)}{1-qt}\right) \right] \\
& = t^{|\la|} b_{\la'}(qt,q) P_{\la'}^{(qt,q)}[Y] \left( \frac{t}{q} \right)^{|\mu|} b_{\mu'}(t,qt) P_{\mu'}^{(t,qt)} \left[  X + \frac{(1-q)Y}{1-qt}  \right]
\end{align}
where we have used 
$P_\la^{(q, t)}[\tau Z] = \tau^{|\la|}P_\la^{(q,t)}[Z]$ ($\tau$ stands for any monomial in $q$ and $t$) and the usual duality \cite{Mac}
\begin{equation}
\omega P_\la^{(q,t)}\left[  \frac{1-q}{1-t} Z\right] = b_{\la'}(t,q) P_{\la'}^{(t,q)}[Z].
\end{equation}
Using the relation $b_\la(q,t) = (t/q)^{|\la|}b_\la(q^{-1},t^{-1})$, we then obtain
\begin{equation}
 \omega_{q,t}^B \, P_{\la,\mu}(x,y;q,t) = t^{|\mu|} b_{\mu'}(t^{-1}, (qt)^{-1}) \,
b_{\la'}((qt)^{-1},q^{-1}) P_{\mu', \la'}(x,y;t^{-1},q^{-1}) \, ,
\end{equation}
which completes the proof.
\end{proof}
The 
scalar product \eqref{newsp} does not behave well in the
limits $q=t=0$ and $q=t=\infty$.
An interesting consequence of Proposition~\ref{propdual} is that it can be used to relate the corresponding specializations of the double Macdonald polynomials.  In effect, from the definition of $\omega_{q,t}^B$, we have
that \eqref{dualomqt} is equivalent to
\begin{equation} 
\LL \omega^B_{1,1} \,  P_{\mu',\la'}(x,y;q,t)  ,   P_{\nu,\omega}(x,y;t^{-1},q^{-1})  \RR_{1,1} = \delta_{\la\nu} \delta_{\mu \omega}.
\end{equation}
The limit $q=t=\infty$ of this result is then well defined and reads
\begin{equation} \label{dualom}
\LL \omega^B_{1,1} \, 
 s_{\mu'}[X+Y] s_{\la'}(y)  ,  
 s_\nu(x) s_{\omega}(y)  \RR_{1,1} = \delta_{\la\nu} \delta_{\mu \omega}.
\end{equation}
The following  
identity involving products of four
Littlewood-Richardson coefficients (which to the best of our knowledge is not in
the literature) follows from the previous 
duality.  
\begin{proposition}\label{les4c} For any partitions 
$\lambda,\mu,\nu$ and $\omega$, we have
\beq
\sum_{\gamma,\eta,\sigma,\tau} (-1)^{|\tau|} 
c^{\gamma}_{\tau' \nu}
c^{\lambda}_{\gamma \eta} c^{\sigma}_{\eta \mu}  c^\omega_{\sigma \tau}
 =  \delta_{\lambda \nu} \delta_{\mu \omega}
\eeq
where $\gamma,\eta,\sigma,\tau$ run over 
all partitions and where $c^{\lambda}_{\mu \nu}$ is the corresponding
Littlewood-Richardson coefficient.   Equivalently, if as in \cite{BVO} we let
\begin{equation}
c^\nu_{\lambda \mu \eta}=\sum_{\xi} c^{\xi}_{\lambda \mu} c^{\nu}_{\xi \eta}
\end{equation}
the identity reads
\begin{equation}
\sum_{\eta,\tau} (-1)^{|\tau|} 
c^{\lambda}_{\tau' \nu \eta}
 c^{\omega}_{\eta \mu \tau} 
 =  \delta_{\lambda \nu} \delta_{\mu \omega}
\end{equation}
\end{proposition}
\begin{proof} 
The scalar product \eqref{newsp} in the limit $q=t=1$ is such that
\begin{equation} \label{orthospec}
\LL s_{\mu}[X] \, s_{\lambda}[X+Y], s_{\gamma}[X]\, s_{\nu}[X+Y] \RR_{1,1}
= \delta_{\mu \gamma} \delta_{\lambda \nu} \, .
\end{equation}
Moreover, the action of $\omega^B_{1,1}$ on $s_{\mu}[X] \, s_{\lambda}[X+Y]$ is easily seen to be
\beq
\omega^B_{1,1} \, s_{\mu}[X] \, s_{\lambda}[X+Y] = (-1)^{|\mu|}s_{\mu}[X] \, s_{\lambda'}[X+Y] \, .
\eeq
It will thus prove convenient to expand  $s_{\mu'}[X+Y] s_{\la'}(y)$ and  
$s_\nu(x) s_{\omega}(y)$ in the  $s_{\mu}[X] \, s_{\lambda}[X+Y]$ basis. We have
\beq \label{eqdeuxfois}
s_{\lambda'}[Y] = s_{\lambda'}[-X+X+Y]=
\sum_{\gamma,\eta} c^{\lambda'}_{\gamma \eta} s_{\gamma}[-X] s_{\eta}[X+Y]  
= \sum_{\gamma,\eta} c^{\lambda'}_{\gamma \eta} (-1)^{|\gamma|} s_{\gamma'}[X] 
s_{\eta}[X+Y] \, .  
\eeq
Hence
\beq \label{expansion1}
\omega_{1,1}^B \, s_{\mu'}[X+Y] s_{\la'}[Y] = \omega_{1,1}^B \sum_{\gamma,\eta,\sigma} 
 (-1)^{|\gamma|} c^{\lambda'}_{\gamma \eta} c_{\eta \mu'}^{\sigma} s_{\gamma'}[X] 
s_{\sigma}[X+Y] = \sum_{\gamma,\eta,\sigma} 
 c^{\lambda'}_{\gamma \eta} c_{\eta \mu'}^{\sigma} s_{\gamma'}[X] 
s_{\sigma'}[X+Y] \, .
\eeq
Similarly, after some straightforward computations (using again 
\eqref{eqdeuxfois}),  we obtain
\begin{equation} \label{expansion2}
s_{\nu}[X] s_{\omega}[Y] =
\sum_{\tau,\pi,\xi} (-1)^{|\tau|} c^{\xi}_{\nu \tau'}
c^{\omega}_{\pi \tau}  s_{\xi}[X] s_{\pi}[X+Y]  \, .
\end{equation}
Finally, replacing the expansions \eqref{expansion1} and 
\eqref{expansion2} in \eqref{dualom}, we get thanks to
\eqref{orthospec} the identity
\beq
\sum_{\gamma,\eta,\sigma,\tau} (-1)^{|\tau|} c^{\lambda'}_{\gamma \eta} c^{\sigma}_{\eta \mu'} c^{\gamma'}_{\nu \tau'} c^\omega_{\sigma' \tau} =  \delta_{\lambda \nu} 
\delta_{\mu \omega}\, .
\eeq
The proposition then follows since
\beq
\sum_{\gamma,\eta,\sigma,\tau} (-1)^{|\tau|} c^{\lambda'}_{\gamma \eta} c^{\sigma}_{\eta \mu'} c^{\gamma'}_{\nu \tau'} c^\omega_{\sigma' \tau} =  
\sum_{\gamma,\eta,\sigma,\tau} (-1)^{|\tau|} c^{\lambda}_{\gamma' \eta'} c^{\sigma'}_{\eta' \mu} c^{\gamma'}_{\nu \tau'} c^\omega_{\sigma' \tau} =
\sum_{\gamma,\eta,\sigma,\tau} (-1)^{|\tau|} 
c^{\gamma}_{\tau' \nu}
c^{\lambda}_{\gamma \eta} c^{\sigma}_{\eta \mu}  c^\omega_{\sigma \tau} \, ,
\eeq
where we used $c^{\lambda}_{\mu \nu}= c^{\lambda'}_{\mu' \nu'}$ and 
$c^{\lambda}_{\mu \nu} = c^{\lambda}_{\nu \mu}$.
\end{proof}

\subsection{Evaluation}  We now provide an explicit formula for the  evaluation of the double Macdonald polynomials. The first point to clarify is the way we could specialize the variables $x$ and $y$. For this we recall that the most general evaluation of the usual Macdonald polynomial is \cite[eqs VI (6.16)-(6.17)]{Mac}
\begin{equation}\label{defw}
P_\la^{(q,t)}\left[ \frac{1-u}{1-t} \right] = \prod_{s \in \la} \frac{ t^{l'(s)} -q^{a'(s)} u }{1-q^{a(s)}  t^{l(s)+1}  } =: w_\la(u;q,t)
\end{equation}
where the plethysm is such that:
\begin{equation}
p_r\left[\frac{1-u}{1-t} \right] = \frac{1-u^r}{1-t^r}.
\end{equation}
Considering the factorized expression \eqref{factor}, one sees that in order to evaluate the first term $P_\la^{(q,qt)} \left[ X+ \frac{q(1-t)}{1-qt}Y\right] $ we need to bring the argument in the proper form, that is, set $X$ and $Y$ such that
\begin{equation}
\label{coeva} X+\frac{q(1-t)}{1-qt}Y=\frac{1-u}{1-qt}.\end{equation}
This clearly requires $X$ and $Y$ to be of the form
\begin{equation}X=q^a(1+qt+\cdots+(qt)^m)\qquad\text{and}\qquad Y=t^b(1+t+\cdots + t^{N-m}),\end{equation}
and the condition \eqref{coeva} further imposes $a=1-m$ and $b=m$ (so that the resulting $u$ is $q^mt^N$). 
We thus define the evaluation as
\begin{equation}\label{defeva}
\mathrm E_{N,m}\bigl(P_{\la,\mu}(x,y;q,t)\bigr)= P_{\la,\mu}(x,y;q,t) \Big|_{x_i=\frac{t^{i-1}}{q^{m-i}}, y_i=t^{m+i-1}}.
\end{equation}

\begin{corollary} \label{coroeval}
With the evaluation defined in \eqref{defeva} and $w_{\lambda}(u;q,t)$ defined in \eqref{defw}, we have
\begin{equation} \label{stableEval}
\mathrm E_{N,m}\bigl(P_{\la,\mu}(x,y;q,t)\bigr) = \frac{t^{m|\mu|}}{q^{(m-1)|\la|}  }  w_\la(q^mt^{N};q,qt) w_\mu(t^{N-m};qt,t).
\end{equation}
\end{corollary}
\begin{proof}
The specialization gives
\begin{equation}
X=x_1+x_2+ \ldots + x_m = \frac{1}{q^{m-1}} + \frac{t}{q^{m-2}} + \ldots + t^{m-1}=q^{1-m}\,\frac {1-(qt)^m}{1-qt}
\end{equation}
and
\begin{equation}
Y=y_{1}+\ldots + y_{N-m} =  t^m + \ldots + t^{N-1}=t^m\,\frac{1-t^{N-m}}{1-t}.
\end{equation}
The evaluation of the two  terms in the factorization \eqref{factor} is immediate:
\begin{align}
\mathrm E_{N,m }\left(P_\la^{(q,qt)} \left[ X+ \frac{q(1-t)}{1-qt}Y\right]\right) 
= \frac{1}{q^{(m-1)|\la|}} P_\la^{(q,qt)} \left[  \frac{1-q^mt^N}{1-qt} \right] = \frac{1}{q^{(m-1)|\la|}} w_\la(q^mt^N;q,qt)
\end{align}
and
\begin{equation}
\mathrm E_{N,m}( P_{\mu}^{(qt,t)}[Y] ) =  P_{\mu}^{(qt,t)}\left[ t^m\, \frac{1-t^{N-m}}{1-t} \right]  =  t^{m |\mu|} P_{\mu}^{(qt,t)}\left[ \frac{1-t^{N-m}}{1-t} \right] = t^{m|\mu|}w_\mu(t^{N-m};qt,t).
\end{equation}
whose product yields the announced result.
\end{proof}

\subsection{Kostka coefficients}

Recall that the integral form of the Macdonald polynomials is  \cite[VI.8.3]{Mac}
\beq
J_{\lambda}(x;q,t) = c_{\lambda}(q,t) P_{\lambda}(x;q,t)
\eeq
where
\beq
c_\la(q,t) = \prod_{s\in \la} (1-q^{a(s)} t^{l(s)+1}) \, .
\eeq
We define the integral form of the double Macdonald polynomials to be
\beq\label{Bnintform}
J_{\lambda,\mu} (x,y;q,t) =  c_{\lambda}(q,qt) \, c_{\mu} (qt,t) \, P_{\lambda,\mu}(x,y;q,t)\, .
\eeq 
Recall also the definition of the modified Macdonald polynomials
\beq \label{Kusuel}
H_{\lambda}(x;q,t)= J_{\lambda}^{(q,t)}\left[\frac{X}{1-t}\right]
\eeq
which, when expanded in the Schur basis,
are such that (this is equivalent to  \cite[VI.8.11]{Mac})
\beq
H_{\lambda}(x;q,t) = \sum_{\mu} K_{\mu \lambda} (q,t)\, s_{\mu}(x)
\eeq
where $K_{\mu \lambda} (q,t)$ is the $q,t$-Kostka coefficient.  It has been
shown in \cite{Hai} that  $K_{\mu \lambda} (q,t) \in \mathbb N[q,t]$.
Recall also that $K_{\mu \lambda} (1,1)=\chi^\mu_{(1^n)}$, where $\chi^\mu_{(1^n)}$ is value of the irreducible $S_n$-character $\chi^\mu$ at the class of the identity \cite[VI.8.16]{Mac}. Equivalently, $K_{\mu \lambda} (1,1)$ corresponds to the number of standard tableaux of shape $\mu$ (cf. \cite[eqs (2.1) and (9.7)]{Ber}).

By analogy, we
define 
\beq
H_{\lambda,\mu}(x,y;q,t) = \varphi(J_{\lambda,\mu}(x,y;q,t))
\eeq
where $\varphi$ is the the homomorphism whose action on the power sums is
\begin{equation} 
p_n[X] \mapsto p_n[X] \qquad {\rm and} \qquad p_n[X+Y] \mapsto \frac{1}{1-t^n}
p_n[X+Y] \, .
\end{equation}
Note that the homomorphism $\varphi$ is equivalent to 
the plethystic substitution 
\beq \label{plethysubs}
X \mapsto X \qquad {\rm and} \qquad X+Y \mapsto \frac{1}{(1-t)} (X+Y)
\eeq
 From Theorem~\ref{teofac}, after some straightforward manipulations, we obtain 
that
\begin{align} \label{HH}
H_{\lambda,\mu}(x,y;q,t) &=  J_{\lambda}^{(q,qt)} \left[ \frac{X}{1-qt} + \frac{qY}{1-qt}\right] J_{\mu}^{(qt,t)} \left[ \frac{Y}{1-t} + \frac{tX}{1-t}\right]\nonumber \\ &
= H_{\lambda}^{(q,qt)}[X+qY] \, H_{\mu}^{(qt,t)}[tX+Y] \, . 
\end{align}
Our interest is to introduce $B_n$ analogues of the Kostka coefficients by expanding 
$H_{\lambda,\mu}(x,y;q,t)$ in terms of the Schur functions associated to the irreducible characters of $B_n$ (see \cite[p. 178]{Mac}): 
\beq s_{\la,\mu}(x,y)=s_\la(x)\, s_\mu(y).
\label{sxy}
\eeq
Observe that from Theorem~\ref{teofac}, these Schur functions correspond
to  the specialization $q=t=0$ of $P_{\la,\mu}(x;q,t)$, namely 
$s_{\la,\mu}(x,y)=P_{\la,\mu}(x,y;0,0)$ (see also \cite{BDLM1,BDLM2}).

Now define the double Kostka coefficients $ K_{\kappa, \gamma \; \la,\mu}(q,t)$ through
the expansion 
\begin{equation}
H_{\la,\mu}(x,y;q,t)= \sum_{\kappa,\gamma} K_{\kappa, \gamma \; \la,\mu}(q,t) s_{\kappa, \gamma}(x,y).
\end{equation}
The reader is referred to Appendix~\ref{app_table} for tables of 
 double Kostka coefficients up to degree $n=3$.
We first connect $K_{\kappa, \gamma \; \la,\mu}(1,1)$ to the representation theory
of the hyperoctahedral group $B_n$.
\begin{proposition} \label{Prop1}
 Let $\lambda$ and $\mu$ be such that $|\lambda|+|\mu|=n$.
Then  $K_{\kappa, \gamma \; \la,\mu}(1,1)$ is the dimension of the irreducible
representation of $B_n$ indexed by the pairs of partitions $\kappa,\gamma$.
In particular,  $K_{\kappa, \gamma \; \la,\mu}(1,1)$ does not depend on $\lambda$ and
$\mu$.
\end{proposition}
\begin{proof}  When $q=t=1$, it is known \cite[eq. (9.6)]{Ber} that 
\beq
H_{\lambda}(x;1,1)=p_1^{|\lambda|}.
\eeq
It thus follows from  \eqref{HH} that
\beq
H_{\lambda,\mu}(x,y;1,1) = H_{\lambda}^{(1,1)}[X+Y] H_{\mu}^{(1,1)}[X+Y]
=\left( p_1[X+Y] \right)^{|\lambda|+|\mu|}=  p_{1^n}[X+Y]  \, .
\eeq
But it is known (see {\cite[p. 178]{Mac}}) that
\beq 
p_{\rho}[X+Y] p_{\sigma}[X-Y] = \sum_{\kappa,\gamma} \chi^{\kappa,\gamma}_{\rho,\sigma} \, s_{\kappa,\gamma} (x,y)
\eeq
where $ \chi^{\kappa,\gamma}_{\rho,\sigma}$ stands for the $B_n$-character indexed by the irreducible representation $\kappa,\gamma$ at the class indexed by $\rho,\sigma$.  If we set $\rho=1^n$
and $\sigma= \emptyset$, we obtain that
\beq
H_{\lambda,\mu}(x,y;1,1) =p_{1^n}[X+Y] = \sum_{\kappa,\gamma} \chi^{\kappa,\gamma}_{1^n,\emptyset} \, s_{\kappa,\gamma} (x,y)
\eeq
which proves the proposition since the class indexed by $1^n,\emptyset$ is the
class of the identity, in which case the character yields the dimension of the representation.
\end{proof}
Proposition~\ref{Prop1} implies that
$K_{\kappa, \gamma \; \la,\mu}(1,1)$ is the number of pairs of standard Young tableaux  of respective shapes $\kappa$ and $\gamma$ filled (without repetitions) 
with the numbers $\{1,2,3,\ldots,n\}$, where $n=|\kappa|+|\gamma|$. For instance, $K_{(2,1),(1) \; \la,\mu}(1,1)=8$ is the number of pairs of standard Young tableaux of shape $(2,1),(1)$:
\begin{equation}\begin{split}
&{\tableau[scY]{1&2\\3 \\}}\; ,\;{\tableau[scY]{4\\  \bl{}}} \quad \qquad {\tableau[scY]{1&3\\2 \\}}\; ,\;{\tableau[scY]{4\\  \bl{}}} \quad \qquad  {\tableau[scY]{1&2\\4 \\}}\; ,\;{\tableau[scY]{3\\  \bl{}}} \quad\qquad {\tableau[scY]{1&4\\2 \\}}\; ,\;{\tableau[scY]{3\\  \bl{}}} \\
& \\
 &{\tableau[scY]{1&3\\4 \\}}\; ,\;{\tableau[scY]{2\\  \bl{}}} \quad\qquad  {\tableau[scY]{1&4\\3 \\}}\; ,\;{\tableau[scY]{2\\  \bl{}}} \quad\qquad {\tableau[scY]{2&3\\4 \\}}\; ,\;{\tableau[scY]{1\\  \bl{}}} \quad\qquad {\tableau[scY]{2&4\\3 \\}}\; ,\;{\tableau[scY]{1\\  \bl{}}}
\end{split}
\end{equation}

We now show the positivity of the coefficients
$K_{\kappa, \gamma \; \la, \mu}(q,t)$.
The proposition implies that the double Kosktas (of $B_n$-type)
can be obtained from the usual 
$q,t$-Kostka (of $S_n$-type). 
Therefore, from the combinatorial point of view, this new hyperoctahedral perspective does not shed any light on the usual $q,t$-Kostka. In Section \ref{Con}, we will discuss how an
interesting new combinatorics appears in the non-stable sector that interpolates between
the $S_n$ and $B_n$ cases.

\begin{proposition}\label{Proppos}
We have
\begin{equation}\label{bipartkostka}
K_{\kappa, \gamma \; \la, \mu}(q,t) = \sum_{\nu,\omega,\alpha, \beta, \rho, \sigma}  K_{\nu \la}(q,qt) \,K_{\omega \mu}(qt,t)\,  c_{\alpha \beta}^\nu\, c_{\rho \sigma}^\omega\, c_{\alpha \rho}^\kappa \,c_{\beta \sigma}^\gamma \,q^{|\beta|} t^{|\rho|}
\end{equation}
where  $c_{\alpha \beta}^\la$ is the corresponding
Littlewood-Richardson coefficient. In particular, 
 $K_{\kappa, \gamma \; \la,\mu}(q,t) \in \mathbb N[q,t]$.
\end{proposition}

\begin{proof} From \eqref{Kusuel} and \eqref{HH}, we have
\begin{align}
H_{\lambda,\mu}(x,y;q,t)& =\sum_{\nu,\omega}   K_{\nu \lambda}(q,qt)\,
 K_{\omega \mu} (qt,t) \,s_{\nu}[X+qY]\,s_{\omega}[tX+Y] \\
& =\sum_{\nu,\omega, \alpha, \beta,\rho,\sigma}  K_{\nu \lambda}(q,qt)\,
 K_{\omega \mu}(qt,t) \,c_{\alpha \beta}^\nu \,c^\omega_{\rho \sigma}\,
s_{\alpha}[X] \,s_{\beta}[qY] \,s_{\rho}[tX]\, s_{\sigma}[Y] \\
& =\sum_{\nu,\omega,\alpha,\beta,\rho,\sigma}  K_{\nu \lambda}(q,qt)\,
 K_{\omega \mu}(qt,t)\, c_{\alpha \beta}^\nu\, c^\omega_{\rho \sigma}  \,q^{|\beta|}t^{| \rho |}
\,s_{\alpha}[X] \,s_{\rho}[X]\, s_{\beta}[Y]\, s_{\sigma}[Y] \\
& =\sum_{\nu,\omega,\alpha,\beta, \rho, \sigma, \kappa,\gamma}  K_{\nu \lambda}(q,qt)\,
 K_{\omega \mu}(qt,t) \,c_{\alpha \beta}^\nu \,c^\omega_{\rho \sigma} \,c^\kappa_{\alpha \rho} \,c^{\gamma}_{\beta \sigma} q^{|\beta|} t^{|\rho|} \,s_{\kappa}[X]\, s_{\gamma}[Y] \\
& =\sum_{\kappa, \gamma} \left( \sum_{\nu,\omega,\alpha,\beta,\rho, \sigma}  K_{\nu \lambda}(q,qt)\,
 K_{\omega \mu}(qt,t) \,c_{\alpha \beta}^\nu \,c^\omega_{\rho \sigma} \,c^\kappa_{\alpha \rho} \,c^{\gamma}_{\beta \sigma}\,  q^{|\beta|} t^{|\rho |} \right) s_{\kappa,\gamma}(x,y).
\end{align}
This proves the first assertion.   The positivity of 
$K_{\kappa, \gamma \; \la, \mu}(q,t)$ then follows from the positivity of the
$q,t$-Kostkas and the Littlewood-Richardson coefficients. 
\end{proof}
\begin{remark}\label{rem1}
  Propositions~\ref{Prop1}  and \ref{Proppos} suggest that there exists a bigraded module
of the regular representation of the  hyperoctahedral group $B_n$ whose
Frobenius series  corresponds to the Schur expansion of
$H_{\lambda, \mu}(x,y;q,t)$.  This point will be discussed further in
Remark~\ref{remfrob}.
\end{remark}

\begin{corollary} \label{Cor_symK} The double Kostka coefficients have the following symmetries:
\begin{equation}\label{symbipartK}
K_{\kappa, \gamma \; \la, \mu}(q,t) = q^{\bar n(\mu', \la')}\,t^{\bar n(\la, \mu)} \,K_{\gamma', \kappa' \; \la,\mu}(q^{-1}, t^{-1}) \qquad {\rm and} \qquad K_{\kappa, \gamma \; \la, \mu}(q,t) =K_{\gamma', \kappa' \; \mu',\la'}(t, q) ,
\end{equation}
where $\bar n(\la,\mu) = n(\la)+ |\mu| + n(\mu')+n(\mu)$.
\end{corollary}
\begin{proof}   For the usual Kostka coefficients, the analogous form of the first symmetry is equivalent to (see e.g., \cite[eq. (9.9)]{Ber}) 
\beq
H_{\lambda}(x;q,t)= q^{n(\lambda')} \,t^{n(\lambda)}\, \omega H_{\lambda}(x;q^{-1},t^{-1})
\eeq
where $\omega$ is such that $\omega s_{\mu}=s_{\mu'}$.  If we extend $\omega$
to the space of bisymmetric functions as
$\omega s_{\mu}[X]= s_{\mu'}[X]$ and $\omega s_{\mu}[Y]= s_{\mu'}[Y]$
(which implies for instance that $\omega s_{\lambda}[X+qY]=s_{\lambda'}[X+qY]$\footnote{This is seen from
$$\omega s_{\lambda}[X+qY]=\sum_{\mu,\nu} c^{\lambda}_{\mu \nu} s_{\mu}[X] q^{|\nu|} s_{\nu}[Y] =\sum_{\mu,\nu} c^{\lambda}_{\mu \nu} s_{\mu'}[X] q^{|\nu|} s_{\nu'}[Y]=\sum_{\mu,\nu} c^{\lambda'}_{\mu' \nu'} s_{\mu'}[X]  s_{\nu'}[qY] 
=s_{\lambda'}[X+qY]$$.}), 
we obtain from
 \eqref{HH} 
\begin{align}
H_{\lambda,\mu}(x,y;q,t)& =  q^{n(\lambda')} (qt)^{n(\lambda)} \omega \left( H_{\lambda}^{(q^{-1},(qt)^{-1})}[X+qY] \right)  (qt)^{n(\mu')} t^{n(\mu)} \omega \left( H_{\mu}^{((qt)^{-1},t^{-1})}[tX+Y] \right) \nonumber \\
 & =  q^{n(\lambda')} (qt)^{n(\lambda)} q^{|\lambda|} \omega \left( H_{\lambda}^{(q^{-1},(qt)^{-1})}[q^{-1}X+Y] \right)  (qt)^{n(\mu')} t^{n(\mu)} t^{|\mu|} \omega \left( H_{\mu}^{((qt)^{-1},t^{-1})}[X+t^{-1}Y] \right) \nonumber \\
 & =   q^{\bar n(\mu', \la')}t^{\bar n(\la, \mu)} 
 \omega \left( H_{\lambda}^{(q^{-1},(qt)^{-1})}[q^{-1}X+Y] \right)  \omega \left( H_{\mu}^{((qt)^{-1},t^{-1})}[X+t^{-1}Y] \right) \\
 & =   q^{\bar n(\mu', \la')}t^{\bar n(\la, \mu)} 
 \omega  H_{\lambda,\mu}(y,x;q^{-1},t^{-1})
\end{align}
But this amounts to 
\beq
H_{\lambda,\mu}(x,y;q,t) =  q^{\bar n(\mu', \la')}t^{\bar n(\la, \mu)}  \omega_B H_{\lambda,\mu}(x,y;q^{-1},t^{-1}) 
\eeq
where $\omega_B s_{\lambda,\mu}(x,y)=s_{\mu',\lambda'}(x,y)$, which implies the 
first relation.

For the second relation, we use the other Macdonald symmetry \cite[eq. (9.8)]{Ber}
\beq
H_{\lambda}(x;q,t) = \omega H_{\lambda'}(x;t,q)
\eeq
to obtain
\beq
H_{\lambda,\mu}(x,y;q,t) =   \omega \left( H_{\lambda'}^{(qt,q)}[X+qY] \right)  
\omega \left( H_{\mu'}^{(t,qt)}[tX+Y] \right) = \omega H_{\mu',\lambda'}(y,x;t,q) 
\, .
\eeq
But this is equivalent to 
\beq \label{omegaBrel}
H_{\lambda,\mu}(x,y;q,t) 
=\omega_B H_{\mu',\lambda'}(x,y;t,q)
\eeq
and the result follows.
\end{proof}

It is known that 
for $|\lambda|=n$  \cite[Propo. 3.5.20]{Haiman} or \cite[ex. 2 p. 362]{Mac}
\beq \label{Kostka1}
K_{(n) \la}(q,t) = t^{n(\la)} \qquad {\rm and} \qquad K_{(1^n) \la}(q,t) = q^{n(\la')} 
\eeq 
These correspond to the two inequivalent representations of 
dimension 1 of $S_n$.
In the case of $B_n$, there are 4 inequivalent representations of dimension 1. 
The corresponding double  Kostkas, which will prove important in the next section,
are given in the next Corollary.
\begin{corollary}\label{speckostka}
Let $|\la|+|\mu|=n$.  We have
\begin{equation}
K_{(n),\emptyset \; \la, \mu}(q,t) = q^{n(\la)}t^{|\mu|+n(\la) + n(\mu)} \qquad {\rm and}
\qquad K_{\emptyset, (n) \; \la, \mu}(q,t)=q^{|\la|+n(\la)}t^{n(\la)+n(\mu)} \, .
\end{equation}
By symmetry, i.e. using the first relation of \eqref{symbipartK}, we 
also have
\begin{equation} \label{eqspec}
K_{\emptyset, (1^n) \; \la, \mu}(q,t)=q^{|\la|+n(\la')+n(\mu')}t^{n(\mu')} \qquad 
{\rm and}
\qquad K_{(1^n), \emptyset \; \la, \mu}(q,t)= q^{n(\la') + n(\mu')}t^{|\mu|+n(\mu')}.
\end{equation}
\end{corollary}
\begin{proof}
We will only prove the expression for $K_{\emptyset, (n) \; \la, \mu}(q,t)$ since that 
of $K_{(n), \emptyset \; \la, \mu}(q,t)$ can be obtained in a similar way.  From
\eqref{HH}, $K_{\emptyset, (n) \; \la, \mu}(q,t)$ is equal to the coefficient of
$s_{\emptyset,(n)}(x,y)=s_{(n)}[Y]$ in $H_{\la,\mu}(x,y;q,t)$.  In the product form \eqref{HH} we can thus set $X=0$ and search for the coefficient of $s_{(n)}[Y]$ in
$H_{\lambda}^{(q,qt)}[qY] H_{\mu}^{(qt,t)}[Y]=q^{|\lambda|}
H_{\lambda}^{(q,qt)}[Y] H_{\mu}^{(qt,t)}[Y]$.
Let us first expand the latter product:
\beq H_{\lambda}^{(q,qt)}[Y] H_{\mu}^{(qt,t)}[Y]= \sum_{\eta,\rho} 
K_{\eta\la}(q,qt) \, K_{\rho\mu}(qt,t) \, s_\eta[Y]\, s_\rho[Y]=
\sum_{\eta,\rho,\om} 
K_{\eta\la}(q,qt) \, K_{\rho\mu}(qt,t)\,c_{\eta \rho}^\om \, s_\om[Y].\eeq
For $\om=(n)$, it follows that the non-vanishing Littlewood-Richardson coefficients $c_{\eta \rho}^{(n)}$ are equal to $\delta_{\eta, (n_1)}\delta_{\rho ,(n_2)}$ for non-negative integers $n_1$ and $n_2$ satisfying $n_1+n_2=n$. The 
homogeneity of the Kostka coefficients then imposes
 $n_1=|\lambda|$ and
$n_2=|\mu|$.
 The desired  coefficient is thus equal to
$q^{|\lambda|} K_{(n_1) \,\lambda}(q,qt)  K_{(n_2)\, \mu}(qt,t)$.  Using \eqref{Kostka1}, the result follows. 
\end{proof}

\section{The Nabla Operator}\label{nablaS}

\subsection{Review of the usual case: Nabla operator and Frobenius series}
A useful technique to demonstrate that a given symmetric polynomial is Schur positive is to link this expansion to the decomposition of a representation into irreducible ones.
By means of the characteristic map,
 a symmetric polynomial is transformed into a {class function} 
 of $S_n$. 
Under this map, a Schur function is associated to an  irreducible character 
of $S_n$. It follows that a symmetric function $F$ associated to a representation is Schur positive: this symmetric function can be expanded as $F=\sum c_\la s_\la$, where by construction $c_\la$ is a nonnegative integer since it corresponds
to the multiplicity of the irreducible representation $\la$ in the larger representation under consideration.

 Let 
\beq
\tilde H_{\lambda}(x;q,t) = t^{n(\lambda)} H_{\lambda}(x;q,t^{-1})
= \sum_{\mu} \tilde K_{\mu \lambda}(q,t) s_\mu(x)\,,
\eeq
(hence $\tilde K_{\mu \lambda}(q,t)=  t^{n(\lambda)}  K_{\mu \lambda}(q,t^{-1}) $).
Garsia and Haiman \cite{GH} were able to construct bigraded $S_n$-modules 
$\mathcal H_\mu$ (the so-called Garsia-Haiman modules)
such that
\beq
 {\rm Frob}_{q,t} (\mathcal H_\mu) = \tilde H_{\lambda}(x;q,t) \, .
\eeq
(More precisely, this result is conjectured in \cite{GH} and proved in \cite{Hai}.)
The Garsia-Haiman modules belong to the larger module $\mathcal D_n$
of diagonal harmonics whose
dimension is ${(n+1)^{n-1}}$. 
Quite remarkably, the bigraded Hilbert series of $\mathcal D_n$ can be calculated from the
the operator $\nabla$
introduced by Bergeron and Garsia \cite{BG} (see also \cite{Betal}
and \cite{Ber}, beginning of Section 9.6) and defined as follows: 
\begin{equation}\nabla \tilde H_\la=q^{n(\la')}\, t^{n(\la)} \tilde H_\la \, .
\end{equation}
In other words, $\nabla$ is defined from its action on $\tilde H_{\la}$ and because these polynomials form a basis, this provides an action on any symmetric function. It has been proven (and this is a highly non-trivial result)
that the action of $\nabla$  on $e_n$, when expanded in the Schur basis,
\begin{equation}\nabla e_n=\sum_\la c_\la(q,t) s_\la\,,\end{equation}
can be interpreted as the bigraded Frobenius series of $\mathcal D_n$ 
(see e.g., \cite[Theorem 4.2.5]{Haiman} and Sections 3.5, 4.1 and 4.2 therein for the subsequent results of this section).
Moreover, in the case $q=1/t$, it has been shown that
\beq \label{tfactor} 
\L\nabla_{q=1/t} e_n,p_1^n\R=\left(\frac{[n+1]_t}{t^{{n}/{2}}}\right)^{n-1}\,,
\eeq
where
\beq [k]_t=\frac{(1-t^k)}{(1-t)} = (1+ t + \ldots + t^{k-1}).
\eeq
Specializing this result to $t=1$ one recovers   $(n+1)^{n-1}$ as the dimension of $\mathcal D_n$.  Note that for $q$ and $t$ generic, there is no factorization of the type
\beq \label{nopower}
\L\nabla e_n,p_1^n\R= \bigl( f(q,t) \bigr)^{n-1}\,,
\eeq
for some polynomial $f(q,t) \in \mathbb N[q,t]$.  
The coefficient of  $e_n$ in $\nabla e_n$ 
\beq
C_n(q,t) = \L \nabla e_n, e_n\R
\,, \eeq
is also of particular interest.
It corresponds to the bigraded Hilbert series of the subspace
$\mathcal A_n$ of alternants in $\mathcal D_n$.  There is a combinatorial
interpretation for $C_n(q,t)$ (reviewed in \cite[Chapter 3]{Hag}), but there is no known closed-form expression.  However, when $q=1/t$, $C_n(q,t)$ reduces to 
\beq \label{catalant}
C_{n}(t^{-1},t) = t^{-\binom{n}{2}}\frac{1}{[n+1]_t}
\left[ 
\begin{array}{c}
2n\\
n
\end{array}
\right]_{t} \,,
\eeq
which is a $t$-analogue of the Catalan number
\beq
C_n = \frac{1}{n+1}\binom{2n}{n} \, .
\eeq
The dimension of the subspace of alternants $\mathcal A_n$ is thus given
by the $n$-th Catalan number.

\subsection{Heuristic strategy for the Nabla operator in the  hyperoctahedral case}

Proposition~\ref{Prop1} provides a connection between the hyperoctahedral group  and the double Macdonald polynomials.  In this subsection, we 
will deepen this connection by the introduction of Nabla operators that
will produce another $B_n$ datum, namely the $B_n$ version of the dimension  formula $(n+1)^{n-1}$.

Recall that  the group $S_n$ is a Coxeter group of type $A$, namely $A_{n-1}$. 
The formula $(n+1)^{n-1}$ turns out to be the $A_{n-1}$ version
of the dimension formula  $(h+1)^r$ \cite{Gor}, where $h$ and $r$ are respectively the Coxeter number and the rank of the corresponding Coxeter group ($h=n$ and $r=n-1$ for $A_{n-1}$).
Now the hyperoctahedral group of rank $n$ is equivalent to the group of signed permutations of $n$ objects. Its generators are $\{s_0, s_1, . . . , s_{n-1}\}$, where $s_0$ is the sign change,  and where the remaining generators are the 
elementary transpositions that generate $S_n$. This description  
makes clear that the 
hyperoctahedral group is the Coxeter group $B_n$. Given that
 for $B_n$ the Coxeter number and the rank are respectively $h=2n$ and $r=n$,  the expected $B_n$-form of the dimension is $(2n+1)^n$. This is precisely the value we will obtain.

We first need to formulate the correct definition of the Nabla operators in the
$B_n$-case.  In order to do so, let us briefly reexamine the usual case from an heuristic-constructive point of view.  From \eqref{Kostka1}, the eigenvalue of $\nabla$ on $\tilde H_\la$ corresponds to 
$K_{(1^n)\la}(q,t^{-1})/ K_{(n)\la}(q,t^{-1})$, that is, to the ratio of the Kostka coefficients associated to representations of $S_n$ of dimension 1.
This is the guiding observation that we will use to define 
the $B$ version of the Nabla operator.

As noted before Corollary~\ref{speckostka}, 
in our case there are  four pairs of partitions whose Kostka coefficient is equal to 1 when $q=t=1$:
\begin{equation} \label{dimension1}
K_{(n),\em\; \la,\mu}(1,1)=K_{(1^n),\em \; \la,\mu}(1,1)=K_{\em,(n) \; \la,\mu}(1,1)=K_{\em,(1^n) \; \la,\mu}(1,1)=1.
\end{equation}
There are thus four choices for the rescaling factor that sets equal to 1 a given double Kostka.  We will choose to rescale
$ K_{(1^n),\em \; \la,\mu}(q,t)$ to 1 by redefining
$ H_{\lambda,\mu}(x,y;q,t)$ as (this choice has no impact on the definition of a Nabla operator)
\beq \label{defHwigle}
\tilde H_{\lambda,\mu}(x,y;q,t)= \frac{1}{ K_{(1^n),\em \; \la,\mu}(q,t^{-1})}
H_{\lambda,\mu}(x,y;q,t^{-1}) =  \sum_{\kappa,\gamma} \tilde K_{\kappa, \gamma \; \la,\mu}(q,t) s_{\kappa, \gamma}(x,y).
\eeq
By \eqref{eqspec} we have the more explicit relationship between $\tilde H_{\lambda,\mu}$ and $ H_{\lambda,\mu}$:
\beq
\tilde H_{\lambda,\mu}(x,y;q,t) 
= q^{-n(\la') - n(\mu')}t^{|\mu|+n(\mu')} H_{\lambda,\mu}(x,y;q,t^{-1}) \, .
\eeq
Observe also that $ \tilde K_{\kappa, \gamma \; \la,\mu}(q,t)$ now belongs to
$\mathbb N[q^{\pm 1},t^{\pm 1}]$.

Following our guiding observation and \eqref{dimension1}, 
there are two natural and independent ways to 
define a Nabla operator in the $B$ case, that is, two natural choices of the eigenvalue assigned to $\tilde H_{\lambda,\mu}(x,y;q,t)$.  First, we can let the 
eigenvalue be
\begin{equation}
\frac{K_{\emptyset,(n)\; \la,\mu}(q,t^{-1})}{K_{(1^n),\emptyset \; \la,\mu}(q,t^{-1})}  = \frac{  q^{|\la|+n(\la)}  t^{-n(\la)-n(\mu)}}{ q^{n(\la')+n(\mu')} t^{-|\mu|-n(\mu')} } 
=q^{|\la|+n(\la)-n(\la')-n(\mu')} t^{|\mu|+n(\mu')-n(\la)-n(\mu)} \, .
\end{equation}
Note that the two pairs of partitions,  $(\em,(n))$ and $((1^n),\em)$, are conjugate of each other.
Hence the first Nabla operator $\nabla^B$ is defined as 
\beq \label{actnabla}
\nabla^B \tilde H_{\la,\mu}= q^{|\lambda|+\hat{n}(\mu',\la')} t^{|\mu|+\hat{n}(\la,\mu)} \tilde H_{\la,\mu} \, ,
\end{equation}
where 
\beq
\hat{n}(\la,\mu) =  n(\mu')- n(\la) - n(\mu).
\eeq
Alternatively, we can also define the eigenvalue of the Nabla operator to be the ratio:
\begin{equation}
\frac{K_{\emptyset,(1^n)\; \la,\mu}(q,t^{-1})}{K_{(n),\emptyset \; \la,\mu}(q,t^{-1})}  = \frac{  q^{|\la|+n(\la')+n(\mu')}  t^{-n(\mu')}}{ q^{n(\la)} t^{-|\mu|-n(\lambda)-n(\mu)} } 
=q^{|\lambda|-n(\lambda)+n(\la')+n(\mu')  } t^{|\mu|-n(\mu')+n(\la)+n(\mu)} \, ,
\end{equation}
which leads to 
\beq \label{actnabla2}
\bar \nabla^B \tilde H_{\la,\mu}= q^{|\lambda|-\hat{n}(\mu',\la')} t^{|\mu|-\hat{n}(\la,\mu)} \tilde H_{\la,\mu} \, .
\end{equation}
The operator $\bar \nabla^B$ does not appear to have generic noteworthy properties. 
In particular, it is not always Schur positive (up to an overall sign) when acting on the
Schur functions $s_{\lambda,\mu}$ (not even when acting on $s_{\emptyset;1^n}$).

 From these two commuting operators, one can consider their product $\nabla^B \bar \nabla^B$, whose eigenvalue is readily evaluated from \eqref{actnabla} and \eqref{actnabla2}  to be $q^{2|\lambda|} t^{2|\mu|} $.  The factor 2 in both exponents further indicates that the square root of $\nabla^B \bar \nabla^B$ is well-defined and thus more fundamental.
This operator, 
defined as
\beq \label{actnablaprod}
\sqrt{\nabla^B \bar \nabla^B} \,
\tilde H_{\la,\mu}= q^{|\lambda|} t^{|\mu|} \tilde H_{\la,\mu} \, ,
\eeq
actually appears  to be more interesting.  Indeed, its action on a $B$-type Schur  displays a Schur-positive expansion (up to an overall sign).
\begin{conjecture}  The Schur-expansion coefficients of 
$\left(\sqrt{\nabla^B \bar \nabla^B} \right)^i s_{\lambda,\mu}$ belong
(up to an overall sign)
to $\mathbb N[q^{\pm 1}, t^{\pm 1}]$ for every integer $i$ and every pair of partitions $\lambda,\mu$.
\end{conjecture} 
In the remainder of the section, we will  be solely concerned with 
the operator $\nabla^B$ which connects with $(2n+1)^n$.

\subsection{The explicit action of $\nabla^B$ on $s_{\emptyset,(n)}$.}
The analog of $\nabla e_n\equiv \nabla s_{(1^n)}$ in the $B_n$ case appears 
to be $\nabla^B s_{\emptyset,(n)}$, with $\nabla^B$ defined by \eqref{actnabla}.  The next proposition gives a simple 
expression for $\nabla^B s_{\emptyset,(n)}$.  
\begin{proposition} \label{propaction} We
have
\beq
\nabla^B s_{\emptyset,(n)} = \frac{1}{(qt)^{\binom{n}{2}}}
h_{n}\Bigl[ [n]_{q,t}X+ [n+1]_{q,t}Y \Bigr] \, ,
\eeq
where 
\beq [n]_{q,t}=\frac{q^n-t^n}{q-t}.
\eeq
\end{proposition}
\begin{proof}
We recall the simple identity (see \cite{Mac} and
e.g. \cite[eq. (1.61)]{Hag})
\beq \label{simpleiden}
h_n[Z+W] = \sum_{\ell=0}^n h_{n-\ell}[Z] \, h_\ell[W] =  \sum_{\ell=0}^n (-1)^{n-\ell}  e_{n-\ell}[-Z] \, h_\ell[W]\,, 
\eeq
where the second equality follows by the duality relation:
with $\omega$ such that
\beq
\omega p_r[X]= (-1)^{r-1} p_r[X]=(-1)^r p_r[-X]\,,
\eeq
it follows that\beq
 p_\lambda[-X] = (-1)^{|\lambda|} \omega p_\lambda[X] \implies 
h_{\lambda}[-X] = (-1)^{|\lambda|} \omega h_{\lambda}[X] =  
(-1)^{|\lambda|} e_{\lambda}[X]\,.
\eeq
Using
 $h_\ell[tZ]=t^{\ell} h_{\ell}[Z]$ (and similar results when $q$ is replaced by $t$ or $h_{\ell}$ is replaced by $e_{\ell}$), and suitable choices for $Z$ and $W$, we can write
\beq
 \sum_{\ell=0}^n \left(-\frac{1}{t} \right)^{n-\ell} 
e_{n-\ell}\left[ \frac{X+qY}{1-q/t} \right] h_{\ell}\left[ \frac{X/t+Y}{1-q/t} \right] =
h_n \left[ \frac{-X/t-qY/t}{1-q/t} + \frac{X/t+Y}{1-q/t} \right]= h_n[Y] \,.
\eeq 
Since $s_{\emptyset,(n)}(x,y)=
s_{(n)}(y)=
h_n[Y] $, we have thus obtained the following expansion
\beq \label{acts}
s_{\emptyset,(n)}=
  \sum_{\ell=0}^n \left(-\frac{1}{t} \right)^{n-\ell} 
e_{n-\ell}\left[ \frac{X+qY}{1-q/t} \right] h_{\ell}\left[ \frac{X/t+Y}{1-q/t} \right]
\eeq
To obtain the action of $\nabla^B$ on $s_{\em,(n)}$, it thus suffices to find its action on each term of the sum on the rhs.
Recall the following expressions (see \cite{Mac} eq. VI.4.8 and 4.9)
\beq
P_{(1^k)}^{(q,t)}[Z] = e_k[Z] \qquad {\rm and} \qquad
P_{(k)}^{(q,t)}[Z] \propto h_k \left[\frac{(1-t)}{(1-q)} Z\right] \, ,
\eeq
where the symbol $\propto$ means that the result holds up to a multiplicative
constant.   Therefore
\beq
 H_{(1^k)}^{(q,q/t)}\left[ {X+qY} \right] \propto
 e_k\left[ \frac{X+qY}{1-q/t} \right] \qquad {\rm and} \qquad
H_{(k)}^{(q/t,1/t)}\left[ {X/t+Y} \right] \propto
 h_k\left[ \frac{X/t+Y}{1-q/t} \right]  \, ,
\eeq
which implies, using \eqref{HH} and \eqref{defHwigle},
that
\beq \label{Hhooks}
\tilde H_{(1^{n-\ell}),(\ell)}(x,y;q,t)  \propto
e_{n-\ell}\left[ \frac{X+qY}{1-q/t} \right] h_{\ell}\left[ \frac{X/t+Y}{1-q/t} \right] \, .
\eeq
In other words, the product $e_{n-\ell}\, h_\ell$ appearing in the  decomposition of $s_{\emptyset;(n)}$ in \eqref{acts} is proportional to $\tilde H_{(1^{n-\ell}),(\ell)}$, which is itself an eigenfunction of $\nabla^B$. Its eigenvalue is read off
 \eqref{actnabla} specialized to the case $\la=(1^{n-\ell})$ and $\mu=(\ell)$:
\beq
\nabla^B \tilde H_{(1^{n-\ell}),(\ell)}= \left(\frac{q}{t}\right)^{n(n+1)/2-n\ell} t^n \tilde H_{(1^{n-\ell}),(\ell)} \, .
\eeq
After straightforward manipulations, we thus get 
\begin{equation}
\begin{split}
\nabla^B s_{\emptyset;(n)}& =
 \left(\frac{q}{t}\right)^{n(n+1)/2} \sum_{\ell=0}^n \left(\frac{t^{n+1}}{q^n} \right)^{\ell}(-1)^{n-\ell}  
e_{n-\ell}\left[ \frac{X+qY}{1-q/t} \right] h_{\ell}\left[ \frac{X/t+Y}{1-q/t} \right]\\
& = \left(\frac{q}{t}\right)^{n(n+1)/2} \sum_{\ell=0}^n   
h_{n-\ell}\left[ \frac{-X-qY}{1-q/t} \right] h_{\ell}\left[ \frac{t^nX/q^n+t^{n+1}Y/q^n}{1-q/t} \right] \\
& = \left(\frac{1}{qt}\right)^{n(n-1)/2} 
h_{n}\Bigl[ [n]_{q,t}X+ [n+1]_{q,t}Y \Bigr]
\end{split}
\end{equation}
\end{proof}
In general the action of $\nabla^B$  on Schur functions, contrary to 
that of $\nabla$,
is not positive (up to an overall sign) when expanded in the Schur basis.
But what is remarkable here is that the Schur expansion of
the action of  $\nabla^B$ on $s_{\emptyset;(n)}$ can  be given in a closed form
for $q$ and $t$ generic (while a similar result for $\nabla e_n$
can only be established
when $q=1/t$).   
\begin{corollary} \label{frobnabla}
The action of $\nabla^B$ on $s_{\emptyset;(n)}$ expanded in the Schur 
basis  $s_{\lambda,\mu}=s_{\lambda}(x) s_{\mu}(y)$ is given by
\beq \label{hvss}
\nabla^B s_{\emptyset;(n)}
=\frac{1}{(qt)^{\binom{n}{2}} } \sum_{\lambda,\mu}
s_{\lambda}\Bigl[ [n]_{q,t} \Bigr] 
s_{\mu}\Bigl[ [n+1]_{q,t} \Bigr] s_{\lambda,\mu} \, .
\eeq
In particular, the Schur expansion coefficients belong to 
$\mathbb N[q^{\pm 1},t^{\pm 1}]$. 
\end{corollary}
\begin{proof} Using \eqref{simpleiden} and the Cauchy identity 
(see \cite{Mac} and  e.g. \cite[eq. (1.63)]{Hag})
$h_n\bigl[XY \bigr]=\sum_{\mu \vdash n}
s_{\mu}\bigl[ X \bigr] s_{\mu}\bigl[ Y \bigr]$, we find
\begin{equation}
\begin{split}
h_{n}\Bigl[ [n]_{q,t}X+ [n+1]_{q,t}Y \Bigr] & =
\sum_{\ell=0}^n h_{n-\ell}\Bigl[ [n]_{q,t}X \Bigr] 
h_{\ell}\Bigl[ [n+1]_{q,t}Y \Bigr]\\
 & =\sum_{\ell=0}^n
\sum_{\lambda \vdash n-\ell}
s_{\lambda}\Bigl[ [n]_{q,t} \Bigr] s_{\lambda}\bigl[ X \bigr] 
\sum_{\mu \vdash \ell}
s_{\mu}\Bigl[ [n+1]_{q,t} \Bigr] s_{\mu}\bigl[ Y \bigr]  \\
& = \sum_{\lambda,\mu}
s_{\lambda}\Bigl[ [n]_{q,t} \Bigr] 
s_{\mu}\Bigl[ [n+1]_{q,t} \Bigr] s_{\lambda,\mu} 
\end{split}
\end{equation}
The corollary then follows from Proposition~\ref{propaction}.  The positivity is immediate since $[n]_{q,t} \in \mathbb N[q,t]$ and the Schur functions are monomial positive.
\end{proof}
The $B_n$ analog of the Schur functions are orthonormal 
\beq\label{orthos}
\L s_{\la,\mu},s_{\omega, \nu} \R_B = \delta_{\la \omega} \delta_{\mu \nu} 
\eeq
with respect to the $B_n$ analog of the Hall scalar product defined as\footnote{The scalar product is given in \cite[Appendix B]{Mac}.  To be more specific, it is the specialization of ($5.3'$) therein to the case $B_n \cong C_2 \sim S_n$ (discussed in more details in the example on page 178).  The power of 2 
is due to the fact that the centralizer of each element in $C_2$ is of order 2.} 
\beq
\L p_{\lambda}[X+Y] p_{\mu}[X-Y], p_{\omega}[X+Y] p_{\nu}[X-Y] \R_B =  
\delta_{\lambda \omega} \delta_{\mu \nu} z_{\lambda} z_{\mu} \, 2^{\ell(\lambda)+\ell(\mu)}\, \, .
\eeq
We now extract from the previous Corollary a closed-form expression for the $B_n$ analog of the $q,t$-Catalan 
$\L \nabla^B e_n, e_n \R$ (compare for instance with \eqref{catalant}
in the case $q=1/t$).  As expected,
it reduces to $\binom{2n}{n}$ when
$q=t=1$.

\begin{corollary}
\beq \label{qtcata}
\L \nabla^B s_{\emptyset;(n)} , s_{\emptyset;(n)} \R_B =  
\frac{1}{(qt)^{\binom{n}2}}\left[ 
\begin{array}{c}
2n\\
n
\end{array}
\right]_{q,t}  
  \eeq 
where
\beq
\left[ 
\begin{array}{c}
2n\\
n
\end{array}
\right]_{q,t} = \frac{[2n]_{q,t}!}{[n]_{q,t}!  \,[n]_{q,t}!} \, \qquad\text{and}\qquad [n]_{q,t}!=[n]_{q,t}\cdots [2]_{q,t}\,[1]_{q,t}\,.
\eeq
In particular
\beq
\L \nabla^B_{q=1/t} \, s_{\emptyset;(n)} , s_{\emptyset;(n)} \R_B =  
\frac{1}{t^{n^2}}
\left[ 
\begin{array}{c}
2n\\
n
\end{array}
\right]_{t^2}  \, .
  \eeq 
\end{corollary}
\begin{proof}  Given the orthonormality  
\eqref{orthos},   to evaluate $\L \nabla^B s_{\emptyset;(n)} , s_{\emptyset;(n)} \R_B $ it suffices to compute the coefficient of $s_{\emptyset,(n)}$
in $\nabla^B s_{\emptyset,(n)}$.  This coefficient corresponds to specifying
$\lambda=\emptyset$ and $\mu=(n)$ in the rhs of \eqref{hvss}:
\beq\L \nabla^B s_{\emptyset;(n)} , s_{\emptyset;(n)} \R_B=
\frac{1}{(qt)^{\binom{n}{2}} }
\, h_{n}\bigl[ [n+1]_{q,t}\bigr] =\frac{1}{(qt)^{\binom{n}{2}} }
\left[ 
\begin{array}{c}
2n\\
n
\end{array}
\right]_{q,t} \,.
\eeq
The previous equality follows from
$[n+1]_{q,t}=q^{n}[n+1]_{t/q}$, the relation
\beq \left[ 
\begin{array}{c}
2n\\
n
\end{array}
\right]_{q,t}= q^{n^2}\left[ 
\begin{array}{c}
2n\\
n
\end{array}
\right]_{t/q}
\eeq
and the identity \cite{Mac}
\beq
h_{n}\bigl[ [n+1]_{t}\bigr] = 
\left[ 
\begin{array}{c}
2n\\
n
\end{array}
\right]_{t}
\,. \eeq

\end{proof}

\begin{corollary} \label{coro1}
We have
\beq
\L \nabla^B s_{\emptyset;(n)} , s_{(1^n);\emptyset} \R_B =  1
\eeq
\end{corollary}
\begin{proof} As in the proof of the previous corollary, 
it suffices to compute the coefficient of $s_{(1^n),\emptyset}$
in $\nabla^B s_{\emptyset,(n)}$, which corresponds to specifying
$\lambda=(1^n)$ and $\mu=\emptyset$ in the rhs of \eqref{hvss}:
\beq\L \nabla^B s_{\emptyset;(n)} , s_{(1^n), \emptyset} \R_B=
\frac{1}{(qt)^{\binom{n}{2}} } s_{1^n}\Bigl[ [n]_{q,t} \Bigr] = 
\frac{1}{(qt)^{\binom{n}{2}} } (q^{n-1}) (q^{n-2}t) \cdots  (qt^{n-2})(t^{n-1})=1 \, .
\eeq
\end{proof}

The final result of this section is an explicit expression for 
the analog of $\L \nabla e_n, p_1^n \R$ given by 
$\L \nabla^B s_{\em,(n)} , p_{\em,(1^n)} \R_B$ (where we recall that 
$p_{\em,(1^n)}=p_{1^n}[X+Y]$).  Observe the similarity with 
\eqref{tfactor} when $q=1/t$.
Quite unexpectedly, the present expression, a priori more complicated, factorizes
for $q$ and $t$ generic (recall the discussion surrounding eq. \eqref{nopower} pertaining to the non-factorization in the usual case). 
\begin{proposition}\label{conqt}
We have
\beq
\L \nabla^B s_{\em,(n)} , p_{\em,(1^n)} \R_B= \left( \frac{ [n+1]_{q,t} + [n]_{q,t} }{(qt)^{(n-1)/2}}\right)^n.  
\eeq
In particular, there follows the two specializations:
\begin{equation}\label{tdimB}
\L \nabla^B_{q=1/t} \, s_{\em,(n)} , p_{\em,(1^n)} \R_B= \left(\frac{[2n+1]_t}{t^{n}}\right)^n
\end{equation}
and
\begin{equation}\label{dimB}
\L \nabla^B_{q=t=1} \, s_{\em,(n)} , p_{\em,(1^n)} \R_B=(2n+1)^n.
\end{equation}
\end{proposition}
\begin{proof}  We first need to obtain the coefficient
$p_{1^n}[X+Y]$ in the expansion of
$\nabla^B s_{\emptyset, (n)}$ in the 
$p_{\lambda}[X-Y] p_{\mu}[X+Y]$ basis.
We have 
\beq
[n]_{q,t}X+ [n+1]_{q,t}Y = \frac{1}{2}\bigl(  [n]_{q,t}+ [n+1]_{q,t} \bigr) (X+Y)
+ \frac{1}{2}\bigl(  [n]_{q,t}- [n+1]_{q,t} \bigr) (X-Y)
\eeq
and thus this amounts, from Proposition~\ref{propaction},  to computing
the coefficient of
$p_{1^n}[X+Y]$ in
\beq \begin{split}
& \frac{1}{(qt)^{\binom{n}{2}} }
h_{n}\left[  \frac{1}{2}\bigl(  [n]_{q,t}+ [n+1]_{q,t} \bigr) (X+Y)
+ \frac{1}{2}\bigl(  [n]_{q,t}- [n+1]_{q,t} \bigr) (X-Y) \right] \\
& \quad = \frac{1}{(qt)^{\binom{n}{2}} }  \sum_{\lambda, \mu} \frac{1}{z_\la z_\mu} p_\la \left[  \frac{1}{2}\bigl(  [n]_{q,t}+ [n+1]_{q,t} \bigr)\right]  p_\mu \left[\frac{1}{2}\bigl(  [n]_{q,t}- [n+1]_{q,t} \bigr)   \right]
 p_\la [X+Y] p_\mu[X-Y]
\end{split}
\eeq
where we used \eqref{simpleiden}, the expansion $h_n = \sum_{\la \vdash n}p_\la /z_\la$ (\cite[eq. I.2.14']{Mac}), and the basic property $p_\la[XY]=p_\la[X]p_\la[Y]$.  Letting $\la=(1^n)$ and $\mu=\emptyset$ the coefficient of 
$p_{1^n}[X+Y]$ is easily found to be
\beq
  \frac{1}{(qt)^{\binom{n}{2}}z_{1^n}} p_{1^n}
\left[  \frac{1}{2}\bigl(  [n]_{q,t}+ [n+1]_{q,t} \bigr) \right]
=  \frac{1}{(qt)^{\binom{n}{2}} n!} \left( p_1 \left[ \frac{1}{2}\bigl(  [n]_{q,t}+ [n+1]_{q,t} \bigr) \right] \right)^n
=  \frac{1}{2^n n!}
\left( \frac{[n]_{q,t}+ [n+1]_{q,t}}{(qt)^{(n-1)/2}} \right)^n  .
\eeq 
The proposition then follows from
\beq 
\langle  p_{1^n}[X+Y] , p_{1^n}[X+Y] \rangle_B = 2^n n! \, .
\eeq
\end{proof}

\begin{remark} \label{remfrob}
As mentioned in the introduction,  $\nabla_{q=t^{-1}}^B s_{\emptyset;(n)}$ seems to coincide with  ${\rm Frob}_{t^{-1},t}(R^{B_n})$, where $R^{B_n}$ is a certain doubly graded quotient  of the coinvariant ring $C^{B_n}$.  From Remark~\ref{rem1}, one would be tempted to believe that there
exists for every pair of partitions $(\lambda,\mu)$ a bigraded submodule 
of $R^{B_n}$ (isomorphic to the regular module of the  hyperoctahedral group $B_n$) whose
Frobenius series when $q=t^{-1}$ corresponds to the Schur expansion of
$\tilde H_{\lambda, \mu}(x,y;t^{-1},t)$.   This cannot be the case however since 
$\nabla^B s_{\emptyset;(n)} -\tilde H_{\lambda,\mu}(x,y;q,t)$ is not Schur-positive in general (even when $q=t^{-1}$).
\end{remark}

\begin{remark}  The Macdonald polynomials $\tilde H_{\lambda}^{(q,t)}[X]$
can be defined (up to normalization constants) as the unique basis such that
\begin{align}
 {\rm (i)} & \quad \tilde H_{\lambda}^{(q,t)}[X(1-t)] \in \mathbb Q(q,t) \{s_{\mu} : \mu \leq \lambda\} \\
 {\rm (ii)} & \quad  \tilde H_{\lambda}^{(q,t)}[X(1-q)] \in \mathbb Q(q,t) \{s_{\mu} : \mu \geq \lambda\} 
\end{align}

The double Macdonald polynomials can also be defined by two similar triangularities.  Let $\phi_t$ be the plethystic substitution  
\beq
X \mapsto X \qquad {\rm and}  \qquad (X+Y) \mapsto (X+Y)(1-t)
\eeq
and recall that $\omega_B$ is the involution such that $\omega_B
s_{\lambda,\mu}(x,y)=s_{\mu',\lambda'}(x,y)$.  The  double Macdonald polynomials $\tilde H_{\lambda, \mu}(x,y;q,t)$ can be characterized (up to normalization constants) as the unique basis of the space of bisymmetric functions such that
\begin{align}
 {\rm (i)} & \quad \phi_t \tilde H_{\lambda,\mu}(x,y;q,t) \in \mathbb Q(q,t) 
\{s_{\omega,\eta} : (\omega,\eta) \leq (\lambda,\mu)\} \\
 {\rm (ii)}  & \quad \omega_B  \circ \phi_q \circ \omega_B
 \tilde H_{\lambda,\mu}(x,y;q,t) \in \mathbb Q(q,t) 
\{s_{\omega,\eta} : (\omega,\eta) \geq (\lambda,\mu)\} 
\end{align}
The first triangularity is seen as follows: $\phi_t \tilde H_{\lambda, \mu}(x,y;q,t)$ is equal up to a constant to $P_{\lambda, \mu}(x,y;q,t^{-1})$, which is lower triangular in the $m_{\lambda,\mu}$ basis
and hence in the $s_{\lambda,\mu}$ basis
(the $s_{\lambda,\mu}$ basis is lower triangular in the  $m_{\lambda,\mu}$ basis
since it corresponds to the specialization $q=t=0$ of $P_{\lambda,\mu}(x,y;q,t)$).
Using the symmetry \eqref{omegaBrel}, we have immediately that $ \tilde H_{\lambda, \mu}(x,y;q,t)$ is upper triangular in the $s_{\lambda,\mu}$ basis.  
The two triangularities in the double case are not as symmetric as in the usual case due to the noncommutativity of $\omega_B$ 
and the plethystic substitution $\phi_q$ (whereas in the usual case 
the involution $\omega$ commutes with any
plethystic substitution). 

Haiman also introduced in \cite{Haiman} wreath Macdonald polynomials.  In the special
case of the complex reflection group $G(2,1,n)$, the wreath Macdonald 
polynomial $\mathcal H_{\mu}(q,t)$ depends on a choice of staircase partition $\delta^{m}$ (a 2-core) and is indexed by a partition $\mu$ of size
$m(m-1)/2+n$, where $n$ is a fixed integer.  They satisfy the triangularities
\beq \label{core1}
\mathcal H_{\mu}(q,t) \otimes \sum_i (-q)^i {\rm char}(\wedge^i \mathfrak h ) \in \mathbb Q(q,t)\bigl \{ \chi^{{\rm Quot}_2(\lambda)}: \lambda \geq \mu , 
{\rm Core}_2(\lambda)=\delta^{m} \bigr\}   
\eeq
and
\beq \label{core2}
\mathcal H_{\mu}(q,t) \otimes \sum_i (-t)^i {\rm char}(\wedge^i \mathfrak h ) \in \mathbb Q(q,t)\bigl \{ \chi^{{\rm Quot}_2(\lambda)}: \lambda \leq \mu , 
{\rm Core}_2(\lambda)=\delta^{m} \bigr\}   
\eeq
where $ \mathfrak h=\mathbb C^n$ is the defining representation of $G(2,1,n)$, 
${\rm Quot}_2(\lambda)$ is the 2-quotient of $\lambda$ (a pair of partitions of total degree $n$) and ${\rm Core}_2(\lambda)$ is the 2-core of $\lambda$.
This is somewhat reminiscent of our construction, 
even more so that if $m$ is large enough Haiman claims that 
the corresponding wreath Macdonald polynomials can be given in terms of 
usual Macdonald polynomials.   However, the objects cannot 
coincide given that the symmetry in the triangularities \eqref{core1} and \eqref{core2} is not present in the double case.  
\end{remark}

\section{Conclusion: the superspace bridge}
\label{Con}

The present construction relies on the formalism previously developed for Macdonald superpolynomials \cite{BDLM1,BDLM2}.   To these superpolynomials, we have associated bisymmetric polynomials by considering the coefficient of a given monomial in the $\ta_i$'s, say $\ta_1\cdots \ta_m$  for a superpolynomial of fermionic-degree $m$, and dividing by the Vandermonde determinant in the commuting variables  $x_1,\cdots, x_m$.   We have seen that for sufficient high fermionic degree, a stable  sector is reached.  In this stable sector, our key result is the product form \eqref{facto}. As a consequence of this remarkable factorization, we have been able to  prove rather directly a number of properties for these so-called double Macdonald polynomials. In particular, we have obtained:

\begin{itemize}
\item The expression for the norm.

\item The expression for the integral form.

\item The evaluation.

\item The positivity and integrability of the double Kostka coefficients.

\item The two Macdonald-type symmetry properties of the double 
Kostka coefficients.

\end{itemize}

In \cite{BDLM1,BDLM2}, we  have presented conjectures related to  the above 
five items but pertaining to generic Macdonald superpolynomials. The present results imply that we now have proofs of these results for 
all cases where $m\geq n$
(in some sense for roughly half the cases).
The precise connection between the  present results and our more general conjectures is worked out in Appendix C.

We had a number of mathematical and physical motivations for undertaking the study of the Macdonald superpolynomials.
One of which, of a combinatorial nature, was to see whether by adding more structure to the usual Macdonald polynomials, one could get unexpected new handles on open problems such as a  combinatorial description of the $q,t$-Kostka coefficients, generalizing the Lascoux-Sch\"utzenberger description of the Kostka-Foulkes coefficients. For the stable sector considered here, this particular hope was not fulfilled: our new double Kostka coefficients are roughly sums of products of the usual ones.

However, we have already presented a conjectural result -- pertaining to  the non-stable sector, which thereby necessarily relies on the superspace formalism -- that could shed some light on the combinatorics of the usual 
$q,t$-Kostka coefficients  \cite{BDLM2}.
It says that the simplest superpolynomials, namely those in the $m=1$ sector, do provide a refinement of the Kostka coefficients. Precisely, this conjectural result gives a relation between the generalized coefficients $K_{\Omega \Lambda}(q,t)$ of fermionic degree $m=1$ and total degree $n$ and the usual $q,t$-Kostka coefficients of degree $n+1$.
\begin{conjecture}\label{kos1}
Let $\Lambda$ be a superpartition of fermionic degree $m=1$, and let
$H_{\Lambda}$ be the modified
Macdonald superpolynomial (see Appendix~\ref{app_proof_conj} for more details).  Let also $\psi$
be the linear application that maps $s_{\Omega}$ to $s_{\Omega^{\circledast}}$.
Then
\begin{equation} 
\psi(H_{\Lambda})= H_{\Lambda^{\circledast}}.
\end{equation} 
\end{conjecture}

This conjecture implies that the usual $q,t$-Kostka coefficient 
$K_{\mu \la}(q,t)$ can be calculated 
from its lower-degree super-relatives as
\begin{equation} \label{relatekostka}
K_{\mu \la}(q,t) = \sum_{\Omega \, | \, \Om^\cd = \mu} K_{\Omega \Lambda }(q,t)
\end{equation} 
where $\Lambda$ is any superpartition that can be obtained from
$\lambda$ by replacing a square by a circle, and the 
sum is over all $\Omega$'s that can be obtained from
$\mu$  by replacing a square by a circle. Moreover, the expression for the sum on the right-hand side is independent of the choice of $\La$.
We thus relate a Kostka coefficient of a given degree to a sum of lower degree Kostka coefficients in the $m=1$ fermionic sector, a process that mimics a sort of transmutation of the fermionic variable into a bosonic one.

For example, consider $H_{(2;1)}(x,\theta;q,t)$.  Its Schur expansion reads (using the diagrammatic representation of superpartitions introduced in Appendix B)
\begin{equation}
H_{  \, _{\Tiny{  {\tableau[scY]{&& \bl \tcercle{}\\ \\ }}}}  } =  t \, s_{ \, _{\Tiny{  {\tableau[scY]{&&& \bl \tcercle{} \\ }}}}  } + q^2 t \, s_{ \, _{\Tiny{  {\tableau[scY]{&& \\ \bl \tcercle{} \\ }}}}  } 
+ (1+qt) \, s_{ \, _{\Tiny{  {\tableau[scY]{&& \bl \tcercle{} \\  \\ }}}}  } 
+ (q+q^2 t) \, s_{ \, _{\Tiny{  {\tableau[scY]{& \\ & \bl \tcercle{}  \\ }}}}  } 
+ (q^2+q^3t) \, s_{ \, _{\Tiny{  {\tableau[scY]{& \\  \\ \bl \tcercle{} }}}}  } 
+ q \, s_{ \, _{\Tiny{  {\tableau[scY]{& \bl \tcercle{} \\  \\ \\ }}}}  } 
+ q^3 \, s_{ \, _{\Tiny{  {\tableau[scY]{ \\  \\ \\   \bl \tcercle{} }}}}  } 
\end{equation}
Now apply $\psi$ :
\begin{align}
\psi \left( H_{  \, _{\Tiny{  {\tableau[scY]{&& \bl \tcercle{}\\ \\ }}}}  } \right) & = H_{  \, _{\Tiny{  {\tableau[scY]{&& \\ \\ }}}}  } \nonumber \\
& =  t \, s_{ \, _{\Tiny{  {\tableau[scY]{&&& \\ }}}}  } + q^2 t \, s_{ \, _{\Tiny{  {\tableau[scY]{&& \\  \\ }}}}  } 
+ (1+qt) \, s_{ \, _{\Tiny{  {\tableau[scY]{&&  \\  \\ }}}}  } 
+ (q+q^2 t) \, s_{ \, _{\Tiny{  {\tableau[scY]{& \\ &   \\ }}}}  } 
+ (q^2+q^3t) \, s_{ \, _{\Tiny{  {\tableau[scY]{& \\  \\  \\ }}}}  } 
+ q \, s_{ \, _{\Tiny{  {\tableau[scY]{&  \\  \\ \\ }}}}  } 
+ q^3 \, s_{ \, _{\Tiny{  {\tableau[scY]{ \\  \\ \\  \\  }}}}  }  \nonumber \\
& =  t \, s_{ \, _{\Tiny{  {\tableau[scY]{&&& \\ }}}}  }
+  (1+qt+q^2 t) \, s_{ \, _{\Tiny{  {\tableau[scY]{&& \\  \\ }}}}  } 
+ (q+q^2 t) \, s_{ \, _{\Tiny{  {\tableau[scY]{& \\ &   \\ }}}}  } 
+ (q + q^2 + q^3t)   \, s_{ \, _{\Tiny{  {\tableau[scY]{& \\  \\  \\ }}}}  } 
+ q^3 \, s_{ \, _{\Tiny{  {\tableau[scY]{ \\  \\ \\  \\  }}}}  } 
\end{align}
which corresponds to the usual Schur expansion of the modified Macdonald polynomials $H_{(3,1)}(x;q,t)$.

  This conjecture can be generalized as follows.  Let $\La$ and $\La^\boxcircle$ be two superpartitions 
  such that $\La^\boxcircle$ is obtained from $\La$ by replacing a circle by a box.   Note that  $\La^\boxcircle$ is not unique for  $m>1$.   
Define $\#(\La,\La^\boxcircle)$ to be the number of circles that lie above the row of $\Lambda$ where the circle has been replaced by a box.
\begin{conjecture}
Let $\widehat H_\La$ be the normalized modified Macdonald superpolynomial defined as
\begin{equation}
\widehat H_\La = v_\La H_\La, \qquad \text{with} \quad v_\La = \prod_{s\in \La/\mathcal B \La} (1-q^{a_{\La^\circledast}(s)} t^{  l_{\La^*}(s)+1 } )
\end{equation}
and let $\psi$ be the linear application that maps:
\begin{equation}
\psi \; : \; s_\La \; \mapsto \; \psi(s_\La) = \sum_{\Omega \,  |\,  \Omega = \La^\boxcircle} (-1)^{\#(\La,\Omega)}s_\Omega.
\end{equation}
Then, we have:
\begin{equation}
\psi(\widehat H_\La) = \sum_{\Omega \,  |\,  \Omega = \La^\boxcircle} (-1)^{\#(\La,\Omega)}\widehat H_\Omega.
\end{equation}
\end{conjecture}

Note the if $\La$ has fermionic degree $m=1$, this reduces to the previous conjecture since then all the factors $v_\bullet$ reduce to 1.
We stress that $ \La/\mathcal B \La$ refers to the set of boxes in $\La$ that are not in $\B\La$, a set denoted by $\F\La$ in Appendix C (see below eq. \eqref{defui}). This set is empty when $m=1$.

\begin{remark} It is straightforward to 
check that the application $\psi$ is such that 
$\psi \circ \psi=0$.  In the language of \cite{DLM0}, $\psi$ corresponds to
$\sum_i{x_i} \partial_{\theta_i}$, the
adjoint of the exterior derivative $\sum_i{\theta_i} \partial_{x_i}$.
\end{remark}

Let us make explict the implication of this result at the level of the Kostkas. We have 
\begin{equation}
\psi(\widehat H_\La) =v_\La \sum_\Omega K_{\Omega \La}(q,t) \psi(s_\Omega)=v_\La \sum_\Omega K_{\Omega \La}(q,t) \sum_{\Delta \,  |\,  \Delta = \Om^\boxcircle} (-1)^{\#(\Om,\Delta)}s_\Delta
\end{equation}
and
\begin{equation}
\psi(\widehat H_\La) =\sum_{\Gamma \,  |\,  \Gamma = \La^\boxcircle} (-1)^{\#(\La,\Gamma)}\widehat H_\Gamma
=\sum_{\Gamma \,  |\,  \Gamma = \La^\boxcircle} (-1)^{\#(\La,\Gamma)}v_\Gamma\sum_\Delta K_{\Delta, \Gamma}\, s_\Delta,
\end{equation}
By comparing the coefficients of $s_\Delta$, we get
\beq \label{Kdiffm}
v_\La \sum_{\Omega \, |\,\Om^\boxcircle= \Delta } (-1)^{\#(\Om,\Delta)} K_{\Omega \La}(q,t) =\sum_{\Gamma \,  |\,  \Gamma = \La^\boxcircle}  (-1)^{\#(\La,\Gamma)}  v_\Gamma K_{\Delta, \Gamma}\,
\eeq
This is thus a relation between a linear combination of Kostkas for fermionic degrees differing by 1. The case $m=1$ is of course special in that the sum on the r.h.s. reduces to a single term.
As already said, when $m>1$, $\La^\boxcircle$ can take $m$ different values, and therefore the r.h.s. sum contains $m$ terms. 
Let us then reverse the point of view and see whether there are situations for which the Kostkas at fermioinc degree $m$ could be computed from those at degree $m-1$. This would be the case if the sum on the l.h.s. could be reduced to a single term. This is actually the case when the diagram of  $\Delta$ has a single removable box. This removable box is thus necessarily the one created  when a circle is changed into a box.

Consider an example. Take $\Delta=(3;2)$, which has one removable box (indicated by $\times$) so that $\Om=(3,1;)$:
\beq \Delta:\quad{\tableau[scY]{&&&\bl\tcercle{}\\&\times\\}}\qquad\qquad \Om:\quad{\tableau[scY]{&&&\bl\tcercle{}\\&\bl\tcercle{}\\}} 
\eeq
The l.h.s. of \eqref{Kdiffm} becomes $v_\La(-1) K_{(3,1;),\La}(q,t)$. Let us set $\La=(2,0;2)$: $\Gamma$ in the sum of the r.h.s. of \eqref{Kdiffm} can take the two values $(1;3,2)$ and $(2;2,1)$:
\beq \La:\quad{\tableau[scY]{&&\bl\tcercle{}\\&\\ \bl\tcercle{}}}\qquad\xrightarrow{\psi}\qquad\La^\boxcircle:\quad
{\tableau[scY]{&&\\&\\ \bl\tcercle{}}}\qquad\text{and}\qquad {\tableau[scY]{&&\bl\tcercle{}\\&\\ \\}}
\eeq 
Hence, the r.h.s. of \eqref{Kdiffm} reduces to $ K_{(3;2)(0;3,2)}(q,t)-
K_{(3;2) (2;2,1)}(q,t)$. We thus end up with the relation:
\beq
-(1-q^2t^2) K_{(3,1;),(2,0;2)}(q,t)= K_{(3;2)(0;3,2)}(q,t)-K_{(3;2) (2;2,1)}(q,t).\eeq
With
\beq K_{(3;2)(0;3,2)}(q,t)=  t +t^2+qt^2+qt^3 +q^2 t^4
 \quad\text{and}\quad K_{(3;2) (2;2,1)}(q,t) =t^2 (1+qt + q^2t +q^2t^2+q^3t^2),
\eeq
we get
\beq -(1-q^2t^2) K_{(3,1;),(2,0;2)}(q,t)=-(1-q^2t^2)(t+qt^2)\qquad \Rw\qquad  K_{(3,1;),(2,0;2)}(q,t)=t+qt^2.\eeq

Another instance where this relation appears to be useful is when $\Delta$ has no removable box. In this case, there are no $\Om$ such that $\Om^\boxcircle=\Delta$, so that  the sum on the l.h.s. of \eqref{Kdiffm} vanishes. The r.h.s. becomes an identity on alternating sums of Kostkas weighted by the factor $v_\Gamma$. Here is an example:
take $\Delta=(3,2;)$, which clearly has no removable box, and $\La=(3,1,0;)$
\beq \La:\quad{\tableau[scY]{&&&\bl\tcercle{}\\&\bl\tcercle{}\\ \bl\tcercle{}}}\qquad\xrightarrow{\psi}\qquad v_\Gamma\;\Gamma:\quad (1-qt)\quad
{\tableau[scY]{&&&\\&\bl\tcercle{}\\ \bl\tcercle{}}},\qquad(1-q^3t^2)\quad {\tableau[scY]{&&&\bl\tcercle{}\\&\\\bl\tcercle{} \\}},
\qquad (1-q^2t)\quad  {\tableau[scY]{&&&\bl\tcercle{}\\&\bl\tcercle{}\\ \\}}.
\eeq
This leads to the relation
\beq
(1-qt)K_{(3,2;),(1,0;4)}(q,t)-(1-q^3t^2)K_{(3,2;),(3,0;2)}(q,t)+
(1-q^2t)K_{(3,2;),(3,1;1)}(q,t)=0.\eeq
With
\beq K_{(3,2;),(1,0;4)}(q,t)=q^2t^2(1+q^2t) \qquad {\rm and}\qquad K_{(3,2;),(3,0;2)}(q,t)=K_{(3,2;),(3,1;1)}(q,t)=t(1+q^2t),
\eeq
this is easily checked to be satisfied.

\begin{appendix}
\section{Monomials and power-sums in superspace in bisymmetric form}
The following proposition corresponds to formulas (128) and (176) of
\cite{IMRN}.  Since it concerns basic results underlying our construction, we nevertheless include its proof for completeness.

\label{mpbi}
\begin{proposition}\label{popo}
The bisymmetric monomial 
deduced from the supermonomial of fermionic degree $m$ is given by 
\begin{equation}\label{mode}
m_{\la,\mu}(x_1, \ldots,x_N)=s_\la(x_1, \ldots, x_m) \,m_\mu(x_{m+1}, \ldots x_N),
\end{equation}
where $s_\la$ and $m_\la$ are respectively the Schur polynomials and the usual monomial functions. 
Similarly, the bisymmetric power-sum reads:
\begin{equation}\label{psde}
p_{\la,\mu}(x_1, \ldots, x_N) = s_\la(x_1, \ldots, x_m) \,  p_\mu (x_1, \ldots, x_N).
\end{equation}
\end{proposition}
%
\n Note that in the first case the two functions on the rhs depend upon distinct set of variables, which is not the case in the expression of the power-sums. 

\begin{proof} 
Let us first establish \eqref{mode}, starting with the  expression for the super-monomial:
\begin{equation}
m_\La = \frac{1}{f_{\La^s}(1)} {\sum_{\sigma \in S_N}}  \mathcal{K}_\sigma \ta_1 \cdots \ta_m x^\La , \qquad \text{where}\qquad f_{\La^s}(1)= \prod_{i\geq 0} n_{\La^s}(i) !
, \end{equation}
with $n_{\La^s}(i)$ being the number of occurrences of $i$ in $\La^s$, and $m$ is the number of fermions.
The permutation $\mathcal{K}_{ij}$  interchanges the pairs $(x_i,\ta_i)$ and $(x_j,\ta_j)$.
Focusing on the term of $m_\La $ that contains $\ta_1 \cdots \ta_m$, we have
\begin{equation}
[\ta_1 \cdots \ta_m ] \; m_\La = \frac{1}{f_{\La^s}(1)} \mathcal A_{1,\ldots, m} \mathcal S_{m+1, \ldots, N} x^\La
\end{equation}
where $\mathcal A$ is the anti-symmetrizer operator (which, here, acts on the variables $x_1 \cdots x_m$) and $\mathcal S$ is the symmetrizer operator (acting on the variables $x_{m+1} \cdots, x_N$).  The coefficient $f_{\La^s}(1)$ ensures that repeated terms count for $1$ in the expression.
Since $\mathcal A$ and $\mathcal S$ are independent operators (acting on different sets of variables), we can write:
\begin{equation}
[\ta_1 \cdots \ta_m ] \; m_\La = \frac{1}{f_{\La^s}(1)} \left(\sum_{w \in S_m} \varepsilon(w) K_w ( x_1^{\La^a_1} \cdots x_m^{\La^a_m} ) \right) \left(  \sum_{\sigma \in S_{N-m}} K_{\sigma}(x_{m+1}^{\La^s_1} \cdots x_{N}^{\La^s_{N-m}} )  \right)  
\end{equation}
where $\varepsilon(w)$ is  the sign of  permutation $w$ and $K_{ij}$ interchanges $x_i$ and $x_j$.  By dividing this expression by the Vandermonde determinant in the variables $x_1 \cdots x_m$, denoted ${\Delta_m}$,  and using the  decomposition 
 $\La=(\La^a ; \La^s) = (\la + \delta^m ; \mu)$, we obtain:
\begin{equation}
[\ta_1 \cdots \ta_m ] \; \frac{1}{\Delta_m} m_\La = \left( \frac{ \sum_{w \in S_m} \varepsilon(w) K_w ( \bar{x}^{\delta^m + \la} )     }{\prod_{1 \leq i<j\leq m}  (x_i-x_j) } \right) \;  \left(   \frac{1}{f_{\mu}(1)}     \sum_{\sigma \in S_{N-m}} K_{\sigma}(x_{m+1}^{\mu_1} \cdots x_{N}^{\mu_{N-m}} )   \right) =: m_{\la,\mu}
\end{equation}
where $\bar x$ denote the variables $x_1, \cdots , x_m$.  
The term in the first parenthesis is nothing but the definition of a Schur function in the variable $x_1 \cdots x_m$, that is $s_\la(x_1, \ldots, x_m)$.  The second term is simply the monomial function in the other variables over the partition $\mu$, that is, $m_\mu(x_{m+1},\ldots,x_N)$.  We have thus recovered \eqref{mode}. 

For the derivation of \eqref{psde}, we proceed in a similar way.
     The complete power-sum superfunction labelled by $\La$ reads
\begin{equation}
p_\La= \tilde{p}_{\La^a_1} \cdots \tilde{p}_{\La^a_m} p_{\La^s_1} \cdots p_{\La^s_{N-m}}
\end{equation}
where $\tilde p_k = \sum_{i=1}^N\ta_i x_i^k$ and $p_k = \sum_{i=1}^Nx_i^k$.  Taking the coefficient in $\ta_1 \cdots \ta_m$ we have :
\begin{equation}
[\ta_1 \cdots \ta_m] \; p_\La = \left( \mathcal A_{1,\ldots,m} \bar x^{\La^a} \right) \, p_{\La^s}(x_1, \ldots, x_N).
\end{equation}
Divide this by the Vandermonde in the $m$ first variables and we finally obtain
\begin{equation}
[\ta_1 \cdots \ta_m] \; \frac{1}{\Delta_m} p_\La = \left( \frac{ \sum_{w \in S_m} \varepsilon(w) K_w ( \bar{x}^{\delta^m + \la} )     }{\prod_{1 \leq i<j\leq m}  (x_i-x_j) } \right) p_\mu(x_1, \ldots, x_N) =: p_{\la, \mu}
\end{equation}
where the term in parenthesis is the Schur functions $s_\la(x_1, \ldots, x_m)$.  This gives \eqref{psde}.
\end{proof}

\section{Induced properties on pairs of partitions: conjugation and dominance order}
\label{proof_order}
\sas

As mentioned in the introduction, a superpartition is a pair of partitions of the form $\La=(\La^a;\La^s)$ where the partition $\Lambda^a$ has $m$ distinct
parts (the $m$-th one can be equal to 0).   We define
\beq
 \Lambda^*=\Lambda^a \cup \Lambda^s \quad  {\rm and} \quad
\Lambda^\circledast=
 (\Lambda^a+1^{m}) \cup \Lambda^s \, .
\eeq
It is manifest that the pair of partitions $(\La^\cd,\La^*)$ fixes  the superpartition $\La$. A diagrammatic representation of $\La$ is obtained from
the Ferrers diagram of $\Lambda^\circledast$ by changing into circles the 
cells of $\Lambda^\circledast/\Lambda^*$.   For instance, if $\La
=(3,1,0;2,1)$, we have
\begin{equation} \label{diag}
     \La^\cd:\quad{\tableau[scY]{&&&\\&\\&\\\\ \\ }} \qquad
    \qquad     \La^*:\quad{\tableau[scY]{&&\\&\\ \\ \\ }} \;, \eeq
which gives
\beq \Lambda: \quad{\tableau[scY]{&&&\bl\tcercle{}\\&\\&\bl\tcercle{}\\ \\
    \bl\tcercle{}}}.\eeq

Recall that we can associate to a superpartition $\Lambda$ of fermionic degree $m$ a pair of partitions
$\lambda$ and $\mu$ in the following way:
 \begin{equation} \label{bijec2}
\La =(\La^a ; \La^s) \lrw (\la, \mu )= (\La^a - \delta^m, \La^s) .
\end{equation}
When $m \geq n$, this establishes an obvious bijection between superpartitions
$(\La^a ; \La^s)$ of fermionic degree $m$ 
such that $|\La^a|+|\La^s|=n+m(m-1)/2$ and pairs of partitions 
$(\la, \mu)$ such that $|\lambda|+|\mu|=n$.
Before describing how the conjugation and the dominance-ordering properties are
induced from superpartitions to pairs of partitions,
we  establish some elementary results that will be used for this analysis. 
In that regard, it is convenient to  first introduce a convention concerning the positions of the boxes corresponding to $\lambda$
within the diagrammatic representation
of the superpartition $\La$.
We choose to place the boxes of $\la$ (the boxes marked with a $\circ$ in the example below) in {\it columns} that are not fermionic (that is, that do not end with a circle).
For instance, consider $\la=(2,1)$, $\mu=(3,1)$ and the corresponding $\Lambda$ for
$m=3,4,5$ (all in the non-stable sector, illustrating the fact that the two partitions get disentangled before $m\geq n=7$):
\beq
\label{ex1} {\tableau[scY]{\circ& \circ\\ \circ}} \quad
 {\tableau[scY]{\bullet&\bullet& \bullet\\ \bullet\\ }} \quad
\stackrel{3}
{\longleftrightarrow}\quad
  {\tableau[scY]{&\circ&&\circ&\bl\tcercle{}\\\bullet&\bullet& \bullet \\&\circ&\bl\tcercle{} \\ \bullet\\ \bl\tcercle{} \\ }}
  \quad\text{or}\quad\stackrel{4}{\longleftrightarrow}\quad
  {\tableau[scY]{&&\circ&&\circ&\bl\tcercle{}\\&&\circ&\bl\tcercle{}\\\bullet&\bullet& \bullet \\ &\bl\tcercle{} \\ \bullet\\ \bl\tcercle{} \\ }}
  \quad\text{or}\quad \stackrel{5}{\longleftrightarrow} \quad
  {\tableau[scY]{&&&\circ&&\circ&\bl\tcercle{}\\&&&\circ&\bl\tcercle{}\\\bullet&\bullet& \bullet \\ &&\bl\tcercle{}\\&\bl\tcercle{} \\ \bullet\\ \bl\tcercle{} \\ }}.
\eeq

\begin{lemma}\label{lem2}
Let  $\la,\mu,\La$ and $m$ be defined as in \eqref{bijec2} and suppose that
$m\geq n=|\la|+|\mu|$.  If $\ell(\la)=\ell$ then
the $\ell$-th entry of $\La^a$ is strictly larger than the first entry of 
$\La^s$, that is,
\beq\La^a_{\ell}>\La^s_1=\mu_1 \, .
\eeq
In particular, in the diagram of the superpartition $\La$, the cells marked
with a $\circ$ appear
strictly above those marked with a $\bullet$.
\end{lemma}
\begin{proof} The condition $m\geq n$ implies 
\beq\label{boundm}
 m 
\geq \la_1+\cdots +\la_\ell+\mu_1\geq \ell+\mu_1>\ell-\la_\ell+\mu_1 \, .
\eeq
Hence
\beq \La^a_{\ell}=\la_\ell +m-\ell> \mu_1 = \Lambda_1^s.
\eeq
\end{proof}
The bound $m\geq n$ is not the optimal one ensuring the separation of $\lambda$
and $\mu$ in $\Lambda$ but it is sufficient for  our purpose. 

\begin{lemma}\label{sepap}
In the diagram of the superpartition $\La$, when $m\geq n$, 
the cells marked
with a $\circ$ appear strictly to the right of
those marked with a $\bullet$.
\end{lemma}
 \begin{proof}
By inspection, we see that the column in $\La$ where the $\circ$ corresponding to the first column of $\la$ are inserted is the $m-\la_1'+1$ one. The lemma 
will then hold if
$m-\la_1'+1>\mu_1$.  Given $\la_1'=\ell(\lambda)=\ell$, 
this corresponds to $m\geq \ell+\mu_1$, which is a consequence of $m\geq n$ as seen in \eqref{boundm}.
\end{proof}
 We now come to the conjugation property. 
\begin{lemma} \label{conju}
Suppose that as in \eqref{bijec2}, we have
$\Lambda \leftrightarrow (\lambda,\mu)$ with $m\geq n$. Then
\begin{equation}
\Lambda' \leftrightarrow  (\mu', \la') \, .
\end{equation}
\end{lemma}
\begin{proof}
This follows directly from the  definition of the  conjugation for superpartitions, obtained by the interchange of rows and columns.
\end{proof}
An example will make this completely obvious: consider the pair $(2),\,(3,1)$, so that, with $m=6$:
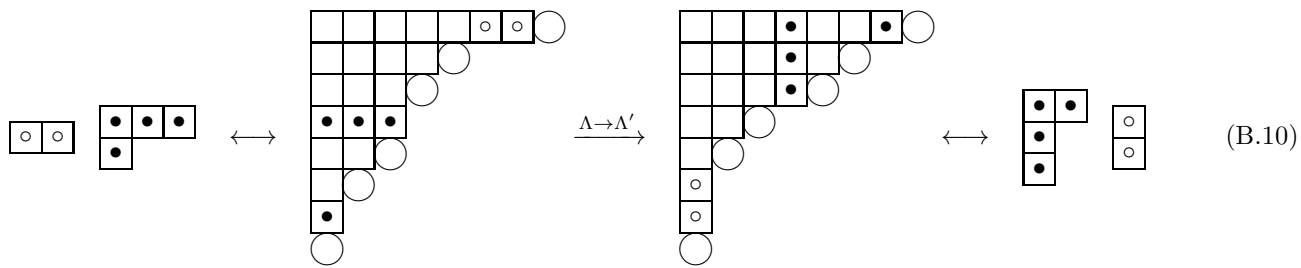
\beq\label{ex2}
 {\tableau[scY]{\circ& \circ}} \quad
 {\tableau[scY]{\bullet&\bullet& \bullet\\ \bullet\\ }} \quad
\longleftrightarrow\quad
  {\tableau[scY]{&&&&&\circ&\circ&\bl\tcercle{}\\&&&&\bl\tcercle{}\\&&&\bl\tcercle{}\\\bullet&\bullet& \bullet \\&&\bl\tcercle{}\\&\bl\tcercle{} \\ \bullet\\ \bl\tcercle{} \\ }}\xrightarrow{\La\rw\La'}\quad
    {\tableau[scY]{&&&\bullet&&&\bullet&\bl\tcercle{}\\&&&\bullet&&\bl\tcercle{}\\&&&\bullet&\bl\tcercle{}
    \\&&\bl\tcercle{}\\&\bl\tcercle{} \\ \circ\\ \circ\\ \bl\tcercle{} \\ }} 
\longleftrightarrow \quad
 {\tableau[scY]{\bullet&\bullet\\ \bullet\\ \bullet\\ }} \quad    {\tableau[scY]{\circ\\ \circ\\}}
\eeq
We observe that 
our convention 
of placing the boxes of the first partition into bosonic columns of $\La$
is preserved by the conjugation operation.

Recall that  the dominance ordering on bi-partitions
was defined in \eqref{domibi}.  In the following lemma we give three equivalent 
form of the second condition in the dominance ordering on bi-partitions. 
For this purpose, it will prove convenient to
relax the condition $|\lambda|=|\mu|$ in the dominance ordering
and say that
\beq \la\geq \mu \quad\text{if}\quad \la_1 + \cdots + \la_j \geq \mu_1 + \cdots + \mu_j \quad \forall j, \eeq
even in cases where $|\la|\ne |\mu|$.  Observe however that the equivalence
\beq
 \la \geq \mu \iff \la' \leq \mu' 
\eeq
only holds if $|\lambda|=|\mu|$.

\begin{lemma} \label{lemmaorder}
Suppose that $|\lambda| \geq |\omega|$.
Then the following three statements are equivalent:
\begin{enumerate}
\item $|\la|+ \mu \geq  |\omega |+ \eta $\\ 
\item  $\mu' \cup (1^{|\lambda|-|\omega|})  \leq \eta'$  \\
\item $\mu'\leq \eta' $
\end{enumerate}
\end{lemma}
\begin{proof}
We have that (1) and (2) are equivalent since
\beq 
 |\la|+ \mu_1 + \cdots + \mu_j \geq    |\omega |+ \eta_1 + \cdots + \eta_j 
\quad \forall j
\iff \mu + (|\lambda|-|\omega|) \geq \eta \iff 
\mu' \cup (1^{|\lambda|-|\omega|})  \leq
\eta' \, ,
\eeq
where we stress that $ \mu + (|\lambda|-|\omega|)$ stands for the partition
$(\mu_1+|\lambda|-|\omega|,\mu_2,\mu_3,\dots)$.   It is immediate that
(2) implies (3).  Finally, (2)
follows from (3) since $(1^{|\lambda|-|\omega|})$ is dominated by any partition (and in particular by the partition
$(\eta_1'+\dots+\eta'_{\ell+1}-|\mu|,\eta'_{\ell+2},\eta'_{\ell+3},\dots)$, where $\ell=\ell(\mu')$).

\end{proof}

The following proposition was essential to deduce Theorem~\ref{Pdef1}
from Theorem~\ref{theo1}, and thus to connect the Macdonald polynomials in superspace 
to the  double Macdonald polynomials. 
\begin{proposition} \label{lemdo} 
Suppose that
$\Lambda \leftrightarrow (\lambda,\mu)$ and $\Omega  \leftrightarrow (\omega,\eta)$, with $m\geq n=|\lambda|+|\mu|=|\omega|+|\eta|$.  Then
\beq
(\la,\mu) \geq (\omega,\eta) \iff \Lambda \geq \Omega
\eeq
\end{proposition}
\begin{proof} 

For the purpose of this proof, we  first modify our convention for the insertion of the boxes of  $\la$ and $\mu$ within $\La$ into a prescription that describes $\La^*$ built 
from the core $\delta^m$.  Reconsider example \eqref{ex1} but now  with $m=7$ and identify the $\circ$ and $\bullet$ as the upper  and lower boxes, respectively, that  lie outside the sub-diagram $\delta^m$. This yields:
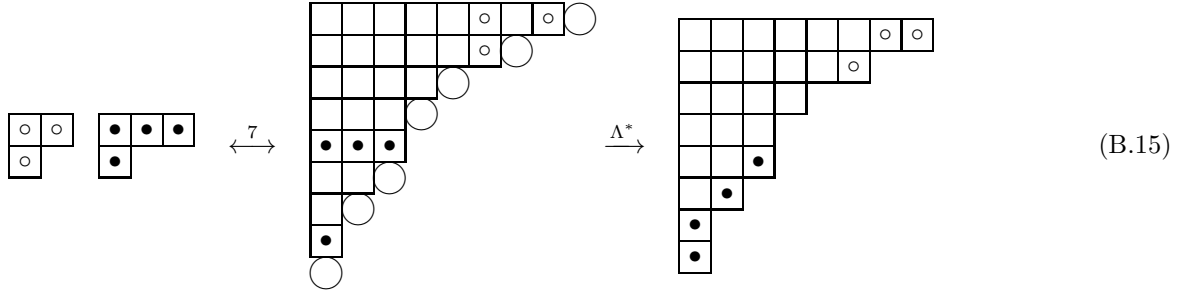
\beq\label{ex2bis}
 {\tableau[scY]{\circ& \circ \\ \circ}} \quad
 {\tableau[scY]{\bullet&\bullet& \bullet\\ \bullet\\ }} \quad
\stackrel{7}{\longleftrightarrow} \quad
  {\tableau[scY]{&&&&&\circ&&\circ&\bl\tcercle{}\\&&&&&\circ&\bl\tcercle{}
  \\&&&&\bl\tcercle{} \\&&&\bl\tcercle{} \\\bullet&\bullet& \bullet \\ &&\bl\tcercle{}\\&\bl\tcercle{} \\ \bullet\\ \bl\tcercle{} \\ }
}
\xrightarrow{\Lambda^*}\quad
  {\tableau[scY]{&&&&&&  \circ& \circ&\bl\\&&&&& \circ\\&&&\\&&  \\&& \bullet\\&\bullet \\  \bullet\\  \bullet }}
\eeq
where the unmarked boxes in $\La^*$ are those of  $\delta^7$. More generally,
$\Lambda^*$ can  be obtained by adding $\lambda$ (resp. $\mu'$)
to the top rows (resp. leftmost columns)  of $\delta^m$. In other words, 
\beq \La^*=  \left(\delta^m+\mu\right)'+ \lambda. \eeq 
It is clear from the above diagrams that the row of $\La$ that corresponds to the first row of $\mu$ (namely, the $(m+1-\mu_1)$-th row) is also the highest row in $\La^*$ containing a $\bullet$. Therefore
Lemmas \ref{lem2} and \ref{sepap} imply  that if $m \geq n$ and
$\Lambda \leftrightarrow (\lambda,\mu)$, then $\La^*$ can always be described as above, in particular, with the $\circ$ and $\bullet$ separated from each others.
In the context of this proof, 
$\Lambda^*$ will always stand for
 the diagrammatic representation \eqref{ex2bis}.
The following elementary observation will be fundamental. 
\begin{itemize}
\item[OBS 1:] Suppose that
the highest $\bullet$ in $\Lambda^*$ lies in row $r$ and column 
$m+1-r$ (the highest $\bullet$ is always alone in its row by construction).  
If $i<r$, then
the number of $\circ$'s strictly below row $i$ is not larger than the number
of rows ($r-i-1$) between rows $i$ and $r$.  Similarly, if $i<r$, then
the number of $\circ$'s strictly to the left of column $m+1-i$ is not larger than 
the number of columns ($r-i-1$) between columns $m+1-i$ and $m+1-r$.
\end{itemize}
This is seen as follows. The number of $\circ$ below the $i$-th row of $\La^*$ is $\la_{i+1}+\cdots+\la_\ell$, where $\ell=\ell(\la)$. The statement is that
\beq \la_{i+1}+\cdots+\la_\ell\leq r-i-1=m+1-\mu_1-i-1\eeq
or equivalently, that
\beq \mu_1+i+\la_{i+1}+\cdots+\la_\ell\leq m.\eeq
But since 
\beq \mu_1+i+\la_{i+1}+\cdots+\la_\ell\leq |\mu|+|\la| = n\leq m,\eeq
the result follows. The column case is treated in a similar way, the number of $\circ$ at the left of the column $m+1-i$ being $\leq 1+\la_{i+2}+\cdots+ \la_\ell$.

First, we show that given our hypotheses, the following implication holds:
\beq \label{implicatinunedirection}
(\la,\mu) \geq (\omega,\eta) \implies \Lambda \geq \Omega \, .
\eeq
The main step towards that goal is to establish that
\beq \label{implicationordre}
(\la,\mu) \geq (\omega,\eta) \implies \Lambda^* \geq \Omega^* \, .
\eeq
We proceed by contradiction.  Suppose that $(\lambda,\mu) \geq (\omega,\eta)$ and
$\Lambda^* \not \geq \Omega^*$, and let $i$ be the {\it highest} row such that
\beq \label{ineq1}
\Lambda_1^*+\cdots \Lambda_i^* < \Omega_1^*+\cdots+\Omega_i^* \, .
\eeq 
Note that given these conditions, we have necessarily $\Omega_i^*>\Lambda_i^*$.

{From} the first condition in \eqref{domibi} and the construction 
\eqref{ex2bis} of $\Lambda^*$, we can conclude that $i> \ell(\omega)$.  
Suppose that $i$ lies
above the highest $\bullet$ in $\Lambda^*$ (which we assume lies in a certain row $r$).
Let $d$ 
be the number of $\circ$ in $\Lambda^*$ 
strictly to the left of column $m+1-i$, so that $d\leqÊ\la_{i+1}+\cdots+\la_\ell\leq |\la|-|\om|$.
Since
$d \leq  |\lambda|-|\omega|$ we can use Lemma~\ref{lemmaorder}
to deduce 
\beq \label{contra1}
\mu_1'+\dots+\mu_{\mu_1}'+1+\cdots+1=|\mu'|+d \leq \eta_1'+\cdots+ \eta'_{\mu_1+d}\, . 
\eeq 
On the other hand,
we have $d \leq r-i-1=m-\mu_1-i$ by OBS 1, which implies that the
length ($\mu_1+d$) of the 
partition $\mu' \cup 1^d$ is smaller than $m+1-i$ (and thus the $\bullet$'s 
corresponding to $\eta_1',\dots,\eta_{\mu_1+d}'$ are all below row $i$).
{From} \eqref{ineq1} and $|\Lambda^*|=|\Omega^*|$, we also
get that 
below row $i$ there are more cells of 
$\Lambda^*$ than cells of $\Omega^*$.  Therefore 
\beq \label{contra2}
|\mu'|+d > \text{number of cells of $\eta$ below row $i$} \geq 
\eta_1'+\cdots+ \eta'_{\mu_1+d} \, ,
\eeq
which leads to a contradiction (compare \eqref{contra1} and \eqref{contra2}).

Now suppose that $i$ does not lie above the highest $\bullet$ in $\Lambda^*$,
and let $\ell=\Lambda^*_{i+1}$.  In this case, due to the hypothesis $(\la,\mu)\geq (\om,\eta)$ (which implies in particular the third expression in Lemma \ref{lemmaorder}), we get
\beq \label{contra3}
\mu_1'+\dots+\mu_{\ell}' \leq \eta_1'+\cdots+ \eta'_{\ell}. 
\eeq
According to the working hypothesis we want to contradict, namely $\Om_i^*>\La_i^*$, there must again be below row $i$ more cells of   
$\Lambda^*$ than of $\Omega^*$.  
Hence (letting $\ell-(i+1)=s$)
\beq \label{contra4}
\mu_1'+\cdots+\mu'_\ell- s(s-1)/2> \text{number of cells of $\eta$ below row $i$} \geq 
\eta_1'+\cdots+ \eta'_{\ell}-s(s-1)/2 \, ,
\eeq
where $s(s-1)/2$ is the number of cells of $\eta'$ and $\mu'$ above row $i+1$ up to column $\ell$ (note that the condition $\Omega_i^* > \Lambda_i^*$ ensures that this number is the same for $\mu'$ and $\eta'$).
Comparing \eqref{contra3} and \eqref{contra4} leads again to a contradiction,
and we can conclude that \eqref{implicationordre} holds.

Finally, observe that
\beq \label{observa}
 \Lambda^*=(\lambda+\delta^m) \cup \mu \quad  {\rm and} \quad
\Lambda^\circledast=(\lambda+\delta^m +  1^m) \cup \mu =
 (\lambda+\delta^{m+1}) \cup \mu \, ,
\eeq
that is, $\Lambda^{\circledast}$ corresponds to $\Lambda^*$ with $m$ replaced by 
$m+1$.  But \eqref{implicationordre} also holds for $m+1$ (it holds for all $m \geq n$), which means that
\beq 
(\la,\mu) \geq (\omega,\eta) \implies \Lambda^* \geq \Omega^* \quad
{\rm and} \quad  \Lambda^\circledast \geq \Omega^\circledast
\eeq
and  \eqref{implicatinunedirection} follows.

We now need to show the reverse implication:
\beq \label{implicationinverse}
(\la,\mu) \geq (\omega,\eta) \Longleftarrow \Lambda \geq \Omega \, .
\eeq
It is sufficient to show that given our hypotheses 
\beq \label{implicationinverseast}
(\la,\mu) \geq (\omega,\eta) \Longleftarrow \Lambda^\circledast \geq \Omega^\circledast \, .
\eeq
First we show that $\Lambda^\circledast \geq \Omega^\circledast$ implies
$\lambda \geq \omega$.
{From} \eqref{observa}, we get 
that the diagram of
$\Lambda^\circledast$ is that of $\Lambda^*$ with $m$ replaced by $m+1$. 
{From} this construction, we have immediately (using $\ell=\ell(\lambda)$)
\beq
\Lambda^\circledast \geq \Omega^\circledast \implies
\lambda_1+ \dots+\lambda_\ell \geq \omega_1+ \dots+ \omega_\ell
\eeq
Hence $\lambda_1+\cdots+\lambda_i \geq \omega_1+\cdots+ \omega_i$
will hold for all $i$ if $|\lambda| \geq |\omega|$.
Suppose that $|\omega|>|\lambda|$ and let $k=|\lambda|+1$.  An
observation similar to OBS 1 will prove useful.
\begin{itemize}
\item[OBS 2:] Suppose that
the highest $\bullet$ in $\Lambda^\circledast$ lies in row $r$.  Then
the number of $\circ$'s in $\Lambda^\circledast$ is not larger than $r-2$.
\end{itemize}
Indeed, the number of $\circ$  is equal to $|\la|$, and in this case
$r=m+2-\mu_1$.  Thus $|\la|\leq r-2=m-\mu_1$ since $|\la|+\mu_1\leq |\la|+|\mu|\leq m$.
It is then immediate that the highest $\bullet$ in $\Lambda^{\circledast}$
lies below row $k$ since OBS 2 gives
$|\lambda| \leq r-2$, or equivalently, $r\geq |\lambda|+2>k$ (that is, the cells of $\mu$ do not contribute to $\Lambda^\circledast$ up to row $k$).
Using   
\beq 
\omega_1 + \dots + \omega_{k} \geq k = |\lambda|+1 \, ,
\eeq
which follows from the hypothesis $|\omega|>|\lambda|$, we then get
\beq
\Omega^\circledast_1+\cdots+\Omega^\circledast_k -k(m+1)+k(k+1)/2 = \omega_1 + \dots + \omega_{k} \geq |\lambda|+1 > \Lambda^\circledast_1+\cdots+\Lambda^\circledast_k 
-k(m+1)+k(k+1)/2
\eeq
which is a contradiction to $\Lambda^\circledast \geq \Omega^\circledast$.
We thus have $\lambda  \geq \omega$.  By conjugation, we also have
$\eta \geq \mu$, 
and thus \eqref{implicationinverseast} holds by
 Lemma~\ref{lemmaorder}.
\end{proof}

\section{Proofs of the conjectures in the stable sector}\label{app_proof_conj}

In this section we recall some of our previous conjectures concerning the Macdonald superpolynomials and show that in the stable sector
they match statements that were demonstrated in this article. 
These conjectures are thus partly validated but we stress that they preserve their conjectural status in the non-stable sector.

Let us first recall that the Macdonald superpolynomials $P_\La(x,\theta)=P_\La(x,\theta;q,t)$ are defined as in Theorem \ref{theo1} by the two conditions
\begin{equation*}
P_\La(x,\theta;q,t) = m_\La(x,\theta) + \text{lower terms}, \qquad \text{and} \qquad  \LL P_\La(x,\ta), P_\Omega (x,\ta)\RR_{q,t} \propto \delta_{\La \Omega}
\end{equation*}
but with the scalar product for the bisymmetric power-sums replaced by its
superspace form:
\begin{equation}
 (-q)^{\binom{m}{2}} \LL p_{\la,\mu}, p_{\om,\eta} \RR_{q,t}  =\LL p_{\La}(x,\ta), p_{\Omega}(x,\ta) \RR_{q,t} 
\end{equation}
where $m$ is the fermionic degree of $\La$.  This matching factor (which comes from $q^{|\La^a|-|\la|}$) plays no role in the orthogonality conditions, and
affects only the value of the norm.

\subsection{Norm and integral form}
We first discuss two conjectures related to the norm and the integral version of the Macdonald superpolynomials. These are as follows.

\begin{conjecture}\cite{BDLM1}\label{con1}
The norm of the Macdonald superpolynomial $P_\La(x,\ta;q,t)$ is given by 
\begin{equation}\label{conj_norme}
 \LL P_\La(x,\ta), P_\La (x,\ta)\RR_{q,t} =(-1)^{\binom{m}{2}} q^{|\La^a|} \frac{h^\uparrow_\La(q,t)}{h^\downarrow_\La(q,t)},
\end{equation}
where
\begin{equation}\label{hdo}
h^\downarrow_\La(q,t) = \prod_{s \in \mathcal B \La} (1-q^{a_{\La^\circledast}(s)} t^{ l_{\La^*}(s)+1  }),
\qquad\text{and}\qquad h^\uparrow_{\La}(q,t)=h^\downarrow_{\La'}(t,q)  .
\eeq
\end{conjecture}
\n In the previous equation
$\mathcal B\La$ stands for the subset of boxes in the diagram  of $\La$ whose row and column  do not both end with a circle.  This is referred  to as the set of bosonic boxes of $\La$. For instance, in example \eqref{diag},  the bosonic boxes are $(1,3),\,(2,1),\,(2,2)$ and $(4,1)$.
\begin{conjecture}\cite{BDLM1}\label{con2}
The integral form of the Macdonald 
superpolynomials $P_\La(x,\ta;q,t)$ reads
\begin{equation} \label{intform}
J_\La(x,\ta;q,t) = {h^\downarrow_\La}(q,t)\, P_\La(x,\ta;q,t),
\end{equation} 
where ${h^\downarrow_\La}(q,t)$ is defined in \eqref{hdo}.  In other words,
the monomial expansion coefficients of $J_\La(x,\ta;q,t)$ belong to $\mathbb Z[q,t]$.
\end{conjecture}

In the following, we will make the connection between  the conjectured formula \eqref{conj_norme} of the norm of the Macdonald superpolynomials and the formula \eqref{normmacdobn} giving the norm of the double Macdonalds. In view of establishing this equivalence, we recall that the two norms differ by the factor $(-q)^{\binom{m}{2}}$.

\begin{lemma}\label{L2norms}
In the stable sector, the super and double norms are related by
the equation
\begin{equation}\label{2norms}
 q^{|\la|} \frac{h^\uparrow_\La(q,t)}{h^\downarrow_\La(q,t)}= 
  b_{\la,\mu}(q,t)^{-1} ,
\end{equation}
where $\La = (\la+\delta^m ; \mu)$ with $m\geq n= |\la|+|\mu|$.
\end{lemma}

We observe that the r.h.s. of \eqref{2norms} is $m$-independent (it is a formula that pertains to the stable sector) while this is not  obviously true for the l.h.s. So we first establish this fact.

\begin{lemma} Let $\La=(\la+\delta^m; \mu)$.  Then $h^\downarrow_{\La}(q,t)$, defined in \eqref{hdo}, does not depend upon $m$ when $m\geq n=|\la|+|\mu|$.
\end{lemma}
\begin{proof}
When $m\geq n$, the  boxes of $\La$  that belong to the subdiagram $\mathcal B \La$ are those corresponding to the inserted parts of $\la$ and $\mu$ in $\La$ (see Appendix \ref{proof_order}).  For example, with $\la,\mu=(3,1,1),(2,2,1)$ and $m=10$:
\begin{equation}
{\tableau[scY]{\circ& \circ& \circ\\ \circ \\ \circ}} \quad
 {\tableau[scY]{\bullet&\bullet \\ \bullet & \bullet \\ \bullet  }} \quad
\xrightarrow{\mathcal B \Lambda}\quad
  {\tableau[scY]{\times&\times&\times&\times&\times&\times&\times&\circ&\times&\times&\circ& \circ&\bl \tcercle{}\\\times&\times&\times&\times&\times&\times&\times&\circ&\times& \bl \tcercle{} \\\times&\times&\times&\times&\times&\times&\times&\circ& \bl\tcercle{} \\\times&\times&\times&\times&\times&\times&\bl\tcercle{}  \\\times&\times&\times&\times&\times& \bl\tcercle{}  \\\times&\times&\times&\times&\bl\tcercle{} \\ \times&\times&\times& \bl\tcercle{}  \\\times&\times&  \bl\tcercle{}  \\ \bullet&\bullet \\ \bullet&\bullet \\ \times & \bl\tcercle{} \\ \bullet \\ \bl\tcercle{} }}
\end{equation}
where the boxes marked with an $\times$ are not in $\B\La$ (they lie at the intersection of a row and a column ending with a circle).
Now, in the expression for  $h^\downarrow_{\La}(q,t)$, the product over the boxes of $\mathcal B \La$ can manifestly be separated into a product over the bosonic boxes of the following two smaller superpartitions:
\begin{equation}\label{subspart}
\bar \Omega :=(\la + \delta^{\ell(\la)} ; \, ), \qquad \text{and} \qquad \bar \La:=(\delta^{\mu_1} ; \mu).
\end{equation}
In the above example, these are
\beq
\bar \Omega :\quad
{\tableau[scY]{\circ&\times&\times&\circ& \circ&\bl \tcercle{}\\
\circ&\times& \bl \tcercle{} \\
\circ& \bl\tcercle{} \\ }}\qquad \qquad
\bar \La :\quad
{\tableau[scY]{\bullet&\bullet \\ \bullet&\bullet \\ \times & \bl\tcercle{} \\ \bullet \\ \bl\tcercle{} }}
\eeq
The decomposition takes the form
\begin{equation}\label{inm}
h^\downarrow_{\La}(q,t) =    \prod_{s \in \mathcal B \bar \Om } (1-q^{a_{\bar \Om^\circledast}(s)} t^{ l_{\bar \Om^*}(s)+1  }) \prod_{s' \in \mathcal B \bar \La} (1-q^{a_{\bar \La^\circledast}(s')} t^{ l_{\bar \La^*}(s')+1  }).
\end{equation}
This is a simple consequence of the fact that the leg-length and arm-length of 
the corresponding boxes 
do not depend on the rest of the diagram (see Corollary \ref{sepap}).  Now, the $m$-independence of both $\bar\Om$ and $\bar\La$ entails that of
 $h^\downarrow_{\La}(q,t)$.  An identical result holds for $h^\uparrow_{\La}(q,t)$.
\end{proof}

We now turn to the proof of Lemma \ref{L2norms}.
\begin{proof}
The objective is  to rewrite the  $\bar \Omega$ and $\bar \La$ contributions in \eqref{inm} solely in terms of the $\la$ and $\mu$ data, respectively.
The first step amounts to reorganize the terms in the first product as follows:
\begin{equation}
q^{  a_{\bar \Omega^\circledast}(s) }  t^{    l_{\bar \Omega^*}(s)+1      }  = q^{  a_{\bar \Omega^\circledast}(s)  -   l_{\bar \Omega^*}(s)-1}  (qt)^{    l_{\bar \Omega^*}(s)+1      } .
\end{equation}
Since all the parts of $\bar \Omega$ are fermionic,  $ a_{\bar \Omega^\circledast}(s) -1 = a_{\bar \Omega^*}(s)$.    
Note also that the boxes of $\mathcal B \bar \Omega$ correspond exactly to those of $\la$.
Then, it is easy to see that $ l_{\bar \Omega^*}(s) = l_\la(r)$ where $r$ 
is the box corresponding to $s$ in $\lambda$.
Similarly, we have $a_{\bar \Omega^*}(s) - l_{\bar \Omega^*}(s) = a_\la( r )$, the subtraction removing precisely the contribution to the arm of the non-bosonic boxes belonging to the staircase partition.  We thus have:
\begin{equation}
\prod_{s \in \mathcal B \bar \Omega} (1-q^{  a_{\bar \Omega^\circledast}(s) }  t^{    l_{\bar \Omega^*}(s)+1      }    )  = \prod_{r\in \la}(1-q^{a_\la(r)} (qt)^{l_\la(r)+1} ).
\end{equation}
Similarly, we reorganize the contribution of $\bar \La$ in \eqref{inm} as follows:
\begin{equation}
q^{  a_{\bar \La^\circledast}(s') }  t^{    l_{\bar \La^*}(s')+1      } = (qt)^{  a_{\bar \La^\circledast}(s') } t^{   l_{\bar \La^*}(s')+1 -  a_{\bar \La^\circledast}(s')  }.
\end{equation}
Since the rows of $\mu$,  when inserted into the staircase partition, constitute rows by themselves (non-fermionic ones), we obviously have $a_{\bar \La^\circledast}(s')  = a_\mu(r')$ with $r'$ the box of $\mu$ corresponding 
to $s'$.
Arguing as before, we also have $l_{\bar \La^*}(s')-  a_{\bar \La^\circledast}(s') = l_\mu(r')$, so that 
\begin{equation}
\prod_{s' \in \mathcal B \bar \La} (1-q^{  a_{\bar \La^\circledast}(s') }  t^{    l_{\bar \La^*}(s')+1      }    ) = \prod_{r' \in \mu}  (1-(qt)^{a_{\mu}(r')} t^{l_\mu(r')+1}).
\end{equation}
Collecting these results yields
\begin{equation}\label{hcc}
h^\downarrow_{\La}(q,t) = \prod_{r\in \la}(1-q^{a_\la(r)} (qt)^{l_\la(r)+1} ) \prod_{r' \in \mu}  (1-(qt)^{a_{\mu}(r')} t^{l_\mu(r')+1}) = c_\la(q,qt) c_\mu(qt,t).
\end{equation}
A similar expression holds for $h^\uparrow_{\La}(q,t)$:
\begin{equation}
h^\uparrow_{\La}(q,t)= h^\downarrow_{\La'}(t,q)=c_{\mu'}(t,qt) c_{\la'}(qt,q)
\end{equation}  
where we used $\La'=(\mu'+\delta^m; \la')$ (see Lemma~\ref{conju}).  The expressions for the norm can thus be related:
\begin{equation}
q^{|\la|} \frac{h^\uparrow_\La(q,t)}{h^\downarrow_\La(q,t)}=q^{|\la|} \frac{c_{\la'}(qt,q)}{c_{\la}(q,qt) } \, \frac{c_{\mu'}(t,qt)}{c_{\mu}(qt,t)} = q^{|\la|}b_{\la}(q,qt)^{-1} b_{\mu}(qt,t)^{-1} = b_{\la,\mu}(q,t)^{-1}.
\end{equation}
\end{proof}

Lemma \ref{L2norms} readily implies that:
\begin{corollary}
Conjecture \ref{con1} is true for $m\geq n$.
\end{corollary}

 On the other hand, the relation
 \eqref{hcc} implies the equivalence of   the integral forms \eqref{intform} and \eqref{Bnintform}:
\begin{corollary}
Conjecture \ref{con2} is true for $m\geq n$.
\end{corollary}

\subsection{Kostka coefficients}  

Let $\varphi$ be the  endomorphism 
of
$\mathbb Q(q,t)[p_1,p_2,p_3,\dots;\tilde p_0,\tilde p_1,\tilde p_2,\dots]$ 
whose action on the power-sums is
\begin{equation}
\varphi(p_n)=\frac{1}{(1-t^n)}p_n\qquad\text{and}\qquad 
\varphi(\tilde p_n)=\tilde p_n.
 \end{equation}
We then define the modified Macdonald polynomials in superspace as
\beq
H_{\La}(x,\theta;q,t) = \varphi \left( J_{\La}(x,\theta;q,t) \right).\eeq
We also introduce  the
Schur superpolynomial as the limit $q=t=0$ of the Macdonald 
superpolynomial
\beq
s_\Lambda(x,\theta)=P_\Lambda(x,\ta;0,0).\eeq
\begin{conjecture}\label{con3} \cite{BDLM1}
 The coefficients $ K_{\Omega \Lambda} (q,t)$ in the expansion of the
modified Macdonald superpolynomials
\begin{equation}
 H_\Lambda (x,\ta;q,t)  = \sum_\Omega K_{\Omega \Lambda} (q,t) \,s_\Omega(x,\ta)
 \end{equation}
are polynomials in $q$ and $t$ with nonnegative  integer coefficients. \end{conjecture}
When $m \geq n$, the Kostka coefficient $K_{\Omega \Lambda} (q,t)$ 
is equal to the coefficient $K_{\kappa,\gamma \, \lambda,\mu}(q,t)$
defined in \eqref{Kosbisym} (where
$\Lambda \leftrightarrow (\lambda,\mu)$ and $\Omega \leftrightarrow (\kappa,\gamma)$)
since the 
transformation of the previous equation into its bisymmetric form does not
affect its expansion coefficients. 
Therefore, as a direct consequence of Proposition \ref{Proppos}, we have:
\begin{corollary}
Conjecture \ref{con3} is true for $m\geq n$.
\end{corollary}

Finally, the following symmetries of the coefficients $K_{\Omega \Lambda} (q,t)$
have been observed \cite[Section 7.1]{BDLM2}: 
\begin{equation}\label{sym1}
K_{\Om \La}(q,t)=K_{\Om' \La'}(t,q),  \qquad \text{and}  \qquad
K_{\Om\La}(q,t)=q^{\bar{n}(\La')} t^{\bar{n}(\La)}\, K_{\Om' \La}(q^{-1},t^{-1}) ,
\end{equation}
where $\bar n(\La)$ is given by
\begin{equation}\label{exprcomplique}
{\bar{n}(\Lambda)}=n(\S{\La})-d^\mathcal{B}(\La)\qquad\text{with}\qquad n(\la)=\sum_i(i-1)\la_i \, .
 \end{equation}
In the previous equation, $\S \La$ is the skew diagram 
$\S \La=\Lambda^{\circledast}/\delta^{m+1}$ (we consider that $n(\lambda/\mu)=n(\lambda)-n(\mu)$).
The term
$d^\mathcal{B}(\La)$ is defined as follows: fill each square
 $s\in\B\La$ 
(we recall that $\B\La$ was defined after Conjecture \ref{con1})
with the number  
of boxes above $s$ that are not in $\mathcal B \La$.  
Then, $d^\mathcal{B}(\La)$ is the sum of these entries.   
For instance, considering example \eqref{diag}, and marking again by an $\times$ the boxes in $\La$ that are not in $\B\La$, we have
\beq
d^\mathcal{B}((3,1,0;2,1))=4
: \qquad {\tableau[scY]{\times&\times&0&\bl\tcercle{}\\1&1\\\times&\bl\tcercle{}\\ 2\\
    \bl\tcercle{}}}.\eeq

As previously mentioned, in the stable sector $K_{\Omega \Lambda}(q,t) =K_{\kappa, \gamma \; \la,\mu}(q,t)$ with $\La=(\la+\delta^m;\mu)$ and $\Omega=(\kappa+\delta^m;\gamma)$.  The symmetries \eqref{sym1} must thus, in the stable sector,
be a consequence of 
Corrolary \ref{Cor_symK}. This is immediate for the first one. In the following lemma, we show that this is also true for the second one
by establishing the equivalence of the expressions $\bar n(\La)$ and $\bar n(\la,\mu)$ (the latter defined in Corollary \ref{Cor_symK}) when $m\geq n$.
\begin{lemma}\label{L2valeursn}
For $\La = (\la+\delta^m ; \mu)$ and $m\geq n= |\la|+|\mu|$, we have
 \begin{equation}\label{2valeursn}
n(\la) + |\mu|+n(\mu) + n(\mu') = n(\mathcal S \La) - d(\La) 
\end{equation}
\end{lemma}

\begin{proof}    
First, consider the expression $n(\mathcal S \La)$.   Since $\mathcal S \La= \La^\circledast/\delta^{m+1}= ( (\la+ \delta^{m+1}) \cup \mu)/\delta^{m+1}$, the expression $n(\mathcal S \La)$  is the ``depth-weighted sum'' $n(\cdot)$ (referring to $i-1$ as the depth of a box in row $i$) of the boxes of the diagrams of $\la$ and $\mu$ when inserted in the staircase partition $(m,m-1,\ldots,0)$.  For the analysis of $n(\S\La)$, it is convenient to view the insertion of the diagrams of $\la$ and $\mu$ into that of $\delta^{m+1}$ to be as described in the proof of Proposition \ref{lemdo} in
Appendix~\ref{proof_order}.
Consider again the example of $\la,\,\mu=(3,1,1),\, (2,2,1)$ and $m=10$.   We have:
\begin{equation} \label{unskewLa}
{\tableau[scY]{\circ& \circ& \circ\\ \circ \\ \circ}} \quad
 {\tableau[scY]{\bullet&\bullet \\ \bullet & \bullet \\ \bullet  }} \quad
\xrightarrow{\mathcal S \Lambda}\quad
  {\tableau[scY]{\times&\times&\times&\times&\times&\times&\times& \times &\times&\times& \circ& \circ& \circ\\
 \times&\times&\times&\times&\times&\times&\times&\times&\times& \circ \\
 \times&\times&\times&\times&\times&\times&\times&\times& \circ \\
\times&\times&\times&\times&\times&\times&\times  \\
\times&\times&\times&\times&\times& \times  \\
\times&\times&\times&\times&\times \\
 \times&\times&\times& \times  \\
 \times&\times&  \times  \\
  \times&\times \\ 
  \times&\bullet \\ 
   \bullet & \bullet \\
    \bullet \\ 
    \bullet }}
\end{equation}
The boxes marked with an $\times$ are not considered in the evaluation of $n(\mathcal S \La)$.  Clearly, we have
\begin{equation}
n(\mathcal S \La) = n(\la) + G_m(\mu)
\end{equation}
where $G_m(\mu)$ is some contribution that depends on the diagram of $\mu$  and  on $m$.  For a general partition $\mu$ inserted into $\delta^{m+1}$ (from the bottom), we can compute $G_m(\mu)$ as follows.  
First observe that the boxes of $\mu$ are ordered as
\begin{large}
\begin{align}
& \quad \vdots \nonumber \\
{\tableau[scY]{  & & & \bl \; \hdots \\ & & \bl \; \hdots \\ & \bl \; \hdots \\ \bl \vdots  }} \qquad
\xrightarrow{G_m(\mu)}\qquad
&{\tableau[scY]{
\times&\times&\times&\times&\times & \times & \times & \times \\
\times&\times&\times&\times&\times & \times &\times \\
\times&\times&\times&\times&\times & \times  \\
\times&\times&\times&\times&\times  \\
 \times&\times&\times& \times \\
 \times&\times&  \times &  \bl \iddots \\
  \times&\times & {{\scriptsize \text{$g_3$} }}     \\ 
  \times& {{\scriptsize \text{$g_2$} }}  & \bl \iddots \\ 
  {{\scriptsize \text{$g_1$} }}& {{\scriptsize \text{$g_5$} }}    \\
 {{\scriptsize \text{$g_4$} }}   &  \bl \iddots  \\ 
  {{\scriptsize \text{$g_6$} }}   \\
   \bl \vdots 
     }}
\end{align}
\end{large}
The first box $g_1$ is located in the $(m+1)$-th row of the total diagram.  Consequently, it has depth $g_1=m$.  Similarly,  we have:
\begin{equation}
g_2=(m-1), \quad g_3=(m-2),\quad  g_4=(m+1), \quad g_5=(m+1-1),\quad  g_6=(m+2),\cdots
\end{equation}
One easily sees the general pattern:
\begin{align}
G_m(\mu) &= (m+1-1) +(m+1-2) + \ldots + (m+1-\mu_1) \nonumber \\
&\quad +(m+2-1) + (m+2-2) + \ldots + (m+2- \mu_2) \nonumber \\
&\quad + \ldots \nonumber \\
&\quad + (m+ \ell(\mu) -1) + (m+\ell(\mu) -2 ) + \ldots + (m+\ell(\mu) - \mu_{\ell(\mu)}) \nonumber \\
&= \sum_{i=1}^{\ell(\mu)} \sum_{j=1}^{\mu_i} (m +i-j).
\end{align}
This summation is readily evaluated:
\beq G_m(\mu)=\sum_{i}\left( (m+i)\mu_i -\mu_i(\mu_i+1)/2\right)
=\sum_i\left( (i-1)\mu_i+m\mu_i-\mu_i(\mu_i-1)/2\right)=n(\mu)+m|\mu|-n(\mu'), \eeq
where we used the relation $\sum_i \mu_i(\mu_i-1)/2=n(\mu')$.
We have thus
\begin{equation}\label{nsla}
n(\mathcal S \La) = n(\la) + n(\mu)+m|\mu|-n(\mu').
\end{equation} 

Let us turn to the contribution of $d^\mathcal{B}(\La)$. For its computation, we need to return to the standard insertion of $\la$ and $\mu$ within $\La$ illustrated in \eqref{ex1}. Now, $d^\mathcal{B}(\La)$ can naturally be split  in two parts
\begin{equation}
d^\mathcal{B}(\La) = d_1^\mathcal{B}(\la) + d_2^\mathcal{B}(\mu),
\end{equation}
simply because all the boxes of $\mathcal B \La$ are those corresponding to the inserted ones of $\la$ and $\mu$.  Clearly, no box of the subdiagram $\La/\mathcal B\La$ lies above the boxes corresponding to $\la$ in $\mathcal B\La$ since the latter are all inserted in the topmost rows.  Consequently, we have $d_1^\mathcal{B}(\la)=0$.   Now, consider the first row $\mu_1$ of $\mu$.
 This row is inserted in the $(m-\mu_1+1)$-th row of  
$\La$ (see \eqref{ex1} and \eqref{ex2} for instance).  
Each box of this row of length $\mu_1$ has thus exactly $m-\mu_1$ boxes of $\La/\mathcal B\La$ above it, for a contribution of $(m-\mu_1)(\mu_1)$ to $d_2^\mathcal{B}(\mu)$.   For the next row of $\mu$, two cases are possible.  If $\mu_2=\mu_1$, this second row of $\mu$ is just below the row $\mu_1$ in the diagram $\La$, in which case its contribution to $d_2^\mathcal{B}(\mu)$ is $(m-\mu_2) (\mu_2)$.  Otherwise $\mu_2 < \mu_1$, so that this row $\mu_2$ is inserted in the $(m-\mu_2+1 +1)$-th row of $ \La$.  Therefore, the number of boxes of $\La/\mathcal B \La$ that lies above each box of this row is $(m-\mu_2+2-1 -1)$, since we have to remove the box belonging to row $\mu_1$.  Its contribution is then also $(m-\mu_2)(\mu_2)$. 
  In general, the contribution 
 the part $\mu_k$ to $d_2^\mathcal{B}(\mu)$ is
\begin{align}
&[( \text{position of the row $\mu_k$ in the diagram $ \La$ }) - 1 - (\text{ $\#$ of rows of $\mu$ above $\mu_k$} )] \times (\mu_k)  \nonumber \\
& = ( (m-\mu_k+k)-1-(k-1))(\mu_k) = (m-\mu_k)\mu_k.
\end{align}
We thus have
\begin{equation}
d^\mathcal{B}(\La)=d_2^\mathcal{B}(\mu) = \sum_i (m-\mu_i)\mu_i =(m-1)|\mu|-2n(\mu').
\end{equation}
All together, this gives:
\beq n(\mathcal S\La)-d^\mathcal{B}(\La)= n(\la) + n(\mu) + |\mu| +n(\mu').\eeq 
\end{proof}

\subsection{Evaluation} In \cite{BDLM1}, we have formulated an intriguing conjecture for the evaluation of the Macdonald superpolynomials.    Let $F(x,\theta)$ be a polynomial in superspace with fermonic degree $m$ and suppose that $N \geq m$.   The evaluation of $F(x,\theta)$ is defined as
\begin{equation}
\mathsf{E}_{N,m}[ F(x,\theta) ]:=\left[ \frac{\partial_{\theta_m}  \cdots \partial_{\theta_1} F(x,\theta) }{\Delta(x_1, \ldots, x_m)} \right]_{x_1=u_1, \ldots, x_N=u_n},
\end{equation}
where the values $u_i$ are given by
\begin{equation}\label{defui}
u_i=\frac{t^{i-1}}{q^{\text{max}(m-i,0)  }},
\end{equation}
and where $\Delta(x_1, \ldots, x_m)=\prod_{1\leq i < j \leq m}(x_i-x_j)$ is the Vandermonde determinant in the variables $x_1, \ldots, x_m$. The conjectured formula involves the quantity $d^\mathcal{F}(\La)$ defined as follows.  Fill each square $s\in \mathcal{F}\La$ (these are the boxes whose column and row both end with a circle)  with the number of boxes above $s$ that are not in $\mathcal{F}\La$.  Then, $d^\mathcal{F}(\La)$ is the sum of these entries.  For instance, considering again the example \eqref{diag}, and marking by $\times$ the boxes in $\La$ that are not in $\mathcal{F}\La$, we have
\beq
d^\mathcal{F}((3,1,0;2,1))=1
: \qquad {\tableau[scY]{ 0& 0& \times &\bl\tcercle{}\\ \times & \times \\ 1&\bl\tcercle{}\\  \times \\
    \bl\tcercle{}}}.
\eeq
\begin{conjecture}\label{conjsurEval}
Let $\La$ be of fermionic degree $m$ and suppose that $N \geq \ell(\La^\circledast)$.  Then, the evaluation formula for the Macdonald superpolynomial  $P_\La(x,\ta;q,t)$ is
\begin{equation}\label{genEval}
\mathsf{E}_{N,m}[P_\La(x,\theta;q,t)] = \frac{t^{  n(\mathcal{S}\La) + d^\mathcal{F}(\La)  } }{q^{(m-1) |  \La^a/\delta^m   |  - n(\La^a/\delta^m)}  h^\downarrow_\La(q,t) } \prod_{s\in \mathcal S \La} \left( 1- q^{a_{\La^{\circledast}}'(s)} t^{  N- l_{\La^\circledast}'(s)    }\right)
\end{equation}
where ${h^\downarrow_\La}(q,t)$ is defined in \eqref{hdo}, while $a'(s)$ and $l'(s)$ are the arm colength and leg colength defined in \eqref{eqcoarms}.
\end{conjecture}
We now show that in the stable sector, this evaluation formula agrees with the one given in \eqref{stableEval}.  
\begin{lemma}  For $\La=(\delta^m+\la ; \mu)$ and $m\geq n = |\la|+|\mu|$, we have
\begin{equation}
\mathrm E_{N,m}(   P_{\la,\mu}(x,y;q,t)   ) = \mathsf{E}_{N,m}[P_\La(x,\theta;q,t)].
\end{equation}
\end{lemma}
\begin{proof}
The product over $\mathcal{S}\La$ in  \eqref{genEval} can be separated into two contributions: the boxes of the diagram of $\la$ and those of the diagram of $\mu$ (see for example \eqref{unskewLa}):
\begin{equation}
\prod_{s\in (\delta^{m+1}+\la)/\delta^{m+1}    } \left( 1- q^{a_{\La^{\circledast}}'(s)} t^{  N- l_{\La^\circledast}'(s)    }\right)  \;  \; \; \prod_{     s' \in (\delta^{m+1} + \mu)'/\delta^{m+1}      } \left( 1- q^{a_{\La^{\circledast}}'(s')} t^{  N- l_{\La^\circledast}'(s')    }\right)  .
\end{equation}
For the product corresponding to the diagram of $\la$, we have   $a_{\La^{\circledast}}'(s)=m-l'_\la(s) + a'_\la(s)$since a box in row $i$ of $\lambda$ has $(m-i+1)$ extra boxes to its left when considered within $S\La$ (and $l'_\la(s)=i-1$).  
Because there are no boxes above the rows of $\la$, we have $l'_{\La^\circledast}(s) = l'_\la(s)$.   
As for the product over the boxes of $\mu$, observe that the topmost box $s$ of each column of $\mu$ (into $\mathcal S \La$) has exactly $m-a'_\mu(s)$ boxes above it (i.e., $m-j+1$ if $s$ is in the $j$-th column and $a'(s)=j-1$).  Thus, for $s'\in (\delta^{m+1} + \mu)'/\delta^{m+1}$,   we have $l'_{\La^\circledast}(s') = l'_\mu(s')+m-a'_\mu(s')$ and since there are no boxes to the left of the cells corresponding to the diagram of $\mu$ in $\mathcal S \La$, we can write $a'_{\La^\circledast}(s') = a'_\mu(s)$.  
With $\La^a/\delta^m  = \la$ and using $h^\downarrow_\La (q,t)= c_\la(q,qt) c_\mu(qt,t) $  which holds true in the stable sector (cf. \eqref{hcc}), we have:
\begin{equation}\label{genEval2}
\mathsf{E}_{N,m}[P_\La(x,\theta;q,t)] = \frac{t^{  n(\mathcal{S}\La) + d^\mathcal{F}(\La)  } }{q^{(m-1) | \la  |  - n(\la)} c_\la(q,qt) c_\mu(qt,t)   }  \prod_{s\in \la}  (1-q^{a'_\la(s) +m -l'_\la(s)}  t^{N-l'_\la(s)})   
\prod_{s'\in \mu}  (1-q^{a'_\mu(s')}  t^{N+a'_\mu(s') -m-l'_\mu(s') } ).
\end{equation}
Now consider the quantity $d^\mathcal{F}(\La)$.  Clearly, it will only depend on the superpartition $\bar \La$ defined in \eqref{subspart}.   
In $\bar \La$, there are $\mu_1-1$ fermionic rows.   Every box in a fermionic row contributes equally to $d^\mathcal{F}(\La)$ and this contribution is the number of parts in $\mu$ greater than this fermionic row.   This entails the following expression: 
\begin{equation}
d^\mathcal{F}(\La)= \sum_{i=1}^{\mu_1}(i-1)[ \text{ multiplicity of parts $\geq i$ into $\mu$ } ] = \sum_{i=1}^{\mu_1}(i-1) [ \text{ Card}\{ j : \mu_j \geq i \} ].
\end{equation} 
But since $\text{Card}\{ j : \mu_j \geq i \}$ corresponds to $\mu'_i$ and $\mu_1= \ell(\mu')$, we have that $d^\mathcal{F}(\La) = n(\mu')$.  Using this result and the expression \eqref{nsla} for $n(\mathcal S \La)$, we thus get:
\begin{equation}
n(\mathcal S \La) + d^\mathcal{F}(\La) = n(\la) + m |\mu| + n(\mu).
\end{equation}
Substituting this result into \eqref{genEval2}, we finally obtain:
\begin{align}
\mathsf{E}_{N,m}[P_\La(x,\theta;q,t)] &= \frac{ t^{m|\mu|} (qt)^{n(\la)}  t^{n(\mu)} }{q^{(m-1)|\la|}}  
\prod_{s\in \la}  \frac{1-q^{a'_\la(s) +m -l'_\la(s)}  t^{N-l'_\la(s)}}{  1-q^{a_\la(s)}  (qt)^{l_\la(s)+1} }   
\prod_{s'\in \mu} \frac{ 1-q^{a'_\mu(s')}  t^{N+a'_\mu(s') -m-l'_\mu(s') } }{  1-(qt)^{a_\mu(s')}  t^{l_\mu(s')+1}  } \nonumber \\
&= \frac{ t^{m|\mu|} }{q^{(m-1)|\la|}}  
\prod_{s\in \la}  \frac{ (qt)^{l'_\la(s) }(1-q^{a'_\la(s) +m -l'_\la(s)}  t^{N-l'_\la(s)} )}{  1-q^{a_\la(s)}  (qt)^{l_\la(s)+1} }   
\prod_{s'\in \mu} \frac{ t^{l'_\mu(s')} (1-q^{a'_\mu(s')}  t^{N+a'_\mu(s') -m-l'_\mu(s') }) }{  1-(qt)^{a_\mu(s')}  t^{l_\mu(s')+1}  } \nonumber \\
& = \frac{t^{m|\mu|}  }{q^{(m-1) |\la|  } } \prod_{s\in \la } \frac{(qt)^{l'_\la(s)}  - q^{a_\la'(s)} (q^mt^N) }{ 1-q^{a_\la(s)}  (qt)^{l_\la(s)+1}  }
 \prod_{s'\in \mu } \frac{(t)^{l'_\mu(s')}  - (qt)^{a_\mu'(s')} (t^{N-m}) }{ 1-(qt)^{a_\mu(s')}  t^{l_\mu(s')+1}  } \nonumber \\
&= \frac{t^{m|\mu|}  }{q^{(m-1) |\la|  } } w_\la(q^mt^N;q,qt) w_\mu(t^{N-m};qt,t)
\end{align}
which is precisely the formula for $\mathrm E_{N,m}(P_{\la,\mu}(x,y;q,t))$.
\end{proof}
\begin{corollary}
Conjecture \ref{conjsurEval} is true for $m\geq n$.
\end{corollary}

\section{Tables of Kostka coefficients} \label{app_table}

{
\tiny{
\begin{table}[ht]
\caption{$K_{\la, \mu \; \omega, \nu }(q,t)'$ for   degree $n=1$.} 
\label{tab22}
\begin{center}
\begin{tabular}{l | c |c | } 
 & $(1),\em $ & $\em,(1)$ \\ \hline
$(1),\em$ & $1$ & $q$ \\ \hline
$\em,(1)$ & $t$ & $1$ \\ \hline
\end{tabular}
\end{center}
\end{table}
\begin{table}[ht]
\caption{$K_{\la, \mu \; \omega, \nu }(q,t)'$ for   degree $n=2$.} 
\label{tab32}
\begin{center}
\begin{tabular}{l | c |c | c | c |c| } 
 & $(2),\em $ & $(1,1),\em$ & $(1),(1 )$ & $\em,(2 )$ & $\em,(1,1 )$ \\ \hline
$(2),\em $ & $1$ & $q$ & $q+q^2 $ & $q^2 $  & $q^3$ \\ \hline
$(1,1),\em$ & $qt$ & $ 1$ & $q+q^2t$ & $q^3t$ & $q^2$\\ \hline
$(1),(1 )$ & $t$ & $ t$ & $1+qt$ & $q$ & $q$\\ \hline
$\em,(2 )$ & $t^2$ & $ qt^3$ & $t+qt^2$ & $1$ & $qt$\\ \hline
$\em,(1,1 )$ & $t^3$ & $ t^2$ & $t+t^2$ & $t$ & $1$\\ \hline
\end{tabular}
\end{center}
\end{table}
\begin{table}[ht]
\caption{$K_{\la, \mu \; \omega, \nu }(q,t)'$ for   degree $n=3$.} 
\label{tab33}
\begin{center}
\begin{tabular}{l | c |c | c | c |c| c|  c| c| c| c| } 
 & $(3),\em$ & $(2,1),\em $& $(2),(1)$ & $(1^3),\em$ & $(1,1),(1 )$ & $(1),(2)$&$\em,(3 )$ & $(1),(1,1)$&$\em,(2,1 )$ &$\em,(1^3 )$ \\ \hline
$(3),\em$ &  $1$ & $q+q^2$ & $q+q^2+q^3$ & $q^3$  & $q^2+q^3+q^4$  & $q^2+q^3+q^4$  & $q^3$  & $q^3+q^4+q^5$ & $q^4+q^5$ & $q^6$  \\ \hline
$(2,1),\em $&   $qt$ & $1+q^2t$ & $q+q^2t+q^3t$ & $q$ & $q+q^2+q^3t$ & $q^2+q^3t+q^4t$ & $q^4t$ & $q^2+q^3+q^4t$  & $q^3+q^5t$ & $q^4$  \\ \hline
$(2),(1)$ &  $t$ & $t+qt$ & $1+qt+q^2t$ & $qt$ & $q+qt+q^2t$ & $q+q^2+q^2t$ & $q^2$ & $q+q^2+q^3t$ & $q^2+q^3$ & $q^3$  \\ \hline
$(1^3),\em$ & $q^3t^3$ & $qt+q^2t^2$ & $q^2t+q^3t^2+q^4t^3$ & $1$ & $q+q^2t+q^3t^2$ & $q^3t+q^4t^2+q^5t^3$ & $q^6t^3$ & $q^2+q^3t+q^4t^2$ & $q^4t+q^5t^2$ & $q^3$  \\ \hline
$(1,1),(1 )$ &   $qt^2$ & $t+qt^2$ & $2qt+q^2t^2$ & $t$ & $1+qt+q^2t^2$ & $q+q^2t+q^3t^2$ & $q^3t$ & $q+2q^2t$ & $q^2+q^3t$ & $q^2$  \\ \hline
$(1),(2)$&   $t^2$ & $t^2+qt^3$ & $t+2qt^2$ & $qt^3$ & $t+qt^2+q^2t^3$ & $1+qt+q^2t^2$ & $q$ & $2qt+q^2t^2 $ & $q+q^2t$ & $q^2t$ \\ \hline
$\em,(3 )$ &     $t^3$ & $qt^4+q^2t^5$ & $t^2+qt^3+q^2t^4$ & $q^3t^6$ & $qt^3+q^2t^4+q^3t^5$ & $t+qt^2+q^2t^3$ & $1$ & $qt^2+q^2t^3+q^3t^4$ & $qt+q^2t^2$ & $q^3t^3$     \\ \hline
$(1),(1,1)$  &   $ t^3$ & $t^2+t^3$ & $t+t^2+qt^3$ & $t^2$ & $t+t^2+qt^2$ & $t+qt+qt^2$ & $qt$ & $1+qt+qt^2$ & $q+qt$ & $q$        \\ \hline
$\em,(2,1 )$ &    $ t^4$ & $t^3+qt^5$ & $t^2+t^3+qt^4$ & $qt^4$ & $t^2+qt^3+qt^4$ & $t+t^2+qt^3$ & $t$ & $t+qt^2+qt^3$ & $1+qt^2$ & $qt$     \\ \hline
$\em,(1^3 )$  &     $ t^6$ & $t^4+t^5$ & $t^3+t^4+t^5$ & $t^3$ & $t^2+t^3+t^4$ & $t^2+t^3+t^4$ & $t^3$ & $t+t^2+t^3$ & $t+t^2$ & $1$    \\ \hline
\end{tabular}
\end{center}
\end{table}
}
}

\end{appendix}

\end{document}